\documentclass[acmsmall,table,screen,nonacm]{acmart}
\pdfoutput=1 

\setcopyright{rightsretained}
\setcctype{by}
\acmJournal{PACMPL}
\acmYear{2025} 
\acmVolume{9} 
\acmNumber{ICFP} 
\acmArticle{}
\acmMonth{8} 
\acmPrice{}
\acmDOI{10.1145/3747510} 

\usepackage[T1]{fontenc} 
\usepackage{stmaryrd}
\usepackage{algorithm}
\usepackage{algorithmic}
\usepackage{enumitem}
\usepackage{xcolor}  

\usepackage{scalerel}
\DeclareMathOperator*{\bigbullet}{\scalerel*{\bullet}{\sum}}

\usepackage{skull,ifoddpage,marginnote}

\newcounter{ToDos}
\newcounter{WarnCounts}

\newcommand{\cons}{\join}
\newcommand{\cat}{\join}
\newcommand{\sngl}[1]{#1}
\newcommand{\nil}{[]}
\newcommand{\nnull}{\textrm{null}}
\newcommand{\lseg}{\mathit{list}}
\newcommand{\gseg}{\mathit{graph}}
\newcommand{\emp}{\mathit{emp}}
\newcommand{\emptymap}{e}
\newcommand{\partgraph}{\textit{partial{-}graph}}
\newcommand{\bingraph}{\textit{binary{-}graph}}
\newcommand{\erasure}[1]{|{#1}|}

\newcommand{\goodg}{\textit{closed}}

\DeclareUnicodeCharacter{2212}{-}   

\newcommand{\notM}{\scalebox{1.2}{$\scriptstyle{\textsf{\textit{O}}}$}}
\newcommand{\rightM}{\scalebox{1.2}{$\scriptstyle{\textsf{\textit{R}}}$}}
\newcommand{\leftM}{\scalebox{1.2}{$\scriptstyle{\textsf{\textit{L}}}$}}
\newcommand{\killed}{\scalebox{1.2}{$\scriptstyle{\textsf{\textit{X}}}$}}

\newcommand{\assign}{\,{:=}\,}
\newcommand{\deref}[1]{#1}
\newcommand{\mutate}{\,{:=}\,}

\newcommand{\ldot}{\mathord{.}\,}
\newcommand{\eqdef}{\mathrel{\:\widehat{=\hspace{0.2mm}}\:}}

\newcommand{\specK}[1]{\ensuremath{\textcolor{blue}{#1}}}

\newcommand{\spec}[1]{\specK{\left\{{#1}\right\}}}

\newcommand{\sspecopen}[1]{\specK{\{{#1}\hphantom{\}}}}
\newcommand{\opensspec}[1]{\specK{\hphantom{\{}{#1}\}}}
\newcommand{\opensspecopen}[1]{\specK{\hphantom{\{}{#1}\hphantom{\}}}}
\newcommand{\dotcup}{\ensuremath{\mathaccent\cdot\cup}}

\newcommand{\pcmF}{\bullet}
\newcommand{\join}{\pcmF}

\makeatletter
\newcommand\ostep[2][]{\ext@arrow 0099{\longrightarrowfill@}{#1}{#2}}
\def\longrightarrowfill@{\arrowfill@{\ \ \ }\relbar\longrightarrow}
\makeatother

\makeatletter
\newcommand\osteps[2][]{\ext@arrow 0099{\longrightarrowfillstar@}{#1}{#2}}
\def\longrightarrowfillstar@{\arrowfill@{\ \ \ }\relbar{\longrightarrow^*}}
\makeatother

\makeatletter
\newcommand\mstep[2][]{\ext@arrow 0099{\longrightarrowfill@}{#1}{#2}}
\def\longrightarrowfill@{\arrowfill@{\ \ \ }\relbar\longrightarrow}
\makeatother

\makeatletter
\newcommand\msteps[2][]{\ext@arrow 0099{\longrightarrowfillstar@}{#1}{#2}}
\def\longrightarrowfillstar@{\arrowfill@{\ \ \ }\relbar{\longrightarrow^*}}
\makeatother

\begin{document}

\title{Verifying Graph Algorithms in Separation Logic: A Case for an Algebraic Approach (Extended Version)}

\author{Marcos Grandury}
\affiliation{%
  \institution{IMDEA Software Institute and Universidad Polit\'{e}cnica de Madrid}
  \city{Madrid}
  \country{Spain}}
\email{marcos.grandury@imdea.org}

\author{Aleksandar Nanevski}
\affiliation{%
  \institution{IMDEA Software Institute}
  \city{Madrid}
  \country{Spain}}
\email{aleks.nanevski@imdea.org}

\author{Alexander Gryzlov}
\affiliation{%
  \institution{IMDEA Software Institute}
  \city{Madrid}
  \country{Spain}}
\email{aliaksandr.hryzlou@imdea.org}

\begin{abstract}
Verifying graph algorithms has long been considered challenging in
separation logic, mainly due to structural sharing between graph
subcomponents. We show that these challenges can be effectively
addressed by representing graphs as a partial commutative monoid
(PCM), and by leveraging structure-preserving functions (PCM
morphisms), including higher-order combinators.

PCM morphisms are important because they generalize separation logic's
principle of local reasoning. While traditional framing isolates
relevant portions of the heap only at the top level of a
specification, morphisms enable contextual localization: they
distribute over monoid operations to isolate relevant subgraphs, even
when nested deeply within a specification.

We demonstrate the morphisms' effectiveness with novel and concise
verifications of two canonical graph benchmarks: the Schorr-Waite graph
marking algorithm and the union-find data structure.

\end{abstract}

\ccsdesc[500]{Theory of computation~Separation logic}
\ccsdesc[500]{Theory of computation~Program verification}

\maketitle

\setcounter{footnote}{0}
\newcommand{\G}{\gamma}
\newcommand{\Ga}{\gamma_{\textrm{adj}}}
\newcommand{\Gm}{\gamma_{\textrm{val}}}

\section{Introduction}\label{sec:intro}
The defining property of separation
logic~\cite{ohe+rey+yan:csl01,ish+ohe:popl01,rey:lics02} is that a
program's specification tightly circumscribes the heap that the
program accesses. Then \emph{framing}, with the associated frame rule
of inference, allows extending the specification with
a \emph{disjoint} set of pointers, deducing that the program doesn't
modify the extension. This makes the
verification \emph{local}~\cite{cal+ohe+yan:lics07}, in the sense that
it can focus on the relevant heap by framing out the unnecessary
pointers.

Another characteristic property of separation logic is that data
structure's layout in the heap is typically defined in relation to the
structure's contents, so that clients can reason about the contents
and abstract away from the heap.
A common example is the predicate $\lseg\ \alpha\ (i,
j)$~\cite{rey:lics02} which holds of a heap iff that heap contains a
singly-linked list between pointers $i$ and $j$, and the list's
contents corresponds to the inductive mathematical sequence $\alpha$.
This way, \emph{spatial} (i.e., about state; involving heaps and
pointers) reasoning gives rise to \emph{non-spatial} (i.e.,
purely mathematical, state-free, about contents) reasoning for the clients.
Fig.~\ref{fig:list-apart} illustrates a heap satisfying
$\lseg\ \alpha\ (i, j)$ when $\alpha = [a, b, c, d]$.

\begin{figure}[t]
  \centering
  \includegraphics[width=0.5\textwidth]{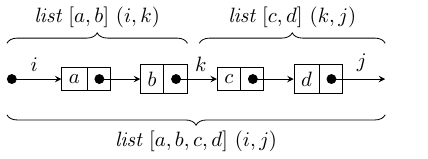}
  \vspace{-2mm}
  \caption{Spatial predicate $\lseg\ \alpha\ (i, j)$ describes the
    layout of a linked list, whose contents is sequence $\alpha$,
    between pointers $i$ and $j$. Dividing the heap is matched by
    dividing $\alpha$.}\label{fig:list-apart}
\end{figure}

The interaction of the above two properties imparts an important
requirement on non-spatial reasoning. Since framing involves
decomposition of heaps into disjoint subheaps, locality mandates that
the mathematical structure representing contents also needs to support
some form of disjoint decomposition. For example, in
Fig.~\ref{fig:list-apart}, the heap divides into two contiguous but
disjoint parts, between pointers $i$ and $k$, and $k$ and $j$,
respectively. Correspondingly, at the non-spatial level, the sequence
$[a, b, c, d]$ laid out between the pointers $i$ and $j$, decomposes
into subsequences $[a, b]$ and $[c, d]$, such that $\lseg\ [a, b]\
(i, k)$ and $\lseg\ [c, d]\ (k, j)$ hold of disjoint subheaps.

While mathematical sequences (along with sets, mathematical trees, and
many other structures) admit useful forms of decomposition, graphs
generally don't. Indeed, traversing a list or a tree reduces to
traversing disjoint sublists or disjoint subtrees from different
children of a common root; thus lists and trees naturally divide into
disjoint subcomponents. But in a graph, the subgraphs reachable from
the different children of a common node need not be disjoint,
preventing an easy division.
Even if a graph is divided into subgraphs with disjoint nodes, one
must still somehow keep track of the edges crossing the divide, in
order to eventually reattach the pieces. Keeping track of the
crossing edges is simple in the case of lists, as there's at most one
such edge (e.g., $k$ in Fig.~\ref{fig:list-apart}),
but graphs generally don't exhibit such regularity.

The lack of natural decomposition is why graphs have traditionally
posed a challenge for separation logic. For example, the Schorr-Waite
graph marking algorithm~\cite{schorr1967} is a well-known graph
benchmark,\footnote{The Schorr-Waite algorithm, described in
Section~\ref{sec:schorr}, marks the nodes reachable from a given root
node in a directed graph. Unlike recursive depth-first implementation, it
uses an iterative approach that simulates the call stack through graph
mutations, temporarily modifying the graph's pointer structure during
traversal and restoring it upon completion.}
that has been verified in separation logic early on
by~\citet{Yang2001AnEO,yang2001}, and in many other logics and tools
such as Isabelle~\cite{meh+nip:cade03} and Dafny~\cite{leino:lpar10}.
While these proofs share conceptually similar mathematical notions
(e.g., reachability in a graph by paths of unmarked nodes), they are
formally encoded differently, with Yang's proof exhibiting by far the
largest formalization overhead.%
\footnote{Yang's proof 
is described in some 40 pages of dense mathematical text. The proof in
Isabelle is mechanized in about 400 lines, and the proof in Dafny is
automated, aside from some 70 lines of annotations. Of course, Yang's
proof is the only one to use framing, which allowed restricting the
focus to connected graphs without sacrificing generality.}

Notably, Yang’s proof avoids using graphs in non-spatial reasoning by
refraining from explicit mathematical graph parameters in spatial
predicates. Instead, these predicates are parametrized by decomposable
abstractions---such as node sequences encoding traversal paths,
marked/unmarked node sets, the graph's spanning tree, etc.---that
indirectly approximate the graph’s structure. However, these proxies
are interdependent: maintaining their mutual consistency and
synchronization with the heap-allocated graph forces the proof to
perpetually alternate between spatial and non-spatial reasoning, an
entanglement that accounts for most of the proof’s
complexity. Furthermore, the absence of an explicit mathematical graph
parameter precludes invoking formal graph-theoretic lemmas, which in
turn limits proof reuse and scalability.

To reconcile spatial and non-spatial reasoning, we first
require a mathematical representation of graphs that inherently
supports decomposition. This motivates our generalization to 
\emph{partial graphs}, which are directed graphs that admit \emph{dangling edges}, 
i.e., edges whose source node belongs to the graph, but whose sink
node doesn't. Partial graphs decompose into subgraphs with disjoint
nodes, where an edge crossing the divide becomes dangling in the
subgraph containing the edge's source.
We formalize this decomposition by casting partial graphs
as \emph{partial commutative monoid (PCM)}. A PCM is a structure
$(U, \join, e)$ where $\join$, pronounced ``join'', is a partial
commutative and associative binary operation over $U$, with unit $e$,
which captures how graphs combine. If $\G = \G_1 \join \G_2$, we say
that $\G$ \emph{decomposes} or \emph{splits} into $\G_1$ and $\G_2$,
or that $\G_1$ and $\G_2$ \emph{join} into $\G$.

PCMs are central to separation logic, e.g., for defining
heaps~\cite{cal+ohe+yan:lics07},
permissions~\cite{bor+cal+ohe+par:popl05},
histories~\cite{ser+nan+ban:esop15}, and meta
theory~\cite{din+bir+gar+par+yan:popl13}. Modern separation logics
further support abstract states described by arbitrary
PCMs~\cite{jun+swa+sie+sve+tur+bir+dre:popl15,nan+ban+del+fab:oopsla19,kri+sum+wie:esop20}.
However PCMs \emph{alone} are insufficient for graph verification:
while they enable spatial reasoning over diverse notions of state,
effective graph verification requires a mechanism to
localize \emph{non-spatial} graph proofs.

We address this gap with the three novel contributions of this paper:
(1) leveraging PCM morphisms in graph reasoning, (2) utilizing
higher-order morphisms (i.e., combinators), leading to (3) new proofs
of Schorr-Waite and of another graph benchmark---the union-find data
structure~\cite{neelk:phd,wan+qin+moh+hob:oopsla19,cha+pot:jar17}. 
The latter is sketched in Section~\ref{sec:unionfind} to illustrate the
generality of the approach, but is detailed in the appendix.

\paragraph{\textbf{Leveraging morphisms in graph reasoning.}}
Given PCMs $A$ and $B$, the function $f : A \rightarrow B$ is a PCM
morphism if the following equations hold.%
\footnote{We only consider
morphisms that distribute unconditionally. Partial morphisms, which
distribute only if a certain condition on $\G_1$ and $\G_2$ is
satisfied, are useful and have been developed to an extent
by~\citet{far+nan+ban+del+fab:popl21}.}%
\begin{align}
f\ e_A\, &=\, e_B\\
f\ (\G_1 \join_A \G_2)\, &=\, f\ \G_1 \join_B f\ \G_2\label{eq:morphdistrib}
\end{align}
Focusing on equation (\ref{eq:morphdistrib}), it says that $f$
distributes over the join operation of the domain PCM and determines
the following two ways in which morphisms apply to verification.

First, morphisms mediate between PCMs, as the join operation of $A$ is
mapped to the join operation of $B$. We will use morphisms between
various PCMs in this paper, but an important example is to morph the
PCM of graphs to the PCM of propositions of separation logic
(equivalently, to the PCM of sets of heaps). Connecting graphs to sets
of heaps allows a suitable heap modification to be abstracted as a
modification of the graph, thereby elevating low-level separation
logic reasoning about pointers into higher-level reasoning about
graphs.
For example, if a node in a graph is represented as a heap storing the
node's adjacency list, mutating the pointers in this heap corresponds
at the graph level to modifying the node's edges.

Second, continuing with morphisms over graphs specifically, the value
of $f$ over a graph decomposed into components ($\G_1$ and $\G_2$ in
equation (\ref{eq:morphdistrib})) can be computed by applying $f$ to
$\G_1$ and $\G_2$ \emph{independently} and joining the results. During
verification, the left-to-right direction of (\ref{eq:morphdistrib})
localizes reasoning to the modified subgraph (say, $\G_1$), while the
right-to-left direction automatically propagates and reattaches
$\G_2$.
This allows morphisms to serve an analogous localizing role to framing
in separation logic, but with two distinctions. The obvious one is
that framing decomposes heaps, whereas morphisms decompose
mathematical graphs (spatial vs.~non-spatial reasoning). More
importantly, framing operates only at the top level of specifications,
whereas morphisms localize deeply inside a context:
equation~(\ref{eq:morphdistrib}) can rewrite within arbitrary context
$I(-)$ to transform $I(f (\G_1 \join \G_2))$ into $I(f \G_1 \join
f \G_2)$ and allow $f \G_1$ and $f \G_2$ to be manipulated
independently inside $I(-)$. 
The latter is the essential feature that 
differentiates non-spatial from spatial reasoning (where top-level framing suffices), and 
we shall use it to enable and streamline graph proofs.

As an illustration of morphisms, consider \emph{filtering} a graph to
obtain a subgraph containing only the nodes with a specific property.
Filtering differs from taking a subgraph in standard graph theory, in
that the filtered subgraph retains---as dangling---the edges into the
part avoided by the filter. This enables reattaching the two parts
later on, after either has been processed.  Clearly, to filter a
composite graph it suffices to filter the components and join the
results, making filtering a PCM morphism on partial graphs.

While not all useful graph abstractions are morphisms, those that
aren't can still usefully interact with morphisms. The canonical
example of such \emph{global} notions is \emph{reachability}
between nodes in a graph~\cite{kri+sum+wie:esop20}.
In verification, reachability often needs to be restricted to paths
traversing only specific (e.g., unmarked) nodes. While one could
implement this restriction by parametrizing reachability with
admissible node sets (e.g., as done by~\citet{leino:lpar10}), a simpler
solution in the presence of morphisms is to evaluate the global
reachability predicate on dynamically filtered subgraphs---effectively
composing reachability with filtering.  This compositional approach
scales naturally, enabling us to express complex specifications (like
Schorr-Waite's) by combining elementary graph operations from a
minimal core vocabulary.

\paragraph{\textbf{Higher-order combinators.}}
In the presence of morphisms, the specifications can utilize still
higher levels of abstraction. Our second contribution notes that
useful graph transformations are encodable as instances of
a \emph{higher-order combinator} (like map over sequences in
functional programming), itself a PCM morphism on partial graphs.
Morphisms thus relate different levels of the abstraction stack that
has heaps at the bottom, and localize the reasoning across all levels.

\newcolumntype{X}{>{\ $}c<{$\ }}

\paragraph{\textbf{New proofs of Schorr-Waite and union-find.}}
The main challenge in these proofs is establishing the invariance of
complex global properties, which we address precisely using contextual
localization, morphisms and combinators. Based on these abstractions,
we also develop a generic theory of partial graphs, which reduces
example-specific reasoning and yields conceptually simple and concise
proofs.  The entire development (graph theory, Schorr-Waite,
union-find) is mechanized in Hoare Type
Theory~\cite{nan+mor+bir:icfp06,nan+mor+shi+gov+bir:icfp08,nan+vaf+ber:popl10,htt:github},
a Coq library for separation logic reasoning via types, with the component
sizes summarized below~\cite{gra+nan+gry:icfp-artefact}.
\[
\begin{tabular}{X|X|X}
 & \textrm{lines of proof} & \begin{array}{c}\textrm{lines of specification}\\ \textrm{(definitions, annotations, notation, code)} \end{array}\\ \hline
\textrm{Graph-theory library} & 2070 & 1738 \\
\textrm{Schorr-Waite} & 110 & 111 \\
\textrm{Union-find} & 49 & 46 
\end{tabular}
\]

Finally, we note that working with morphisms imparts a distinct
algebraic 
character to specifications and proofs in this paper. While
our reasoning spans both spatial and non-spatial domains, the bulk of
the effort is on the non-spatial side, where decomposition is governed
by the PCM operation $\join$ rather than the separation logic's
characteristic spatial connective $*$. This shift manifests in
specifications: assertions use $*$ sparingly as decomposition occurs
primarily via $\join$ at the abstract graph level. It's this
prioritization of the PCM structure and morphisms over traditional
spatial decomposition that justifies the term ''algebraic'' in the
paper’s title.

We adopt the classical separation logic formulation
from~\citet{ohe+rey+yan:csl01}, Yang's
dissertation~\shortcite{yang2001} and~\citet{rey:lics02}, whose
inference rules (e.g., frame rule, consequence, etc.) are standard and
thus elided here. However, to enable the interleaving of heap-level
assertions with graph-theoretic transformations, we depart somewhat
from that work by allowing the assertions to embed unrestricted
mathematical formulas, including direct applications of morphisms.

This is an extended version of an ICFP 2025 paper~\cite{gra+nan+gry:icfp25}, which it extends with appendices.

\section{Background}\label{sec:overview}
To pinpoint the patterns of separation logic that inform our approach
to graphs, we consider the following program for computing the length
of a singly-linked list headed at the pointer $i$.
\[
  L \eqdef n \assign 0; j \assign i; \textrm{\textbf{while}}\ j \neq \nnull\
  \textrm{\textbf{do}}\  
       j \assign \deref{j.\textrm{next}}; 
       n \assign n+1
       \ \textrm{\textbf{end\ while}}
\]
To specify $L$, one must first describe how a singly-linked list is
laid out in the heap. For that, the proposition $\lseg\ \alpha\ (i,j)$
from Section~\ref{sec:intro} is defined below to hold of a heap that
stores a linked list segment between pointers $i$ and $j$, whose
contents is the mathematical sequence
$\alpha$. Following~\citet{rey:lics02}, we overload $\cat$ to denote
attaching an element to a head or a tail of a sequence, and
concatenating two sequences.
\begin{align*}
  \lseg\ \nil\ (i,j)\ &\eqdef\ \emp \wedge i = j\\
  \lseg\ (a\cons\alpha)\ (i,j)\ &\eqdef\ \exists k\ldot i \Mapsto a,k * \lseg\ \alpha\ (k,j)
\end{align*}
The definition is inductive in $\alpha$ and says: (1) The empty
sequence $\nil$ is stored in a heap between pointers $i$ and $j$ iff
the heap is empty and $i = j$; (2) The sequence $a\cons \alpha$ is
stored in the heap between pointers $i$ and $j$ iff $i$ points to a
list node storing $a$ in the value field, and some pointer $k$ in the
next field, so that $\alpha$ is then stored between $k$ and $j$ in a
heap segment \emph{disjoint} from $i$.

Here $\emp$, $\Mapsto$ and $*$ are the spatial propositional
connectives of separation logic: 
$\emp$ holds of a heap iff the heap is empty; 
$i \Mapsto v_0,\ldots, v_n$ holds of a heap that contains only the
pointers $i,\ldots,i+n$ storing the values $v_0,\ldots,v_n$,
respectively;
the separating conjunction $P * Q$ holds of a heap that can be divided
into two disjoint subheaps of which $P$ and $Q$ hold respectively.
Separation logic propositions form a PCM with $*$ as the
commutative/associative operation, and $\emp$ its unit.

\newcommand{\len}[1]{\#{#1}}
\newcommand{\llen}[1]{\#(#1)}

The following Hoare triple then applies the \emph{sequence length} function
$\len{(-)}$ to say that, upon $L$'s termination, the contents
$\alpha_0$ of the initial $\nnull$-terminated list is unchanged, but
its length $\len{\alpha_0}$ is deposited into the variable $n$.
\[
\spec{\lseg\ \alpha_0\ (i, \nnull)}\, L\, \spec{\lseg\ \alpha_0\ (i, \nnull) \wedge n = \len{\alpha_0}}
\]
We next outline the part of the proof for this Hoare triple that
examines the loop body in $L$, highlighting three key high-level
aspects that we later adapt to graphs.

\makeatletter
\newcounter{codelength}
\newcommand{\lineno}{\stepcounter{codelength}\textsc{\thecodelength}.\quad}
\newcommand{\linelb}[1]{{\refstepcounter{codelength}\ltx@label{#1}\textsc{\thecodelength}.\quad}}
\makeatother

\begin{figure}[t]
\begin{align*}
\linelb{lp-ln1} & \spec{\exists \alpha\ \beta\ldot \lseg\ \alpha\ (i, j) * \lseg\ \beta\ (j, \nnull) \wedge n = \len{\alpha} \wedge \alpha_0 = \alpha \cat \beta \wedge j \neq \nnull}\\
\linelb{lp-ln2} & \spec{\exists \alpha\ b\ \beta'\ldot \lseg\ \alpha\ (i, j) * \lseg\ (b\cons\beta')\ (j, \nnull) \wedge n = \len{\alpha} \wedge \alpha_0 = \alpha \cat (b \cons \beta')}\\
\linelb{lp-ln3} & \spec{\exists \alpha\ b\ \beta'\ k\ldot \lseg\ \alpha\ (i, j) * j \Mapsto b, k * \lseg\ \beta'\ (k, \nnull) \wedge n = \len{\alpha} \wedge \alpha_0 = \alpha \cat (b \cons \beta')}\\
\linelb{lp-ln4} &    j \assign \deref{j.\textrm{next}};\\
\linelb{lp-ln5} & \spec{\exists \alpha\ b\ \beta'\ j'\ldot \lseg\ \alpha\ (i, j') * j' \Mapsto b, j * \lseg\ \beta'\ (j, \nnull) \wedge n = \len{\alpha} \wedge \alpha_0 = \alpha \cat (b \cons \beta')}\\
\linelb{lp-ln6} & \spec{\exists \alpha\ b\ \beta'\ldot \lseg\ (\alpha\cat\sngl{b})\ (i, j) * \lseg\ \beta'\ (j, \nnull) \wedge n = \len{\alpha} \wedge \alpha_0 = (\alpha\cat\sngl{b}) \cat \beta'}\\
\linelb{lp-ln7} &    n \assign n+1;\\
\linelb{lp-ln10} & \spec{\exists \alpha\ b\ \beta'\ldot \lseg\ (\alpha\cat\sngl{b})\ (i, j) * \lseg\ \beta'\ (j, \nnull) \wedge n = \len{\alpha} + 1 \wedge \alpha_0 = (\alpha\cat\sngl{b}) \cat \beta'}\\
\linelb{lp-ln11} & \spec{\exists \alpha\ b\ \beta'\ldot \lseg\ (\alpha\cat\sngl{b})\ (i, j) * \lseg\ \beta'\ (j, \nnull) \wedge n = \llen{\alpha\cat\sngl{b}} \wedge \alpha_0 = (\alpha\cat\sngl{b}) \cat \beta'}\\
\linelb{lp-ln12} & \spec{\exists \alpha'\ \beta'\ldot \lseg\ \alpha'\ (i, j) * \lseg\ \beta'\ (j, \nnull) \wedge n = \len{\alpha'} \wedge \alpha_0 = \alpha' \cat \beta'}
\end{align*}\vspace{-6mm}
\caption{Proof outline that the loop body of the length-computing program $L$ preserves the loop invariant $I$.}\label{fig:length}
\end{figure}

\paragraph*{\textbf{Dangling pointers.}}
The first aspect is that separation logic inherently relies
on \emph{dangling pointers} to capture computations' intermediate
states. For example, the loop invariant for $L$ is
\[
I \eqdef \exists \alpha\ \beta\ldot \lseg\ \alpha\ (i, j) * \lseg\ \beta\ (j, \nnull) \wedge n = \len{\alpha} \wedge \alpha_0 = \alpha \cat \beta
\]
stating that the original sequence $\alpha_0$ divides into $\alpha$
(processed subsequence) and $\beta$ (remaining subsequence), with $n$
tracking $\alpha$'s length (progress computed so far). Importantly,
pointer $j$ connects $\alpha$'s tail to $\beta$'s head, making it
dangling for $\alpha$ since it references memory outside of $\alpha$'s
heap.

\paragraph*{\textbf{Spatial distributivity.}}
The second aspect is that the $\lseg$ predicate distributes over
$\cat$, in the sense of the following equivalence, already illustrated
in Fig.~\ref{fig:list-apart}.
\begin{align}
\lseg\ (\alpha\cat\beta)\ (i, j) &\iff \exists k\ldot \lseg\ \alpha\ (i, k) * \lseg\ \beta\ (k, j) \label{eq:distrib}
\end{align}
Equation~(\ref{eq:distrib}) is clearly related to the distributivity
of morphisms~(\ref{eq:morphdistrib}), even though sequences don't form
a PCM as concatenation isn't commutative. Nonetheless, the equation
underscores that distributivity, in one form or another, is a crucial
notion in separation logic.
It's typical use in separation logic is to transfer ownership of
pointers between heaps; in the case of $L$, to redistribute $\lseg$ so
that the currently counted element, pointed to by $j$, is moved in
assertions from $\beta$ to $\alpha$.

To illustrate, we review the proof in Fig~\ref{fig:length} that $I$ is
the loop invariant for $L$. Line~\ref{lp-ln1} conjoins $I$ with the
loop condition $j \neq \nnull$, to indicate that the execution is
within the loop. Line~\ref{lp-ln2} derives that $\beta$ is non-empty,
hence of the form $b \join \beta'$, as otherwise $j$ would have been
$\nnull$ by the definition of $\lseg$. Line~\ref{lp-ln3} shows the
first use of distributivity to \emph{detach} $b$ from the head of
$\beta$.  Line~\ref{lp-ln4} mutates $j$ into
$\deref{j.\textrm{next}}$, which is reflected in line~\ref{lp-ln5},
where a fresh variable $j'$ names the value of $j$ before the
mutation. This line derives by the standard inference rules for
pointer mutation and framing~\cite{rey:lics02}. Line~\ref{lp-ln6}
shows the second use of distributivity to \emph{attach} $b$ to the
tail of $\alpha$; it also reassociates the concatenations in
$\alpha_0$ correspondingly. Line~\ref{lp-ln7} increments $n$, which is
reflected in lines~\ref{lp-ln10} and~\ref{lp-ln11}. Finally,
line~\ref{lp-ln12} re-establishes $I$ for the updated values of $j$
and $n$.

\paragraph*{\textbf{Non-spatial distributivity.}}
The third aspect is that the distribution over $\join$ is important
for the non-spatial parts of the proof as well. In particular,
transitioning from line~\ref{lp-ln10} to line~\ref{lp-ln11} requires
that the length function distributes over $\join$; specifically, that
$\llen{\alpha \join b} = \len{\alpha}+1$, or more generally, that
\[\llen{\alpha \join \beta} = \len{\alpha} + \len{\beta}\]
to associate the new value of $n$ to the length of the new processed
sequence $\alpha'$. The transformation occurs in the context $n = (-)$
in lines~\ref{lp-ln10}--\ref{lp-ln11}, illustrating a simple case of
contextual localization.

While standard treatments of separation logic rely on non-spatial
distributivity only implicitly, graph verification requires that it be
made explicit and central. This is because graph specifications
commonly compose morphisms, in turn making contextual localization
essential for effective proofs. By bringing these algebraic
foundations---distributivity and morphisms---to the foreground, our
approach allows them to be systematically leveraged in verification.

\section{Partial Graphs}\label{sec:partialgraphs}
\paragraph*{\textbf{Dangling edges.}}
The example in Section~\ref{sec:overview} shows how
dangling \emph{pointers} link disjoint list segments to specify
intermediate computation states.
Similarly, specifying graph algorithms requires dangling \emph{edges}
to connect disjoint parts of a graph. However, unlike lists, which
have a single dangling pointer, graphs generally have multiple
dangling edges bridging the divide (Fig~\ref{fig:graph-apart}).
To parallel Section~\ref{sec:overview}, where the dangling pointer $j$
was treated as a parameter of $\lseg$, one might consider grouping the
dangling edges into a set that parametrizes the predicate $\gseg$, the
graph analogue to $\lseg$.
\citet{bor+cal+ohe:space04} explored this approach, but it
resulted in an unsatisfactory definition, as the resulting predicate
relied on a program-specific traversal order in addition to the graph
itself.

We instead \emph{embed} dangling edges directly into the graph
representation, yielding the \emph{partial graphs} informally
introduced in Section~\ref{sec:intro}.
Specifically, we choose a graph representation where each node is
associated with its adjacency list, but we allow for the possibility
that a node in the adjacency list \emph{need not be present in the
graph itself}. More formally, a partial graph (or simply graph) of
type $T$ is a \emph{partial finite map} on nodes (isomorphic to
natural numbers), that, if defined on a node $x$, maps $x$ to a pair
consisting of a value of type $T$ ($x$'s contents, value, or mark) and
the sequence of nodes adjacent to $x$ ($x$'s immediate successors,
adjacency list/sequence, children), with the proviso that the map is
undefined on the node $\nnull$ (i.e. $0$).
\[
\partgraph\ T \eqdef \textit{node} \rightharpoonup_{\textit{fin}} T \times \textit{seq}\ \textit{node} 
\]

The \emph{domain} of a map is the finite subset of the type
$\textit{node}$ on which the map is defined. A graph contains an edge
from $x$ to $y$ if $y$ appears among the children of $x$ in the
map. An edge from $x$ to $y$ is dangling if $y$ isn't in the map's
domain.
We write $\bingraph\ T$ for the subtype of $\partgraph\ T$ where the
adjacency list of each node has exactly two elements (left/right
child), with the element set to $\nnull$ if the corresponding child
doesn't exist.
Notation $x \mapsto (v, \alpha)$ denotes the \emph{singleton graph},
comprising just the node $x$, with mark $v$ and adjacency list
$\alpha$. The notation simplifies to $x \mapsto \alpha$ when $T$ is
the unit type, as then the mark $v$ isn't important.
\begin{figure}[t]
  \centering
  \includegraphics[width=0.4\textwidth]{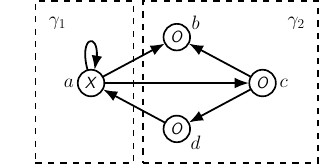}
  \caption{Graph decomposition. Node $a$ has mark $\killed$, and nodes
  $b$, $c$, $d$ have mark $\notM$. Graph $\G_1$ has dangling edges
  from $a$ to $b$ and $c$; graph $\G_2$ has a dangling edge from $d$
  to $a$.}  \label{fig:graph-apart}
\end{figure} 

For example, the graph $\G_1$ in Fig.~\ref{fig:graph-apart} consists
of a single node $a$ with contents $\killed$, and adjacency list $[a,
b, c]$, capturing that $a$ has edges to itself, and to nodes $b$ and
$c$. The latter edges are dangling, as $b$ and $c$ are outside of
$\G_1$. In our notation, $\G_1 = a \mapsto (\killed,[a, b, c])$.

\newcommand{\nodes}{\mathit{nodes}} 
\newcommand{\adj}{\textit{sinks}}

\paragraph*{\textbf{Spatial distributivity.}}
The definition also directly leads to a notion of graph
(de)composition. To see this, notice that finite maps admit the
operation of disjoint union which combines the maps $\G_1$ and $\G_2$,
but only if $\G_1$ and $\G_2$ have disjoint domains.
\[
\G_1 \join \G_2 \eqdef \left\{\begin{array}[c]{ll}
 \G_1\cup \G_2 & \mbox{if $\textit{dom}\ {\G_1} \cap \textit{dom}\ {\G_2} = \emptyset$}\\
 \mbox{undefined} & \mbox{otherwise}
\end{array}\right.
\]
The operation is denoted $\join$ to draw the analogy with
consing/concatenating in the case of sequences, and also because
$\join$ is partial, commutative and associative, meaning that graphs
with $\join$ form a PCM. The unit $\emptymap$ is the empty graph,
i.e. the everywhere-undefined map.
%

For example, the full graph in Fig.~\ref{fig:graph-apart} is a
composition $\G_1 \join \G_2$, where the subgraph $\G_1$ has already
been described, and $\G_2 = b \mapsto (\notM,[])\,\join\,c \mapsto
(\notM,[b,d])\,\join\, d \mapsto (\notM,[a])$.
The representation captures, among other properties, that $\G_2$ has a
node $d$ with a dangling edge to $a$. However, the composite graph
$\G_1\join \G_2$ has no edges dangling.

Similarly to $\lseg$, we can now define the predicate $\gseg$ that
defines how a graph is laid out in the heap.
This can be achieved in several different ways, optimizing for the
structure of the graphs of interest. We will conflate nodes with
pointers and, for general partial graphs, lay out the adjacency list
of each node as a linked list in the heap.
\begin{align*}
  \gseg\, \emptymap\ &\eqdef \emp\\
  \gseg\, (x \mapsto (v, \alpha) \join \G)\ &\eqdef 
   \exists i\ldot x \Mapsto v, i * \lseg\ \alpha\ i\ \nnull * \gseg\, \G
\end{align*}
For binary graphs, however, a simpler representation avoids linking
altogether. Because adjacency lists have exactly two children, a node
$x$ can be represented as a heap with three cells: the pointer $x$
storing the mark $v$, with the \emph{subsequent} pointers $x+1$ and
$x+2$ pointing to $x$'s children.
\begin{align*}
  \gseg\, \emptymap\ &\eqdef \emp\\
  \gseg\, (x \mapsto (v, [l, r]) \join \G)\ &\eqdef x \Mapsto v, l, r * \gseg\, \G
\end{align*}
The following equation~(\ref{dist}) characterizes \emph{both} definitions of
$\gseg$ as PCM morphisms from partial graphs to separation logic
propositions. This lifts heap framing and separation logic reasoning
to graphs, much like $\lseg$ did for sequences with
equation~(\ref{eq:distrib}).
\begin{align}
\gseg\, (\G_1 \join \G_2) \iff \gseg\, \G_1 * \gseg\, \G_2  \label{dist}
\end{align}

To streamline the presentation, in the remainder of the paper we focus
on binary graphs. The simpler definition of $\gseg$ will let us
concentrate on non-spatial reasoning, where partial graphs interact
with PCMs beyond separation logic propositions, and where contextual
localization requires the use of morphisms. We consider the basics of
this interaction next.

\newcommand{\filter}[1]{/_{\!{#1}}} 
\newcommand{\cfilter}[1]{/_{\!{#1}}} 
\newcommand{\rreach}{\textit{reach}}
\newcommand{\reachone}[1]{\rreach\ #1}
\newcommand{\reachtwo}[2]{\rreach\ #1\ #2}
\newcommand{\reachthree}[3]{\reachvia\ #1\ #2\ #3}
\newcommand{\reachvia}{\textit{reach-by}}
\newcommand{\sreach}{\textit{stack-reach}}

\paragraph*{\textbf{Non-spatial distributivity.}}
We first introduce some common notation and write: $\nodes\ \G$
instead of $\mathit{dom}\ \G$ for the set of nodes on which $\G$ is
defined as a finite map, to emphasize that $\G$ is a graph;
$\nodes_0\ \G$ for the disjoint union $\{\nnull\}\,\dotcup\,\nodes\ \G$; and $\Gm\,x$ and
$\Ga\,x$, respectively, for the contents and the adjacency list
of the node $x$, so that
$\G\,x = (\Gm\,x, \Ga\,x)$.

Given the graph $\G$ and node $x$, we define $\G{\setminus} x$ to be
the graph that removes $x$ and its outgoing edges from $\G$, but keep
the edges sinking into $x$ as newly dangling. We can characterize
$\G{\setminus} x$ as a function as follows: $\G{\setminus} x$ agrees
with $\G$ on all inputs except possibly $x$, on which $\G{\setminus}
x$ is undefined. 

We can now state the following two important equalities that expand
(alt.: unfold) the graph $\G$.
\begin{align}
    \G &= (x \mapsto \G\ x) \join (\G{\setminus} x) \qquad \mbox{if $x \in \nodes\ \G$}\label{expand1}\\[2mm]
    \G &= \textstyle\bigbullet\limits_{x\,\in\,\mathit{nodes}\,\G} x \mapsto \G\ x \label{expand2}
\end{align}
Equality~(\ref{expand1}) expands the graph $\G$ around a specific node
$x$, so that $x$'s contents and adjacency list can be considered
separately from the rest of the graph. This is analogous to how a
pointer was separated from $\lseg$ in Section~\ref{sec:overview}, so
that it could be transferred from one sequence to another.
The equality~(\ref{expand2}) iterates the expansion to characterize
the graph in terms of the nodes.

We also require the following functions and combinators over graphs.
\begin{alignat}{3}
  {\G\filter{S}} \eqdef & \!\!\textstyle\bigbullet\limits_{x\,\in\,S\,\cap\,\mathit{nodes}\,\G} && \hspace{-5mm}x \mapsto \G\ x & (\textit{filter}) \label{filterm-def} \\[4mm]
  \erasure{\G} \eqdef &\ \textstyle\bigbullet\limits_{x\,\in\,\mathit{nodes}\,\G}\ && \hspace{-5mm}x \mapsto \Ga\ x & (\textit{erasure}) \label{erasure-def} \\[4mm]
  \textit{map}\ f\ \G \eqdef &\ \textstyle\bigbullet\limits_{x\,\in\,\mathit{nodes}\,\G}\ && \hspace{-5mm}x \mapsto f\ x\ (\G\ x) \qquad& (\textit{map}) \label{map-def} \\[4mm]
  {\G\filter{v_1,\ldots, v_n}} \eqdef &\ {\G\filter{\Gm^{-1}\{v_1,\ldots, v_m\}}} &&& (\textit{filter by contents}) \label{filterc-def}
\end{alignat}
Filtering takes a subset of the nodes of $\G$ that are also in the set
$S$, without modifying the nodes' contents or adjacency lists. Erasure
replaces the contents of each node with the singleton element of unit
type, thus eliding the contents from the notation. In particular
$\erasure{-} : \partgraph\ T \rightarrow \partgraph\ \textit{unit}$.
Map modifies the contents and the adjacency list of each node
according to the mapped function. If $f : node \rightarrow T_1 \times
seq\ node \rightarrow T_2 \times seq\ node$ then $\textit{map}\ f
: \partgraph\ T_1 \rightarrow \partgraph\ T_2$.%
\footnote{When mapping over binary graphs, we allow $f : node \rightarrow T_1 \times
(node\times node) \rightarrow T_2 \times (node \times node)$, as
adjacency lists have exactly two elements. When $T_2 = \mathit{unit}$,
we elide $T_2$ and allow $f : node \rightarrow T_1 \times
(node \times node) \rightarrow node \times node$.}
Filtering by contents selects the nodes whose contents is one of
$v_1, \ldots, v_n$, by (plain) filtering over $\Gm^{-1}\{v_1, \ldots,
v_n\}$. The latter is the inverse image of $\Gm$; thus, the set of
nodes that $\Gm$ maps into $\{v_1, \ldots, v_n\}$.

Sinks of a graph is the set of nodes in the range of $\G$ (i.e., nodes
that possess an incoming edge).
\begin{equation}\label{sinks-def}
  \adj\ \G \eqdef \bigcup\limits_{\nodes\,\G} \Ga
\end{equation}
We then say that $\goodg\ \G$ holds iff $\G$ contains no dangling
edges. A plain (i.e., standard, non-partial) graph is a partial graph
that is closed.
\begin{align}
  \goodg\ \G 
  \eqdef &\ \adj\ \G \subseteq \nodes_0\ \G \label{closed-def}
\end{align}

The final graph primitive in our vocabulary is $\reachtwo{\G}{x}$,
which computes the set of nodes in $\G$ reachable from the node
$x$. The definition is recursive, unioning $x$ with the nodes
reachable from every child of $x$ via a path that avoids $x$. The
definition is well-founded because the (finite) set of graph's nodes
decreases with each recursive call, and thus eventually becomes empty.
\begin{equation}\label{def:reachvia}
  \reachtwo{\G}{x} \eqdef
  \left\{\begin{array}{ll}
    \{x\} \cup \bigcup\limits_{\Ga\,x} \reachone{(\G{\setminus} x)}\ & \mbox{if $x \in \nodes\ \G$} \\
    \emptyset & \mbox{otherwise}
  \end{array}\right.
\end{equation}

\newcolumntype{L}{>{$}l<{$\ }}
\newcolumntype{C}{>{\ $}c<{$\ }}
\begin{figure}[t]
  \rowcolors{2}{gray!7}{white}
  \begin{center}
  \begin{tabular}{L!{\vrule width 1.1pt}C|C|C}
   \rowcolor{gray!15}
    & \G_1 & \G_2 & \G_1\join\G_2 \\ \hline
    (-) & a \mapsto (\killed, [a, b, c]) & 
      \begin{array}[l]{rl}
         &\hphantom{\join}\ b \mapsto (\notM, [])\\
         &\join\ c \mapsto (\notM, [b,d])\\
         &\join\ d \mapsto (\notM, [a])
      \end{array} &
      \begin{array}[l]{rl}
         &\hphantom{\join}\ a \mapsto (\killed, [a, b, c])\\
         &\join\ b \mapsto (\notM, [])\\
         &\join\ c \mapsto (\notM, [b, d])\\
         &\join\ d \mapsto (\notM, [a])
       \end{array} \\
     \nodes\ (-) & \{a\}  & \{b,c,d\} &  \{a,b,c,d\}  \\ 
    {(-)\filter{\{b\}}} & e & b \mapsto (\notM,[]) &  b \mapsto (\notM,[])\\  
  {(-)\filter{X}} & a \mapsto (\killed,[a,b,c])& e & a \mapsto (\killed,[a,b,c])\\ 
    \erasure{-}
    &\begin{array}[l]{rl}
      &a \mapsto [a,b,c]
    \end{array} &
    \begin{array}[l]{rl}
      &\hphantom{\join}\ b \mapsto [] \\
      &\join\ c \mapsto [b,d] \\
      &\join\ d \mapsto [a] 
    \end{array} &
    \begin{array}[l]{rl}
      &\hphantom{\join}\ a \mapsto [a,b,c]  \\
      &\join\ b \mapsto []  \\
      &\join\ c \mapsto [b,d] \\
      &\join\ d \mapsto [a] 
    \end{array}\\ 
    \textit{map}\ f\ (-)
    &\begin{array}[l]{rl}
      &a \mapsto (\notM,[b,c])
    \end{array} &
    \begin{array}[l]{rl}
      &\hphantom{\join}\ b \mapsto (\killed,[]) \\
      &\join\ c \mapsto (\killed,[d]) \\
      &\join\ d \mapsto (\killed,[]) 
    \end{array} &
    \begin{array}[l]{rl}
      &\hphantom{\join}\ a \mapsto (\notM,[b,c])  \\
      &\join\ b \mapsto (\killed,[])  \\
      &\join\ c \mapsto (\killed,[d]) \\
      &\join\ d \mapsto (\killed,[]) 
    \end{array}\\ 
    \adj\ (-) & \{a, b, c\} & \{a, b, d\} & \{a,b,c,d\} \\  \hline
  \goodg\ (-) & \textit{false} & \textit{false} & \textit{true} \\ 
    \reachtwo{(-)}{a} & \{a\} & \emptyset & \{a,b,c,d\} 
\end{tabular} 
\end{center}
  \caption{Action of graph abstractions on the graph
  from Fig.~\ref{fig:graph-apart}. The morphisms are illustrated at
  the top, and the non-morphic (aka.~global) properties at the
  bottom. In the top part, the mapped function $f$ exchanges $\notM$
  and $\killed$ and takes the tail of the adjacency list of each node:
  $f\ x\ (v, \alpha) \eqdef (\overline{v}, \textit{tail}\ \alpha)$,
  where $\overline{\notM} = \killed$ and $\overline{\killed}
  = \notM$.} \label{fig:graph-apart-ex}
\end{figure}

Fig.~\ref{fig:graph-apart-ex} illustrates the definitions on the graph
$\G_1 \join \G_2$ from Fig.~\ref{fig:graph-apart}.  Erasure,
$\textit{map}$, both filters, $\nodes$ and $\adj$ distribute over
$\join$ (and are actually morphisms), while $\goodg$ and $\rreach$
don't.

\begin{lemma}[Morphisms]\label{lemma:distrib}
Functions $\nodes$, $(-)\filter{S}$, $(-)\filter{v_i}$, $\erasure{-}$,
$\textit{map}\ f$, and $\adj$ are morphisms from the PCM of graphs, to
an appropriate target PCM (sets with disjoint union for $\nodes$,
graphs for $(-)\filter{S}$, $(-)\filter{v_i}$, $\erasure{-}$,
$\textit{map}\ f$, and sets with plain union for $\adj$).
\begin{enumerate}[ref=\ref{lemma:distrib} (\arabic*)]
\item $\nodes\ (\G_1\join\G_2) = \nodes\ \G_1\ \dotcup\ \nodes\ \G_2$ and $\nodes\ e = \emptyset$ \label{nodes-dist}
\item ${(\G_1\join\G_2)\filter{S}} = {\G_1\filter{S}} \join  
{\G_2\filter{S}}$ and ${e\filter{S}} = e$ \label{filterm-dist}
\item ${(\G_1\join\G_2)\filter{v_1,\ldots, v_n}} = {\G_1\filter{v_1,\ldots, v_n}} \join
{\G_2\filter{v_1,\ldots,v_n}}$ and ${e\filter{v_1,\ldots, v_n}} = e$ \label{filterc-dist}
\item $\erasure{\G_1\join\G_2} = \erasure{\G_1} \join
\erasure{\G_2}$ and $\erasure{e}= e$ \label{erasure-dist}
\item $\textit{map}\ f\ (\G_1\join\G_2) = \textit{map}\ f\ \G_1 \join
\textit{map}\ f\ \G_2$ and $\textit{map}\ f\ e = e$ \label{map-dist}
\item $\adj\ (\G_1\join\G_2) = \adj\ \G_1 \cup \adj\ \G_2$ and $\adj\ e = \emptyset$ \label{sinks-dist}
\end{enumerate}
\end{lemma}

\begin{lemma}[Filtering]\label{lemma:filter}
\quad
\begin{enumerate}[ref=\ref{lemma:filter} (\arabic*)]
\item ${\G\filter{S_1\, \dotcup\, S_2}} = {\G\filter{S_1}} \join {\G\filter{S_2}}$ and ${\G\filter{\emptyset}} = \emptymap$
\item ${\G\filter{v_1, \ldots, v_m, w_1,\ldots, w_n}} = {\G\filter{v_1,\ldots, v_m}} \join {\G\filter{w_1,\ldots, w_m}}$, for disjoint $\{v_i\}$, $\{w_j\}$\label{filterdisj}
\item ${\G\filter{S_1 \cap S_2}} = {\G\filter{S_1}\filter{S_2}}$
\item ${\G\filter{v_1, \ldots, v_m}\filter{w_1,\ldots, w_n}} = \emptymap$, for disjoint $\{v_i\}$, $\{w_j\}$\label{filter-disj}
\end{enumerate}
\end{lemma}
\begin{lemma}[Mapping]\label{lemma:mapeq}
$\textit{map}\ f_1\ \G = \textit{map}\ f_2\ \G$ iff\ \ 
$\forall x \in \nodes\ \G\ldot f_1\ x\ (\G\ x) = f_2\ x\ (\G\ x)$.
\end{lemma}
\begin{lemma}[Reachability]\label{lemma:reach}
\quad
\begin{enumerate}[ref=\ref{lemma:reach} (\arabic*)]
\item $\reachtwo{\G}{x} = \reachtwo{\erasure{\G}}{x}$ \label{reach-eq}
\item if $y \notin \reachtwo{\G}{x}$, then $\reachtwo{\G}{x} = \reachtwo{(\G {\setminus} y)}{x}$ \label{notin-reach-single}
\item if $y \in \reachtwo{\G}{x}$, then $\reachtwo{\G}{x} = \reachtwo{(\G {\setminus} y)}{x} \cup \reachtwo{\G}{y}$ \label{in-reach-single}
\end{enumerate}
\end{lemma} 
\begin{lemma}[Closure]\label{closed-sub1}
\quad
\begin{enumerate}[ref=\ref{closed-sub1} (\arabic*)]
\item If $\goodg\ \G$ then $\goodg\ ({\G\filter{\rreach\,\G\,x}})$, for every $x$. \label{closed-sub11}
\item If $\goodg\ (\G_1 \join \G_2)$, and $x \in \nodes\ (\G_1 \join \G_2)$, and $\nodes\ \G_1 = \rreach\
  (\G_1 \join \G_2)\ x$ (i.e., $\G_1$ is the subgraph reachable from
  $x$), then $\goodg\ \G_1$, and $x \in \nodes\ \G_1$, and
  $\nodes\ \G_1 = \rreach\ \G_1\ x$. \label{closed-sub2}
\end{enumerate}
\end{lemma}

We elide the proofs here (they are in our Coq graph library), and just
comment on the intuition behind each. Lemmas~\ref{lemma:distrib}
and~\ref{lemma:filter} hold because combinators iterate a node-local
transformation over a set of nodes.
Lemma~\ref{lemma:mapeq} holds because $\textit{map}$ applies the
argument function only to nodes in the graph.
Lemma~\ref{reach-eq} holds because reachability isn't concerned with
the contents of the nodes.
Lemma~\ref{notin-reach-single} holds because a node $y$ that isn't
reachable from $x$ doesn't influence the reachability relation, and
can thus be removed from the graph. Lemma~\ref{in-reach-single} holds
because, given $y$ that's reachable from $x$, another node $z$ is
reachable from $x$ iff it's reachable from $y$, or is otherwise
reachable from $x$ by a path avoiding $y$.
Lemma~\ref{closed-sub11} restates in the notation of partial graphs
the well-known property that in a standard non-partial graph, the
nodes reachable from some node $x$ form a connected
subgraph. Lemma~\ref{closed-sub2} is a simple consequence of
Lemma~\ref{closed-sub11}.

We close by illustrating how contextual localization helps prove that
a global property (here, $\goodg$) is preserved under graph
modifications---a pattern used extensively in the Schorr-Waite and
union-find verifications (Sections~\ref{sec:schorr}
and~\ref{sec:unionfind}). The idea is captured in the following lemma
and proof.

\begin{lemma}\label{lemma:closedlocal}
Let $\G$ be a binary graph such that $\goodg\ \G$, and $x \in \nodes\ \G$ and
$y \in \nodes_0\ \G$. The graph $\G'$ obtained by modifying $x$'s
child (left or right) to $y$, also satisfies $\goodg\ \G'$.
\end{lemma}
\begin{proof}
Without loss of generality, we assume that it's the left child of $x$
that's modified. In other words, we take $\G = x \mapsto (v,
[x_l,x_r]) \join \G{\setminus} x$ and $\G' = x \mapsto (v,
[y,x_r]) \join \G{\setminus} x$, for some $v$, $x_l$ and $x_r$.

Having assumed $\goodg\ \G$, we need to prove $\goodg\ \G'$; that is
$\adj\ \G' \subseteq \nodes_0\ \G'$.
While $\goodg$ itself is a global property, the main components of its
definition, $\adj$ and $\nodes$, are morphisms. The proof rewrites
within the context $- \subseteq -$ to distribute the morphisms as
follows.
\[
\begin{array}[c]{rl>{\ \ }l<{}}
\adj\ \G'\!\! &= \ \ \adj\ (x \mapsto (v, [y,x_r]) \join \G{\setminus}x)=  & \\
              &= \ \ \adj\ (x \mapsto (v, [y,x_r]) \cup \adj\ \G{\setminus}x & \text{Distributivity of $\adj$}\\
              &= \ \ \{y,x_r\} \cup \adj\ \G{\setminus}x & \text{Definition of $\adj$} \\ 
              &\subseteq \ \ \{y\} \cup \{x_l,x_r\} \cup \adj\ \G{\setminus}x &  \\
              &= \ \ \{y\} \cup \adj\ (x \mapsto (v, [x_l,x_r])\cup \adj\ \G{\setminus}x & \text{Definition of $\adj$} \\
              &= \ \ \{y\} \cup \adj\ (x \mapsto (v, [x_l,x_r]) \join \G{\setminus}x) &\text{Distributivity of $\adj$}\\
              &= \ \ \{y\} \cup \adj\ \G\\
              &\subseteq \ \ \nodes_0\ \G & \text{By assumptions $y \in \nodes_0\ \G$}\\
              &                           & \text{and $\goodg\ \G$ (i.e.}\\
              &                           & \text{$\adj\ \G \subseteq \nodes_0\ \G$)}\\
              &= \ \ \{\nnull\}\,\dotcup\,\nodes\ (x \mapsto (v, [x_l,x_r]) \join \G{\setminus}x) \\
              &= \ \ \{\nnull\}\,\dotcup\,\nodes\ (x \mapsto (v, [x_l,x_r]))\,\dotcup\,\nodes\ {\G{\setminus}x} & \text{Distributivity of $\nodes$}\\
              &= \ \ \{\nnull\}\,\dotcup\,\{x\}\, \dotcup\,\nodes\ \G{\setminus}x & \text{Definition of $\nodes$}\\
              &= \ \ \{\nnull\}\,\dotcup\,\nodes\ (x \mapsto (v, [y,x_r])\,\dotcup\, \nodes\ {\G{\setminus}x} &\text{Definition of $\nodes$}\\
              &= \ \ \nodes_0\ (x \mapsto (v, [y,x_r]) \join \G{\setminus}x) &\text{Distributivity of $\nodes$}\\
              &= \ \ \nodes_0\ \G'
\end{array} 
\]
By distributing the morphisms, the proof separates the reasoning about
the node $x$ from that about the subgraph $\G{\setminus}x$. While it
manipulates the values of $\adj$ and $\nodes$ at $x$ to relate
$\adj\ \G'$ and $\nodes_0\ \G'$ to $\adj\ \G$ and $\nodes_0\ \G$,
respectively, it also exploits the fact that $\G$ and $\G'$ share the
same subgraph $\G{\setminus}x$.
By isolating the treatment of $x$ from the rest of the graph, the
proof enacts a style of reasoning that's fundamentally local, albeit
distinct from the typical notion of locality of separation logic
that's achieved by framing. This alternative, contextual, form of
locality arises not from eliding the rest of the structure, but from
explicitly factoring it through morphisms. 
\end{proof}

\newcommand{\marked}{\textit{marked}} 
\newcommand{\pre}{\textit{pre}}
\newcommand{\post}{\textit{post}}
\newcommand{\inv}{\textit{inv}}
\newcommand{\invx}[4]{\inv\ {#1}\ {#2}\ {#3}\ {#4}}
\newcommand{\invp}{\inv'}
\newcommand{\invpx}[5]{\invp\ {#1}\ {#2}\ {#3}\ {#4}\ {#5}}
\newcommand{\gdiff}{\textit{restore}}
\newcommand{\gdiffx}[3]{\gdiff\ {#2}\ {#3}\ {#1}}
\newcommand{\smarked}{\textit{stack-marked}}
\newcommand{\sconsec}{\textit{inset}}
\newcommand{\sconsecx}[2]{\sconsec\ {#2}\ {#1}}
\newcommand{\allu}[1]{#1 = {#1\filter{\notM}}}
\newcommand{\allm}[1]{#1 = {#1\filter{\killed}}}
\newcommand{\setmarknull}[1]{\marked_0\ #1}
\newcommand{\xinmarknull}[2]{#1 \in {#2\filter{M}} \vee t=\nnull}
\newcommand{\unique}{\textit{uniq}}
\newcommand{\empl}{\textit{nil}}
\newcommand{\pathp}{\textit{path}} 
\newcommand{\proj}[3]{#1_{#2}\ #3}
\newcommand{\pathViaU}{\textit{pathViaU}}
\newcommand{\St}{\alpha}

\newcommand{\llast}{\textit{last}}
\newcommand{\hhead}{\textit{head}}
\newcommand{\nnext}{\textit{next}}
\newcommand{\pprev}{\textit{prev}}
\newcommand{\reachviau}{\textit{reach-{\scriptsize$\notM$}}}

\newcommand{\GG}{\G_0}
\section{Schorr-Waite Algorithm}\label{sec:schorr}
The goal of a graph-marking algorithm is to traverse a graph starting
from some root node $r$ and mark the reachable nodes. An obvious way
to implement this functionality is as a recursive function that
traverses the graph in a depth-first, left-to-right manner. However,
as graph marking is typically employed in garbage collection---when
space is sparse---recursive implementation isn't optimal, as it uses
up space on the stack to keep track of the recursive calls and execute
backtracking. The idea of Schorr-Waite's algorithm is that the
information about backtracking can be maintained within the graph
itself, while the graph is traversed \emph{iteratively} in a loop.

We consider the variant of Schorr-Waite that operates over binary
graphs. We also record the node's status in the traversal by setting
the mark to: $\notM$ if the node is unmarked, i.e., the traversal
hasn't encountered the node; $\leftM$ if the node has been traversed
once towards the left subgraph; $\rightM$ if the traversal of the left
subgraph has completed, and the traversal of the right subgraph began;
$\killed$ if both subgraphs have been traversed.
Thus, we proceed to use graphs of type $\bingraph\ \{\notM, \leftM, \rightM,
\killed\}$. 

\begin{figure}[t]
  \centering
  \includegraphics[trim=9mm 0 0 0,clip,width=0.8\textwidth]{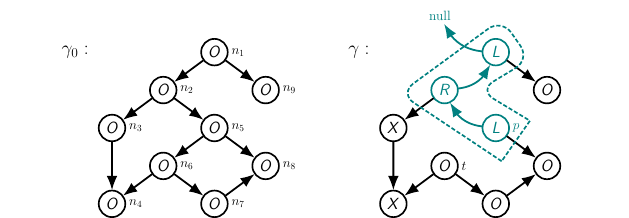} 
  \caption{The unmarked graph $\GG$ with nodes $n_1,...,n_9$, named in
    depth-first, left-to-right traversal order. The partially marked
    graph $\G$ shows an intermediate state, with $t$ the current tip
    node, and $p$ its predecessor. The dashed line encloses the
    traversal stack (top = $p$), which contains exactly the nodes
    marked $\leftM$ or $\rightM$. A node marked $\leftM$
    (resp.~$\rightM$) has its left (resp.~right) edge redirected to
    the predecessor through which it was reached. This inverts the
    edges on the stack: the path $n_1\rightarrow n_2 \rightarrow n_5$
    in $\GG$ becomes $n_5 \rightarrow n_2 \rightarrow n_1$ in $\G$.}
\label{fig:initialcurrentG}
\end{figure}

Along with modifying the nodes' marks, the edges of the graph are
modified during traversal to keep the backtracking information.
This is illustrated in Fig.~\ref{fig:initialcurrentG}, where $\GG$ is
the original unmarked graph, and $\G$ is an intermediate graph halfway
through the execution.
The figure illustrates the traversal that started at $n_1$, proceeded
to $n_2$, fully marked the left subgraph rooted at $n_2$, reached the
node $p = n_5$ and is just about to traverse the left subgraph of $p$
starting from the node $t = n_6$.  The variables $t$ (tip) and $p$
(predecessor) are modified as the traversal advances. The idea of
Schorr-Waite is that once the traversal has obtained the node $t$ from
which to proceed, the corresponding edge (or pointer) from $p$ to $t$
can temporarily be repurposed and redirected towards $p$'s predecessor
in the traversal (here $n_2$).
As similar repurposing has been carried out for $n_2$ and $n_1$ when
they were first encountered, this explicitly inverts in $\G$ the path
encompassing the sequence of nodes $\St = [n_1, n_2, n_5]$ in $\G_0$
(dashed line in Fig.~\ref{fig:initialcurrentG}). The sequence $\St$
serves the same role as the call stack in a recursive implementation,
as it records the nodes whose subgraphs are currently being traversed,
and the relative order in which each node has been reached. We refer
to $\St$ as \emph{the stack}, with $\St$'s \emph{last} element (also
stored in $p$) being the stack's top. A node on the stack can be
marked $\leftM$ or $\rightM$, but not $\notM$ or $\killed$, as the
latter signifies that the node's traversal hasn't started, or has
finished, respectively. Conversely, a node marked $\leftM$ or
$\rightM$ must be on the stack. These properties constitute some of
the main invariants of the algorithm and are formalized in
Section~\ref{sec:invariants}.

\begin{figure}[t]
  \centering
  \includegraphics[trim=4mm 0 0 0,clip,width=0.8\textwidth]{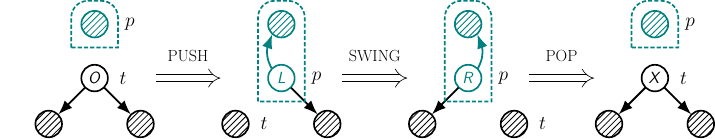} 
\caption{Operations PUSH, SWING, and POP on the pivot node (unshaded),
  coordinate the marking of the pivot with the stack (dashed line) and
  with edge inversion. SWING and POP require that $t$ is marked or 
  $\nnull$.}\label{fig:markcycle}
\end{figure}
Fig.~\ref{fig:markcycle} zooms onto the tip node $t$ (unshaded node,
pivot) to illustrate how the traversal coordinates the marking and
edge redirection around $t$ with the modifications of $t$ and $p$ in
three separate operations: PUSH, SWING, and POP. When $t$ is first
encountered, it's unmarked. PUSH promptly marks it $\leftM$, and
pushes it onto the stack (dashed line in Fig.~\ref{fig:markcycle}) by
redirecting its left edge towards $p$. The traversal continues by
advancing $t$ towards the left subgraph, and $p$ to what $t$
previously was. Once the left subgraph is fully marked, SWING restores
pivot's left pointer, but keeps the pivot enlinked onto the stack by
using the pivot's right pointer. The tip is swung to the right
subgraph, and the mark changed to $\rightM$, to indicate that the left
subgraph has been traversed, and we're moving to the right
subgraph. Once the pivot's right subgraph is traversed as well, POP
sets the mark to $\killed$ and restores the right edge. This returns
the edges of the pivot to their originals from the initial graph, but
also unlinks (i.e., pops) the pivot from the stack. Nodes stored into
$p$ and $t$ are correspondingly shifted up, and the marking cycle
continues from the new top of the stack.

\newcommand{\tm}{\textit{tm}}
\newcommand{\pM}{\textit{pm}}
\newcommand{\tmp}{\textit{tmp}}
\renewcommand{\algorithmiccomment}[2][.2\linewidth]{%
  \leavevmode\hfill\makebox[#1][l]{//~#2}}

\begin{figure}[t]
  \centering
  \includegraphics[width=\textwidth]{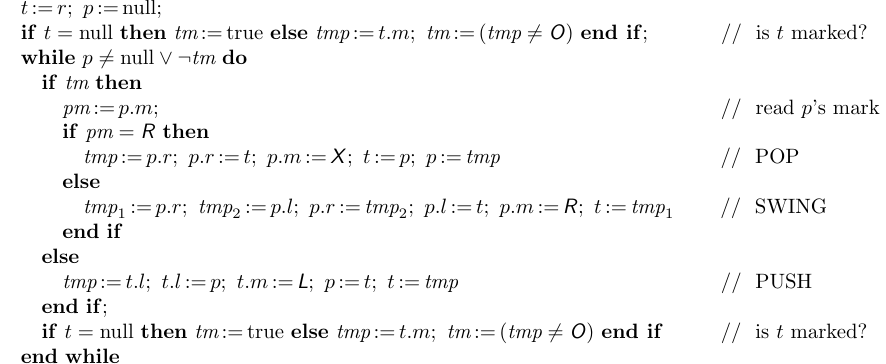} 
 \caption{Schorr-Waite algorithm. The algorithm follows the monadic
   style common to many separation logics, strictly dividing stateful
   \emph{commands} from pure \emph{expressions}. This requires
   pointer-dependent conditions (in if/while) to first dereference
   into variables, since expressions can't contain commands.
For example, $\tm$ is
   assigned before and at the loop's end to enable its use in the
   loop condition.}\label{alg:cap}
\end{figure}

Algorithm in Fig.~\ref{alg:cap} takes the root node $r$ from which the
traversal begins, and starts by setting the tip $t$ to $r$, and $p$ to
$\nnull$. An invariant of the algorithm is that $p$ is always the top
of the stack. Thus, $p$ equals $\nnull$ if the stack is empty,
otherwise $p$ is marked $\leftM$ or $\rightM$, but never $\notM$ or
$\killed$.  Next, $t$'s marking status is computed into $\tm$. In
general, if $t$ is unmarked and non-$\nnull$, then $t$ is encountered
for the first time.
At each iteration, the algorithm mutates the graph---illustrated
in Fig.~\ref{fig:mainops}---using: (1) PUSH, if $t$ is unmarked and
non-null, thus encountered for the first time, and can be pushed onto
the stack; (2) SWING, if $t$ is marked or $\nnull$ (encountered
before), and $p$ isn't marked $\rightM$ (thus, \emph{is} marked
$\leftM$) signifying that $p$'s left subgraph has been traversed; (3)
POP, if $t$ is marked or $\nnull$ and $p$ is marked $\rightM$
signifying that $p$'s right subgraph has been traversed. The
algorithm terminates when $p$ is $\nnull$ (the stack is empty) and $t$
is marked or $\nnull$ (there's nothing to push onto the stack).

\begin{figure}[t]
  \centering
  \includegraphics[width=\textwidth]{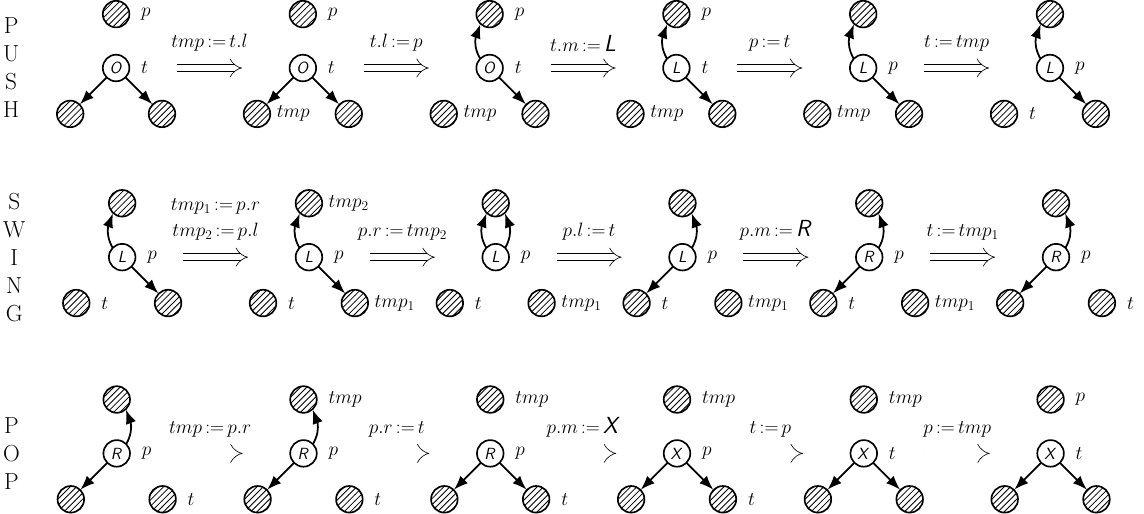} 
\caption{Details of the main operations of Schorr-Waite.}\label{fig:mainops}
\end{figure}

The verification task is to prove that the graph obtained upon
termination modifies the initial unmarked graph as follows: (1) the
nodes reachable from $r$, and only those, are marked $\killed$, and
(2) the edges are restored to their initial versions.

\section{Schorr-Waite's Invariants}\label{sec:invariants}
To formalize the above specification, the first step is to state the
main invariants of Schorr-Waite that relate the initial graph $\GG$ to
the current graph $\G$, the stack $\alpha$,
and the nodes $t$ and $p$.
In this formalization, we rely solely on combining graph abstractions
from the small vocabulary introduced in
Section~\ref{sec:partialgraphs}. This ensures that later proofs remain
compatible with graph decomposition, and depend only on general lemmas
over this core vocabulary.

The relations to be captured are the following, all observed in the
example graphs from Fig.~\ref{fig:initialcurrentG}. These invariants
hold throughout the execution of Schorr-Waite, \emph{except} inside
the three main operations, PUSH, SWING, and POP, when the invariants
are temporarily invalidated, to be restored by the time the operation
terminates.
\begin{enumerate}
\item[(\ref{invbook})] The stack $\St$ is a sequence of distinct nodes
  (i.e., each node is unique), also distinct from $\nnull$, with $p$
  the top of stack (i.e., $p$ is the last element of $\St$, or
  $\nnull$ if $\St$ is empty). The tip $t$, being a child of $p$, is a
  node in $\G$, or $\nnull$ if the child doesn't exist.
\item[(\ref{invclosed})] The graph $\G$ is closed, i.e., it has no dangling edges.
\item[(\ref{invstack})] The stack $\St$ contains exactly the nodes
  that are partially marked (labeled $\leftM$ or
  $\rightM$). Intuitively, this holds because $\St$ implements the
  call stack, and thus only contains the nodes whose subtrees are
  being currently marked.
\item[(\ref{invinvert})] The stack $\St$ describes a path in $\G$ that
  \emph{respects the markings}, in the following sense: if a node in
  $\St$ is marked $\leftM$ (resp.~$\rightM$), then its left
  (resp.~right) child in $\G$ is the node's predecessor in the
  traversal, and thus the node's predecessor in $\alpha$.  For
  example, in the graph $\G$ in Fig.~\ref{fig:initialcurrentG}, the
  node $n_5$ is marked $\leftM$ and its left child $n_2$ is its
  predecessor in $\St$. In other words, $\St$ stores the nodes in the
  relative order in which they are reached.
\item[(\ref{invrestore})] Reorienting the edges of the nodes in $\St$
  produces the original graph $\GG$ (modulo node marks).
\item[(\ref{invreach})] The unmarked nodes (labeled $\notM$) are ``to
  the right'' of $\St$ and $t$, because the algorithm implements a
  left-to-right traversal order. More precisely, each unmarked node is
  reachable, by an unmarked path, either from $t$ or from a right
  child of some node in $\St$. For example, in the graph $\G$ in
  Fig.~\ref{fig:initialcurrentG}, the nodes $n_7, n_8, n_9$ (unmarked)
  are ``to the right'' of $\St$ and $t$, whereas the nodes $n_3, n_4$
  (marked $\killed$) are ``to the left''.
\end{enumerate}

\begin{figure}[t]
\begin{align}
\!\!\!\! \invpx{\GG}{\G}{\St}{t}{p} \eqdef\ 
  &   \unique\ (\nnull \join \St) \wedge p = \llast\ (\nnull \join \St) \wedge t \in \nodes_0\ \G \wedge \hbox{} \tag{$a$} \label{invbook}\\
  &   \goodg\ \G \wedge \hbox{}\tag{$b$} \label{invclosed}\\
  &   \nodes\ {\G\filter{\leftM,\rightM}} = \St \wedge \hbox{}\ \tag{$c$}\label{invstack}\\
  &   \sconsecx{\G}{\St} = \erasure{\G} \wedge \hbox{} \tag{$d$}\label{invinvert}\\
  &   \gdiffx{\G}{\St}{t} = \erasure{\GG} \wedge \hbox{} \tag{$e$}\label{invrestore}\\
  &   \nodes\ {\G\filter{\notM}} \subseteq {\textstyle \bigcup\limits_{\St}}\ (\reachone{(\G\filter{\notM})}\circ\G_r) \cup \reachtwo{(\G\filter{\notM})}{t} \tag{$f$}\label{invreach}\\
  \invx{\GG}{\G}{t}{p} \eqdef\ & \exists \St\ldot \invpx{\GG}{\G}{\St}{t}{p} \notag 
\end{align}
\caption{Main invariants of Schorr-Waite, formally. The parameter
  $\G_0$ is the initial graph, $\G$ is the current partially marked
  graph, $\St$, $t$ and $p$ are the current stack, tip and
  predecessor, respectively.}\label{fig:loopinv}
\end{figure}

Fig.~\ref{fig:loopinv} presents the formal statements of the above
invariants, which we proceed to explain. 
The encoding of (\ref{invbook}) and (\ref{invclosed}) in
Fig.~\ref{fig:loopinv} is direct. The statement (\ref{invstack}) says
that the stack $\St$, viewed as a set rather than a sequence, equals
the set of nodes in $\G$ that are marked $\leftM$ or $\rightM$.
Statement (\ref{invreach}) denotes by $\G_r\ x$ the right child of $x$
in $\G$,\footnote{Dually, $\G_l\ x$ is the left child of $x$, so that
$\Ga\ x$ is the sequence $[\G_l\ x, \G_r\ x]$.} and directly says that
an unmarked node $x$ (i.e., $x \in \nodes\ {\G\filter{\notM}}$) is
reachable from $t$ by an unmarked path ($x \in
\reachtwo{(\G\filter{\notM})}{t}$), or there exists a node $y \in \St$
such that $x$ is reachable from the right child of $y$, also by an
unmarked path ($x \in
\bigcup_{\St}\ (\reachone{(\G\filter{\notM})}\circ\G_r)$).
This leaves us with the statements (\ref{invinvert}) and
(\ref{invrestore}), which we discuss next.

If one ignores for a moment the property of respecting the markings,
then the English description of (\ref{invinvert}) and
(\ref{invrestore}) says that the stack $\St$ is the contents of a
linked list embedded in the graph $\G$, and that reversing the linkage
of this list produces the original graph $\GG$. Thus, one might
consider following the approach to linked lists described in
Section~\ref{sec:overview}, and attempt to relate $\St$ with the
\emph{heap layouts} of $\G$ and $\GG$, by somehow relating
$\lseg\,\St$ with $\gseg\ \GG$ and $\lseg\,(\textit{reverse}\,\St)$
with $\gseg\ \G$.\footnote{In fact, relating $\St$ to the heap layouts
of $\G$ and $\GG$ via the $\lseg$ predicate is precisely the approach
of Yang's proof.}
However, in our setting there is a more direct option, that elides the
detour through heaps, and relates $\St$ to $\G$ and $\GG$ by means of
the following graph transformation.
 
In particular, we first define the \emph{higher-order} function
$\textit{if-mark}$ that branches on the marking of an individual node,
to output a transformation of the node's edges.
\[
\begin{array}{rcl}
\textit{if-mark} & : & (node \rightarrow node) \rightarrow node \rightarrow \hbox{}\\
& & \qquad \{\notM, \leftM, \rightM, \killed\} \times (node \times node) \rightarrow node \times node \\[2ex]
\textit{if-mark}\ f\ x\ (m, [l, r]) & \eqdef &
  \left\{\begin{aligned}
    &[f\ x,\ r] && \mbox{if $m = \leftM$} \\
    &[l,\ f\ x] && \mbox{if $m = \rightM$} \\
    &[l, r] && \mbox{otherwise}
  \end{aligned}\right.
\end{array}
\]
The transformation that $\textit{if-mark}$ computes is guided by the
argument function $f$, itself mapping nodes to nodes. Thus, supplying
different values for $f$ produces different transformations, but
within the general pattern encoded by $\textit{if-mark}$.
In more detail, $\textit{if-mark}$ applies to a node $x$ of $\G$,
$x$'s mark $m$ and children $[l, r]$ as follows. If $m$ is $\leftM$
(resp.~$\rightM$), then $x$'s left (resp.~right) child is replaced in
the output by $f\ x$, while the right (resp.~left) child is passed
along unchanged. Otherwise, if $x$ is marked $\notM$ or $\killed$, its
children are returned unmodified.

We then apply the $\textit{map}$ combinator to $\textit{if-mark}$ in
two different ways, to define the following two functions that will
help us formalize the invariants (\ref{invinvert}) and
(\ref{invrestore}) in a uniform way.
\begin{align}
\begin{split}
  \sconsec &\ \, :\ \textit{seq}\ node \rightarrow \bingraph\ \{\notM,\leftM,\rightM,\killed\} \rightarrow \bingraph\ \textit{unit}\\
  \sconsecx{\G}{\St} &\eqdef \textit{map}\ (\textit{if-mark}\ (\pprev\ (\nnull \join \St)))\ \G 
\end{split} \label{invert-def}\\
\begin{split}
  \gdiff &\ \, :\ node \rightarrow \textit{seq}\ node \rightarrow \bingraph\ \{\notM,\leftM,\rightM,\killed\} \rightarrow \bingraph\ \textit{unit}\\
  \gdiffx{\G}{\St}{t}&\eqdef \textit{map}\ (\textit{if-mark}\ (\nnext\ (\St \join t)))\ \G
 \end{split} \label{restore-def}
\end{align}
Here $\pprev\ (\nnull \join \St)$ is a function that takes a node $x$
and returns the predecessor of the first occurrence of $x$ in $(\nnull
\join \St)$ if $x \in \St$, or $\nnull$ if $x \notin \St$.  Similarly,
$\nnext\ (\St \join t)\ x$ is the successor of the first occurrence of
$x$ in $(\St \join t)$ if $x \in \St$, or $t$ if $x \notin \St$.
For example, if $\St$ is the sequence $[n_1, n_2, n_5]$, then
$\pprev\ (\nnull \join \St)\ n_2 = n_1$, $\pprev\ (\nnull \join
\St)\ n_1 = \nnull$, $\nnext\ (\St \join t)\ n_2 = n_5$, and
$\nnext\ (\St \join t)\ n_5 = t$.

From the definition, it follows that graphs $\sconsecx{\G}{\St}$ and
$\gdiffx{\G}{\St}{t}$ modify the graph $\G$ by manipulating the edges
related to the stack $\St$, and otherwise erasing $\G$'s marks.
Specifically, $\sconsecx{\G}{\St}$ replaces the child (left/right,
based on mark) of each node in $\St$ with its predecessor in $\nnull
\join \St$. For $\St$ to respect the markings, as required by the
invariant (\ref{invinvert}), this transformation must actually
preserve the graph. Thus, invariant (\ref{invinvert}) is formally
stated in Fig.~\ref{fig:loopinv} as $\sconsecx{\G}{\St} =
\erasure{\G}$.
Similarly, $\gdiffx{\G}{\St}{t}$ replaces the child (left/right, based
on mark) of each node in $\St$ with its successor in $\St \join t$,
thereby inverting the path $\St$ in $\G$, and redirecting the
appropriate child of $p$ towards $t$. Thus, invariant
(\ref{invrestore}) is formally stated in Fig.~\ref{fig:loopinv} as
$\gdiffx{\G}{\St}{t} = \erasure{\G_0}$.

\paragraph*{\textbf{Example.}}
Consider the graphs $\G$ and $\GG$ from
Fig.~\ref{fig:initialcurrentG}, taking $p = n_5$, $t = n_6$, and $\St
= [n_1,n_2,n_5]$. The following calculation illustrates that
$\gdiffx{\G}{\St}{t}$ computes the erasure of $\GG$. 
\begin{align*}
  \gdiffx{\G}{\St}{t} =\ & \textit{map}\ (\textit{if-mark}\ (\nnext\ (\St \cons t)))\ \G \\
  =\ & \hbox{}\hphantom{\join}\ \,n_1\, {\mapsto}\, \textit{if-mark}\ (\nnext\ (\St \cons t))\ n_1\ (\leftM, [\nnull, n_9])\\
    & \hbox{}\join n_2\, {\mapsto}\, \textit{if-mark}\ (\nnext\ (\St \cons t))\ n_2\ (\rightM, [n_3, n_1])\\
    & \hbox{}\join n_5\, {\mapsto}\, \textit{if-mark}\ (\nnext\ (\St \cons t))\ n_5\ (\leftM, [n_2, n_8])\\
    & \hbox{}\join \textit{map}\ (\textit{if-mark}\ (\nnext\ (\St \cons t)))\ (\G{\setminus}\{n_1,n_2,n_5\})\\
   =\ & \hbox{}\hphantom{\join}\ \,n_1\,{\mapsto}\, [\nnext\ (\St \cons t)\ n_1, n_9] \\
     &  \hbox{}\join n_2\, {\mapsto}\, [n_3, \nnext\ (\St \cons t)\ n_2] \\
     & \hbox{}\join n_5\, {\mapsto}\, [\nnext\ (\St \cons t)\ n_5, n_8] \\
     & \hbox{}\join \erasure{\G{\setminus}\{n_1,n_2,n_5\}}\\
  =\ & n_1\,{\mapsto}\,[n_2, n_9] \join n_2\,{\mapsto}\,[n_3, n_5] \join n_5\,{\mapsto}\,[t, n_8] \join \erasure{\G{\setminus}\{n_1,n_2,n_5\}}\\
  =\ & \erasure{\GG}
\end{align*}
When applied to $\G{\setminus}\{n_1,n_2,n_5\}$, the mapping returns
the erasure $\erasure{\G{\setminus}\{n_1,n_2,n_5\}}$, because
$\textit{if-mark}$ elides the contents, and doesn't modify the
children of the nodes that aren't marked $\leftM$ or
$\rightM$. However, the calculations modifies the edges out of $n_1$,
$n_2$ and $n_5$ in $\G$ to obtain precisely the edges of $\GG$.

On the other hand, the role of $\sconsec$ is to ensure that $\St$ is a
marking-respecting path in $\G$, and thus a valid stack. For the same
values of $\St$ and $\G$ as above, $\sconsecx{\G}{\St}$ returns
$\erasure{\G}$, because the passed $\St$ is indeed the stack in
$\G$. However, if we passed a permutation of $\St$, the equality to
$\erasure{\G}$ won't hold. For example, consider passing
$[n_2,n_5,n_1]$ for $\St$.
\begin{align*}
   \sconsecx{\G}{\St} =\ & \textit{map}\ (\textit{if-mark}\ (\pprev\ (\nnull \cons \St)))\ \G \\
   =\ & \hbox{}{\hphantom{\join}}\ \,n_1\, {\mapsto}\, \textit{if-mark}\ (\pprev\ (\nnull \cons \St))\ n_1\ (\leftM, [\nnull, n_9])\\
     & \hbox{} \join n_2\,{\mapsto}\, \textit{if-mark}\ (\pprev\ (\nnull \cons \St))\ n_2\ (\rightM, [n_3, n_1])\\
     & \hbox{} \join n_5\, {\mapsto}\, \textit{if-mark}\ (\pprev\ (\nnull \cons \St))\ n_5\ (\leftM, [n_2, n_8])\\
     & \hbox{} \join \textit{map}\ (\textit{if-mark}\ (\pprev\ (\nnull \cons \St)))\ (\G{\setminus}\{n_1,n_2,n_5\})\\
    =\ & \hbox{}\hphantom{\join}\ \, n_1\,{\mapsto}\,[\pprev\ (\nnull \cons \St)\ n_1, n_9] \\
       & \hbox{} \join n_2\, {\mapsto}\, [n_3, \pprev\ (\nnull \cons \St)\ n_2] \\
      & \hbox{} \join n_5\,{\mapsto}\,[\pprev\ (\nnull \cons \St)\ n_5, n_8] \\
      & \hbox{}\join \erasure{\G{\setminus}\{n_1,n_2,n_5\}}\\
   =\ & n_1\, {\mapsto}\, [n_5, n_9] \join n_2\,{\mapsto}\, [n_3, \nnull] \join n_5\,{\mapsto}\,[n_2, n_8] \join \erasure{\G{\setminus}\{n_1,n_2,n_5\}}\\
   \neq\ & \erasure{\G}
  \end{align*}
  The calculation considers the nodes $n_1$, $n_2$ and $n_5$, as these
  are the nodes marked $\leftM$ or $\rightM$. However, because the
  nodes aren't in the traversed order, the computation fails to
  encode $\erasure{\G}$. \hfill $\square$

\medskip

The significance of using the combinator $\textit{map}$ to define
$\sconsec$ and $\gdiff$ is that the general lemmas from
Section~\ref{sec:partialgraphs} apply, immediately deriving that both
$\sconsec$ and $\gdiff$ are morphisms in the argument $\G$, and thus
facilitating contextual localization in the proofs in
Section~\ref{sec:swproof}.
Similarly, using a higher-order function \textit{if-mark} lets us
prove general lemmas about it once, and use them several times. For
example, Lemma~\ref{sconsec-empty} below states a property about
\textit{if-mark} that immediately applies to both $\sconsec$ and
$\gdiff$. Similarly, Lemma~\ref{lemma:ifmark} is used to prove both
equations (\ref{eq:invert}) and (\ref{eq:gdiff}) of
Lemma~\ref{sconsec-p}, which are in turn each used several times in
the proofs of the various components of Schorr-Waite in
Section~\ref{sec:swproof}.

\begin{lemma}\label{sconsec-empty}
If $\nodes\ {\G\filter{\leftM,\rightM}} = \emptyset$ (i.e., $\G$ has
no partially marked nodes) then
$\textit{map}\ (\textit{if-mark}\ f)\ \G = \erasure{\G}$ for any
$f$. In particular, $\sconsecx{\G}{\St} = \gdiffx{\G}{\St}{t} =
\erasure{\G}$ for any $\St$ and $t$.
\end{lemma}
\begin{lemma}\label{lemma:ifmark}
If $f_1\ x = f_2\ x$ for every 
$x \in \nodes\ {\G\filter{\leftM, \rightM}}$ 
then $\textit{map}\ (\textit{if-mark}\ f_1)\ \G =
\textit{map}\ (\textit{if-mark}\ f_2)\ \G$.
\end{lemma}
\begin{lemma}\label{sconsec-p}
Let $\alpha$ be a sequence containing the nodes of
${\G\filter{\leftM,\rightM}}$. Then the following equations hold.
  \begin{align}
    \sconsecx{\G}{(\St \join p)} &= \sconsecx{\G}{\St}  \label{eq:invert}\\ 
    \gdiffx{\G}{(\St \join p)}{t} &= \gdiffx{\G}{\St}{p} \label{eq:gdiff}
  \end{align}
\end{lemma}

Lemma~\ref{sconsec-empty} holds because \textit{if-mark} modifies only
partially marked nodes, and applies erasure to the rest of the graph.

In Lemma~\ref{lemma:ifmark}, if $f_1$ and $f_2$ agree on partially
marked nodes, then by definition, $\textit{if-mark}\ f_1\ x\ (\G\ x) =
\textit{if-mark}\ f_2\ x\ (\G\ x)$ for all $x \in \nodes\ \G$, because
$\textit{if-mark}$ modifies only partially marked nodes. The
conclusion then follows by Lemma~\ref{lemma:mapeq}.

Equation~(\ref{eq:invert}) follows from Lemma~\ref{lemma:ifmark}, by
taking $f_1\ x= \pprev\ (\nnull \join \St \join p)\ x$ and $f_2\ x =
\pprev\ (\nnull \join \St)\ x$, and observing that 
for every $x \in \nodes\ {\G\filter{\leftM,\rightM}}$
it must be $\pprev\ (\nnull \join \St \join p)\ x =
\pprev\ (\nnull \join \St)\ x$, because $x \in \St$.
Equation~(\ref{eq:gdiff}) follows similarly by taking $f_1\ x =
\nnext\ (\St \join p \join t)\ x$ and $f_2\ x = \nnext\ (\St \join
p)\ x$.

\section{Proof of Schorr-Waite}\label{sec:swproof}
The proof of the algorithm is divided into two steps. Following
\citet{Yang2001AnEO}, Section~\ref{sec:connected} first verifies
Schorr-Waite assuming that the input graph is connected from the root
node $r$.
Section~\ref{sec:pop} verifies the POP fragment, and 
Section~\ref{sec:global} extends to general graphs by framing the
nodes \emph{not} reachable from $r$.
The subproof utilize diverse forms of contextual localization, as we
discuss below.

\subsection{Proof for Connected Graphs}\label{sec:connected}
We first establish the following specification.
\begin{equation}
\begin{array}{c}\label{sw-connect}
\spec{\gseg\ \GG \wedge \goodg\ \GG \wedge r \in \nodes\ \GG \wedge \nodes\ \GG = \reachtwo{(\GG\filter{\notM})}{r}}\\
\textit{Schorr-Waite}\,(r)\\
\spec{\exists \G\ldot \gseg\ \G \wedge \erasure{\GG} = \erasure{\G} \wedge \allm{\G}}
\end{array}
\end{equation}
The precondition in (\ref{sw-connect}) says that the input graph $\GG$
is well-formed ($\gseg\ \GG$), closed, contains the root node $r$, and
is unmarked and connected from $r$ ($\nodes\ \GG =
\reachtwo{(\GG\filter{\notM})}{r}$).
The postcondition posits an ending graph $\G$ which, aside from node
marking, equals the input graph ($\erasure{\GG} = \erasure{\G}$), and
is fully marked itself ($\allm{\G}$).

\makeatletter
\newcounter{codesw}
\renewcommand{\lineno}{\stepcounter{codesw}\textsc{\thecodesw}.\quad}
\renewcommand{\linelb}[1]{{\refstepcounter{codesw}\ltx@label{#1}\textsc{\thecodesw}.\quad}}
\makeatother
\begin{figure}[!t]
  {
\begin{align*}
\linelb{sw-ln01} &\spec{\gseg\ \GG \wedge \goodg\ \GG \wedge r \in \nodes\ \GG \wedge \nodes\ \GG = \reachtwo{(\GG\filter{\notM})}{r}}\\
\linelb{sw-ln02} &\spec{\gseg\ \GG \wedge \invx{\GG}{\GG}{r}{\nnull}}\\
\lineno  &t \assign r;\ p \assign \nnull;\\
\lineno  &\spec{\exists \G\ldot \gseg\ \G \wedge \invx{\GG}{\G}{t}{p}}\\ 
\linelb{sw-ln05} &\textrm{\textbf{if}}\ t = \nnull\ \textrm{\textbf{then}}\ \tm \assign \textrm{true}\
  \textrm{\textbf{else}}\ \tmp \assign \deref{t.m};\ \tm \assign (\tmp \neq \notM)\ \textrm{\textbf{end if}};\\
\linelb{sw-ln06} &\spec{\exists \G\ldot \gseg\ \G \wedge \invx{\GG}{\G}{t}{p} \wedge \tm = (t \in \setmarknull{\G})} 
\qquad \hspace{20mm}\textrm{// loop invariant}\\
\lineno  &\textrm{\textbf{while}}\ p \neq \nnull \vee \neg \tm\ \textrm{\textbf{do}}\\
\linelb{sw-ln08} &\quad \spec{\exists \G\ldot \gseg\ \G \wedge \invx{\GG}{\G}{t}{p} \wedge \tm = (t \in \setmarknull{\G}) \wedge
    (p \neq \nnull \vee \neg \tm)}\\
\lineno  &\quad \textrm{\textbf{if}}\ \tm\ \textrm{\textbf{then}}\\
\lineno  &\qquad \spec{\exists \G\ldot \gseg\ \G \wedge \invx{\GG}{\G}{t}{p} \wedge
    p \neq \nnull \wedge t \in \setmarknull{\G}}\\
\linelb{sw-ln11} &\qquad \pM \assign \deref{p.m};\\ 
\lineno  &\qquad \spec{\exists \G\ldot \gseg\ \G \wedge \invx{\GG}{\G}{t}{p} 
   \wedge t \in \setmarknull{\G} \wedge \Gm\ p = \pM \in \{L, R\}}\\
\lineno  &\qquad \textrm{\textbf{if}}\ \pM = \rightM\ \textrm{\textbf{then}}\\
\linelb{sw-ln14}  &\qquad\quad \spec{\exists \G\ldot \gseg\ \G \wedge \invx{\GG}{\G}{t}{p}
    \wedge t \in \setmarknull{\G} \wedge \Gm\ p = \rightM}\\
\linelb{sw-ln15} &\qquad\quad \tmp \assign \deref{p.r};\ p.r \mutate t;\ p.m \mutate \killed;\
  t \assign p;\ p \assign \tmp \hspace{31.1mm} \textrm{//\ POP} \\
\linelb{sw-ln16}  &\qquad\quad \spec{\exists \G\ldot \gseg\ \G \wedge \invx{\GG}{\G}{t}{p}}\\ 
\lineno  &\qquad \textrm{\textbf{else}}\ \\
\linelb{sw-ln18}  &\qquad\quad \spec{\exists \G\ldot \gseg\ \G \wedge \invx{\GG}{\G}{t}{p}
    \wedge t \in \setmarknull{\G} \wedge \Gm\ p = \leftM}\\
\linelb{sw-ln19} &\qquad\quad \tmp_1 \assign \deref{p.r};\ \tmp_2 \assign \deref{p.l};\ p.r \mutate \tmp_2;\
    p.l \mutate t;\ p.m \mutate \rightM;\ t \assign \tmp_1\ \quad \textrm{//\ SWING} \\
  \linelb{sw-ln20}  &\qquad\quad \spec{\exists \G\ldot \gseg\ \G \wedge \invx{\GG}{\G}{t}{p}}\\ 
\lineno  &\qquad \textrm{\textbf{end if}} \\
\lineno  &\quad \textrm{\textbf{else}}\ \\ 
\linelb{sw-ln23}  &\qquad \spec{\exists \G\ldot \gseg\ \G \wedge \invx{\GG}{\G}{t}{p} \wedge t \notin \setmarknull{\G}}\\
\linelb{sw-ln24} &\qquad \tmp \assign \deref{t.l};\ t.l \mutate p;\ t.m \mutate \leftM;\
    p \assign t;\ t \assign \tmp\ \hspace{36mm}\ \textrm{//\ PUSH} \\
\linelb{sw-ln25}  &\qquad \spec{\exists \G\ldot \gseg\ \G \wedge \invx{\GG}{\G}{t}{p}}\\ 
\lineno  &\quad \textrm{\textbf{end if}}; \\
\linelb{sw-ln27} &\quad \spec{\exists \G\ldot \gseg\ \G \wedge \invx{\GG}{\G}{t}{p}}\\ 
\linelb{sw-ln28} &\quad \textrm{\textbf{if}}\ t = \nnull\ \textrm{\textbf{then}}\ \tm \assign \textrm{true}\
  \textrm{\textbf{else}}\ \tmp \assign \deref{t.m};\ \tm \assign (\tmp \neq \notM)\ \textrm{\textbf{end if}} \\
\lineno  &\quad \spec{\exists \G\ldot \gseg\ \G \wedge \invx{\GG}{\G}{t}{p} \wedge \tm = (t \in \setmarknull{\G})}\\
\lineno  &\textrm{\textbf{end while}}\\ 
\linelb{sw-ln31} &\spec{\exists \G\ldot \gseg\ \G \wedge \invx{\GG}{\G}{t}{p} \wedge p = \nnull \wedge t \in \setmarknull{\G}}\\
\linelb{sw-ln32} &\spec{\exists \G\ldot \gseg\ \G \wedge \erasure{\GG}\ = \erasure{\G} \wedge \allm{\G}}
\end{align*}
}%
\caption{Proof outline for the connected graph specification of
  Schorr-Waite. The fragments verifying the pointer primitives in
  lines~\ref{sw-ln05}, \ref{sw-ln11}, \ref{sw-ln15}, \ref{sw-ln19},
  \ref{sw-ln24}, \ref{sw-ln28} are elided, and illustrated
  separately. 
Notation $\setmarknull{\G}$ abbreviates
  $\nodes_0\ ({\G\filter{\leftM,\rightM,\killed}})$, the set of marked
  nodes of $\G$, including $\nnull$.\vspace{5mm}}\label{fig:sw}
\end{figure}

Fig.~\ref{fig:sw} presents the corresponding proof outline, most of
which is self-explanatory by standard Hoare logic. For clarity, we
elide the details about the stateful commands in lines~\ref{sw-ln05},
\ref{sw-ln15}, \ref{sw-ln19}, \ref{sw-ln24}, and~\ref{sw-ln28},
summarizing their effect through assertions. The elided parts are
shown elsewhere: for POP (line~\ref{sw-ln15}) in
Section~\ref{sec:pop}, for SWING (line~\ref{sw-ln19}), PUSH
(line~\ref{sw-ln24}), line~\ref{sw-ln05} and line~\ref{sw-ln28} in the
appendices. The proof of the remaining stateful command in
line~\ref{sw-ln11} is simple and elided altogether.

We detail the key non-trivial aspects of the proof outline: the
non-spatial reasoning steps from the precondition at
line~\ref{sw-ln01} to line~\ref{sw-ln02}, and from line~\ref{sw-ln31}
to the postcondition at line~\ref{sw-ln32}. In both cases, the graph
is at some point decomposed into disjoint parts (e.g., by
Lemma~\ref{filterdisj}), and the proof proceeds to analyze the
relationship between those parts. These steps illustrate another
example of contextual localization, here within a proof outline that
can be characterized as global, since all assertions in
Fig.~\ref{fig:sw} refer to the full graph $\G$, rather than its parts.

\paragraph*{\textbf{Precondition (lines~\ref{sw-ln01}-\ref{sw-ln02}).}}
Referring to Fig.~\ref{fig:sw}, line~\ref{sw-ln02} obtains from
line~\ref{sw-ln01} by unfolding the definition of $\inv$ from
Fig.~\ref{fig:loopinv} with $t = r$, $p = \nnull$, and instantiating
the existentially quantified stack $\St$ with the empty sequence
$\nil$. By setting $\St = \nil$, the unfolding derives proof
obligations: $\unique\ \nnull$, $\nnull = \llast\ (\nnull \join
\nil)$, $r \in \nodes_0\ \GG$, $\goodg\ \GG$,
$\nodes\ {\GG\filter{\leftM,\rightM}} = \emptyset$,
$\sconsecx{\GG}{\nil} = \erasure{\GG}$, $\gdiffx{\GG}{\nil}{t} =
\erasure{\GG}$, and $\nodes\ {\GG\filter{\notM}} \subseteq
\reachtwo{(\GG\filter{\notM})}{r}$.
The first four obligations are immediate, and properties
$\sconsecx{\GG}{\nil} = \erasure{\GG}$ and $\gdiffx{\GG}{\nil}{t} =
\erasure{\GG}$ follow by Lemma~\ref{sconsec-empty}.  The property
$\nodes\ {\GG\filter{\notM}} \subseteq
\reachtwo{(\GG\filter{\notM})}{r}$ is also easily shown because
$\nodes\ {\GG\filter{\notM}} \subseteq \nodes\ \GG$ and $\nodes\ \GG =
\reachtwo{(\GG\filter{\notM})}{r}$ by the assumption in
line~\ref{sw-ln01}.

To show the remaining $\nodes\ {\GG\filter{\leftM,\rightM}} =
\emptyset$, note that $\nodes\ {\GG\filter{\leftM,\rightM}} \subseteq
\nodes\ \GG$ and $\reachtwo{(\GG\filter{\notM})}{r} \subseteq
\nodes\ {\GG\filter{\notM}}$ because filtering and reachability select
a subset of graph's nodes. Thus, $\nodes\ {\GG\filter{\leftM,\rightM}}
\subseteq \nodes\ \GG = \reachtwo{(\GG\filter{\notM})}{r} \subseteq
\nodes\ {\GG\filter{\notM}}$. By Lemma~\ref{filterdisj},
$\GG\filter{\leftM,\rightM}$ and $\GG\filter{\notM}$ select
node-disjoint subgraphs of $\GG$. But
$\nodes\ {\GG\filter{\leftM,\rightM}}$ can be a subset of
$\nodes\ {\GG\filter{\notM}}$ with which it's disjoint, only if
$\nodes\ {\GG\filter{\leftM,\rightM}} = \emptyset$.

\paragraph*{\textbf{Postcondition (lines~\ref{sw-ln31}-\ref{sw-ln32}).}}
From $\invx{\GG}{\G}{t}{p}$ in line~\ref{sw-ln31}, it follows $p =
\llast\ (\nnull \cons \St)$ and $\unique\ (\nnull \cons \St)$. As
$p = \nnull$, it must be $\St = \nil$, i.e., the stack is
empty.
Then, by (\ref{invrestore}) in Fig.~\ref{fig:loopinv}, $\erasure{\GG}
= \gdiffx{\G}{\nil}{t}$, which in turn equals $\erasure\G$ by
Lemma~\ref{sconsec-empty}, thus proving $\erasure\GG=\erasure\G$ in
the postcondition.
By Lemma~\ref{filterdisj}, the remaining $\G = {\G\filter{\killed}}$
from the postcondition is equivalent to
$\nodes\ {\G\filter{\notM,\leftM,\rightM}} = \emptyset$. By
Lemma~\ref{filterdisj} again, and distributivity of $\nodes$, the
latter is further equivalent to
$\nodes\ {\G\filter{\notM}} \cup \nodes\ {\G\filter{\leftM,\rightM}}
= \emptyset$. By (\ref{invstack}) in Fig.~\ref{fig:loopinv},
$\nodes\ {\G\filter{\leftM,\rightM}} = \emptyset$, so it suffices to
show $\nodes\ {\G\filter{\notM}} = \emptyset$.
By (\ref{invreach}) of Fig.~\ref{fig:loopinv},
$\nodes\ {\G\filter{\notM}} \subseteq
\reachtwo{(\G\filter{\notM})}{t}$, i.e., unmarked nodes of $\G$ are
reachable from $t$ by unmarked paths.
But $t$ itself is marked
in line~\ref{sw-ln31},
so no node is reachable from $t$ by an unmarked path. Hence,
$\reachtwo{(\G\filter{\notM})}{t} = \emptyset$, and thus
$\nodes\ {\G\filter{\notM}} = \emptyset$.

\newcommand{\pp}{p_0}
\newcommand{\TT}{t_0}

\subsection{Proof of POP}\label{sec:pop}
The pre- and postcondition for POP derive from
lines~\ref{sw-ln14} and~\ref{sw-ln16} of Fig.~\ref{fig:sw}.
\begin{equation}\label{pre}
\begin{array}[c]{c} 
\spec{\exists \G\ldot \gseg\ \G \wedge \invx{\GG}{\G}{t}{p} \wedge t \in \setmarknull{\G} \wedge \Gm\ p = \rightM}\\
\textrm{POP} \\
\spec{\exists \G'\ldot\ \gseg\ \G' \wedge \invx{\GG}{\G'}{t}{p}}
\end{array}
\end{equation}
The precondition says that the heap implements a well-formed graph
$\G$, that satisfies the invariant, given $t$ and $p$. Additionally,
$t$ is marked or $\nnull$ and $p$ is marked $\rightM$, as POP is
invoked only if these properties are satisfied. The postcondition
asserts that the heap represents a new graph $\G'$ that satisfies the
invariant for the updated values of $t$ and $p$.  
The specification is given solely in terms of graphs to hide the
internal low-level reasoning about pointers, and expose only the more
abstract graph properties required in Fig.~\ref{fig:sw}.

The proof outline is in Fig.~\ref{fig:pop-outline-spatial} and divides
into two parts. The first part (lines~\ref{pop-ln1}-\ref{pop-ln9})
serves to show that the stateful commands implementing POP are safe to
execute, and result in a valid graph $\G'$.
The second part (lines~\ref{pop-ln9}-\ref{pop-ln10}) serves to show
that $\G'$ actually satisfies the Schorr-Waite invariants.  Along with
the analogous steps in the proofs of SWING and PUSH, this is the most
important, and most substantial part of the whole verification. We
discuss the two parts in more detail.

\newcommand{\pl}{p_l}
\newcommand{\pr}{p_r}

\makeatletter
\newcounter{codepop}
\renewcommand{\lineno}{\stepcounter{codepop}\textsc{\thecodepop}.\quad}
\renewcommand{\linelb}[1]{{\refstepcounter{codepop}\ltx@label{#1}\textsc{\thecodepop}.\quad}}
\makeatother
\begin{figure}[!t]
  \begin{align*}
\linelb{pop-ln1}    & \spec{\exists \G\ldot \gseg\ \G \wedge \invx{\GG}{\G}{t}{p}
    \wedge t \in \setmarknull{\G} \wedge \Gm\ p = \rightM}\\
\linelb{pop-ln2}  &\qquad \sspecopen{\gseg\ \G \wedge \TT = t \wedge \pp = p}\\
                  &\qquad \opensspec{\textcolor{teal}{\hbox{}\wedge \invx{\GG}{\G}{\TT}{\pp}
      \wedge \TT \in \setmarknull{\G} \wedge \G\ \pp = (\rightM, [p_l,p_r])}} \nonumber\\
\linelb{pop-ln3}  & \qquad\qquad \spec{\gseg\ (\pp \mapsto (\rightM, [p_l,p_r]) \join {\G{\setminus}\pp}) 
\wedge \TT = t \wedge \pp = p}\\
\linelb{pop-ln4} & \qquad\qquad \spec{(\pp \Mapsto \rightM, \pl, \pr \wedge \TT = t \wedge \pp = p)\
       \textcolor{teal}{*\ \gseg\ {\G{\setminus}\pp}}}\\
\linelb{pop-ln5}   &\qquad\qquad\qquad\spec{\pp \Mapsto \rightM,\pl,\pr \wedge \TT = t \wedge \pp = p}\\
    &\qquad\qquad\qquad\ \tmp \assign \deref{p.r};\ p.r \mutate t;\ p.m \mutate \killed;\
  t \assign p;\ p \assign \tmp; \qquad \textrm{//\ POP} \nonumber\\
\linelb{pop-ln6}   &\qquad\qquad\qquad\spec{\pp \Mapsto \killed,\pl,\TT \wedge t = \pp \wedge p = \pr}\\
\linelb{pop-ln7}    &\qquad\qquad\spec{(\pp \Mapsto \killed,\pl,\TT \wedge t = \pp \wedge p = \pr)\
        \textcolor{teal}{*\ \gseg\ {\G{\setminus}\pp}}}\\ 
\linelb{pop-ln8}    &\qquad\qquad\spec{\gseg\ (\pp \mapsto (\killed,[\pl,\TT]) \join {\G{\setminus}\pp}) \wedge
      t = \pp \wedge p = \pr}\\
\linelb{pop-ln9}  &\qquad \sspecopen{\gseg\ (\pp \mapsto (\killed,[\pl,\TT]) \join {\G{\setminus}\pp}) \wedge
      t = \pp \wedge p = \pr}\\
    &\qquad \opensspec{\textcolor{teal}{\hbox{}\wedge \invx{\GG}{\G}{\TT}{\pp}
    \wedge \TT \in \setmarknull{\G} \wedge \G\ \pp = (\rightM, [p_l,p_r])}}\nonumber\\
\linelb{pop-ln10}  &\spec{\exists \G'\ldot \gseg\ \G' \wedge \invx{\GG}{\G'}{t}{p}}
\end{align*}
\caption{Proof outline for POP. The property
  $\textcolor{teal}{\invx{\GG}{\G}{\TT}{\pp} \wedge \TT \in
    \setmarknull{\G} \wedge \G\ \pp = (\rightM, [p_l,p_r])}$
  propagates from line~\ref{pop-ln2} to line~\ref{pop-ln9} because it
  doesn't describe the heap or variables mutated by the
  program.}\label{fig:pop-outline-spatial}
\end{figure}

\paragraph*{\textbf{POP produces valid graph (lines~\ref{pop-ln1}-\ref{pop-ln9}).}}
Transitioning from line~\ref{pop-ln1} to line~\ref{pop-ln2} involves
eliminating the existential quantifier and saving the current values
of $t$, $p$, and the left and right child of $p$ into $\TT$, $\pp$,
$p_l$ and $p_r$, respectively. Saving these values prepares for
line~\ref{pop-ln3} to omit the last three conjuncts of
line~\ref{pop-ln2}, which will be re-attached in
line~\ref{pop-ln9}. The move is valid, because the conjuncts are
unaffected by execution, as they don't involve the heap or the variables
mutated by POP. 
Proceeding with the proof outline, line~\ref{pop-ln3} expands $\G$
around $\pp$ (=$p$), using equation~(\ref{expand1}), so that
distributivity of $\gseg$~(\ref{dist}) applies in line~\ref{pop-ln4}
to frame away $\gseg\ {\G{\setminus}\pp}$.
This is similar to how the pointer $j$ in Section~\ref{sec:overview}
was separated from the rest of the heap, in order to verify the
command that mutates it.
The derivation of line~\ref{pop-ln6} is elided because it's standard,
involving several applications of framing and the inference rules for
the stateful primitives. It suffices to say that the line reflects the
mutations of the pointer $\pp$, and the assignment of new values to
$t$ and $p$, in accord with the illustration of POP in
Fig.~\ref{fig:markcycle}. In particular, the right edge of $\pp$ is
redirected towards $\TT$, the node is marked $\killed$, and $t$ and
$p$ are set to $\pp$ and $\pr$, respectively.
The remainder of the proof outline up to line~\ref{pop-ln9} restores
the framed graph and the non-spatial conjuncts omitted in
line~\ref{pop-ln3}.

\begin{figure}[t!]
\rowcolors{2}{gray!7}{white}
{\allowdisplaybreaks 
\[{\displaystyle 
\begin{array}[c]{rl>{\quad}l<{}}
\rowcolor{gray!7}
  &\hspace{-5mm}(\mbox{\ref{invstack}})\ nodes\ {\G'\filter{\leftM,\rightM}} = &\text{Def. of $\G'$} \\
  &=\nodes\ {(\pp \mapsto (\killed, [\pl, \TT]) \join \G{\setminus}{\pp})\filter{\leftM,\rightM}} &\text{Lem. \ref{nodes-dist} \& \ref{filterc-dist} (distrib.)} \\
  &= \nodes\ {(\pp \mapsto (\killed, [\pl, \TT]))\filter{\leftM,\rightM}}\ \dotcup\ \nodes\ {(\G{\setminus}{\pp})\filter{\leftM,\rightM}} & \text{Def. of \textit{filter} (\ref{filterc-def})}\\ 
  &= \nodes\ {(\G{\setminus}{\pp})\filter{\leftM,\rightM}} & \text{Assump.~(\ref{pop-pre2}.\ref{invstack})}\\
            &= \St & \\[1.5mm]
  &\hspace{-5mm}(\mbox{\ref{invinvert}})\ \sconsecx{\G'}{\St} = & \text{Def. of $\G'$} \\
            &= \sconsecx{(\pp \mapsto (\killed, [\pl, \TT]) \join \G{\setminus}{\pp})}{\St} & \text{Lem. \ref{map-dist} (distrib.)} \\
            &= \sconsecx{(\pp \mapsto (\killed,[\pl,\TT]))}{\St} \join \sconsecx{\G{\setminus}{\pp}}{\St} \hspace*{5mm}& \text{Def. of $\sconsec$ (\ref{invert-def})} \\
            &= \pp \mapsto [\pl,\TT] \join \sconsecx{\G{\setminus}{\pp}}{\St} & \text{Lem.~\ref{sconsec-p}} \\
  &= \pp \mapsto [\pl,\TT] \join \sconsecx{\G{\setminus}{\pp}}{(\St \join \pp)} & \text{Assump.~(\ref{pop-pre2}.\ref{invinvert})} \\
            &= \pp \mapsto [\pl,\TT] \join \erasure{\G{\setminus}{\pp}} & \text{Def. of \textit{erasure} (\ref{erasure-def})} \\
            &= \erasure{\pp \mapsto (\killed,[\pl,\TT])} \join \erasure{\G{\setminus}{\pp}} & \text{Lem. \ref{erasure-dist} (distrib.)} \\
            &= \erasure{\pp \mapsto (\killed,[\pl,\TT]) \join \G{\setminus}{\pp}} & \text{Def. of $\G'$} \\
            &= \erasure{\G'} & \\[1.5mm]
  &\hspace{-5mm}(\mbox{\ref{invrestore}})\ \gdiffx{\G'}{\St}{\pp}= & \text{Def. of $\G'$}\\
            &= \gdiffx{(\pp \mapsto (\killed, [\pl, \TT]) \join \G{\setminus}{\pp})}{\St}{\pp} &\text{Lem. \ref{map-dist} (distrib.)}\\
            &= \gdiffx{(\pp \mapsto (\killed, [\pl, \TT]))}{\St}{\pp} \join \gdiffx{\G{\setminus}{\pp}}{\St}{\pp}&\text{Def. of $\gdiff$ (\ref{restore-def})}\\
            &= \pp \mapsto [\pl, \TT] \join \gdiffx{\G{\setminus}{\pp}}{\St}{\pp}&\text{Lem.~\ref{sconsec-p}}\\
            &= \pp \mapsto [\pl, \TT] \join \gdiffx{\G{\setminus}{\pp}}{(\St \join \pp)}{\TT}&\text{Def. of $\gdiff$ (\ref{restore-def})}\\
  &= \gdiffx{(\pp \mapsto (\rightM, [\pl, \pr]))}{(\St \join \pp)}{\TT} &\\
  &\hspace{4mm} \join\ \gdiffx{\G{\setminus}{\pp}}{(\St \join \pp)}{\TT}&\text{Lem. \ref{map-dist} (distrib.)}\\ 
            &= \gdiffx{(\pp \mapsto (\rightM, [\pl, \pr]) \join \G{\setminus}{\pp})}{(\St \join \pp)}{\TT} &\text{Assump.~(\ref{pop-pre2}.\ref{invrestore})}\\
            &= \erasure{\GG} &\\[1.5mm]
  &\hspace{-5mm}(\mbox{\ref{invreach}})\ \nodes\ {\G'\filter{\notM}} = & \text{Def. of \textit{filter}, $\nodes$, $\G$ \& $\G'$}\\ 
  &= \nodes\ {\G\filter{\notM}} & \text{Assump.~(\ref{pop-pre2}.\ref{invreach})}\\
  &\subseteq {\textstyle \bigcup\limits_{\St\cdot \pp}} (\reachone{(\G\filter{\notM})}\circ\G_r) \cup \reachtwo{(\G\filter{\notM})}{\TT} & \text{$\TT \notin \nodes\ {\G\filter{\notM}}$} \\
  &= {\textstyle\bigcup\limits_{\St\cdot \pp}} (\reachone{(\G\filter{\notM})}\circ\G_r)& \text{Comm.\&Assoc. of $\cup$}\\
  &= {\textstyle\bigcup\limits_{\St}}\ (\reachone{(\G\filter{\notM})}\circ\G_r) \cup \reachtwo{(\G\filter{\notM})}{(\G_r\ \pp)} & \text{$\G_r\ \pp \notin \nodes\ {\G\filter{\notM}}$}\\
  &= {\textstyle\bigcup\limits_{\St}}\ (\reachone{(\G\filter{\notM})}\circ\G_r)&\text{($\G_r = \G_r'$ on $\St$) \& (${\G\filter{\notM}} = {\G'\filter{\notM}}$)} \\
  &= {\textstyle\bigcup\limits_{\St}}\ (\reachone{(\G'\filter{\notM})}\circ\G'_r)& \text{$\pp \notin \nodes\ {\G'\filter{\notM}}$}\\ 
  &= {\textstyle\bigcup\limits_{\St}}\ (\reachone{(\G'\filter{\notM})}\circ\G'_r) \cup \reachtwo{(\G'\filter{\notM})}{\pp} & 
\end{array} 
}\]}
\caption{POP preserves invariants (\ref{invstack})-(\ref{invreach})
for $\G = \pp
 \mapsto (\rightM, [\pl, \pr]) \join \G{\setminus}{\pp}$ and $\G' =
 \pp \mapsto (\killed, [\pl, \TT]) \join \G{\setminus}{\pp}$. 
 \vspace{-2mm}
}\label{fig:pop-outline-non-spatial}
\end{figure}

\paragraph*{\textbf{Invariant preservation (lines~\ref{pop-ln9}-\ref{pop-ln10}).}}
This part of the proof is fully non-spatial and relies almost entirely
on the distributivity of various morphisms, along with a few general
graph lemmas from Section~\ref{sec:partialgraphs}. It follows the same
pattern of contextual locality as Lemma~\ref{lemma:closedlocal},
exploiting that the graphs in lines~\ref{pop-ln9} and~\ref{pop-ln10}
share the subgraph $\G{\setminus}\pp$.

As the first step, we reformulate the problem as the following
implication, whose premise restates the assertions about $\G$ from
line~\ref{pop-ln9}, while the conclusion specifies the ending graph
$\G'$.
\begin{align}
& \G = \pp \mapsto (\rightM, [\pl, \pr]) \join \G{\setminus}{\pp} \wedge \hbox{} \label{pop-pre1}\\
& \invpx{\GG}{\G}{(\St \join \pp)}{\TT}{\pp} \wedge \hbox{} \label{pop-pre2}\\
& \TT \in \setmarknull{\G} \implies \hbox{}\label{pop-pre3}\\
\exists \G'\ldot & \G' = \pp \mapsto (\killed,[\pl,\TT]) \join \G{\setminus}{\pp} \wedge \hbox{} \label{pop-post1}\\
& \invpx{\GG}{\G'}{\St}{\pp}{\pr}\label{pop-post2}
\end{align}

To see that the premise follows from line~\ref{pop-ln9}, note that
(\ref{pop-pre1}) expands $\G$ around $\pp$, using $\G\ \pp = (\rightM,
[\pl, \pr])$ from line~\ref{pop-ln9}; that (\ref{pop-pre2}) holds
because $\invx{\GG}{\G}{\TT}{\pp}$ in line~\ref{pop-ln9} implies that
$\pp$ is the top of stack of $\G$, which thus has the form $\St \join
\pp$ for some $\alpha$; and that (\ref{pop-pre3}) is an explicit
conjunct in line~\ref{pop-ln9}.

In the conclusion of the implication, (\ref{pop-post1}) reflects that
the $\pp$ is now marked $\killed$, and that its right edge is restored
towards $\TT$. On the other hand, (\ref{pop-post2}) indicates that the
new stack is $\St$ (node $\pp$ having been popped from the prior stack
$\St \join \pp$), and the new tip and stack's top are $\pp$ and $\pr$
respectively. These changes correspond to the illustration of POP in
Fig.~\ref{fig:mainops}, if one takes $\pp$, $\TT$ and $\pr$ as the
initial values of $p$, $t$, and $p$'s right child.  It's also readily
apparent that (\ref{pop-post1}) and (\ref{pop-post2}) imply the
postcondition of POP in line~\ref{pop-ln10}.

The second step of the proof proceeds to establish the invariant
$\invpx{\GG}{\G'}{\St}{\pp}{\pr}$ from~(\ref{pop-post2}) out of
$\invpx{\GG}{\G}{(\St \join \pp)}{\TT}{\pp}$ in~(\ref{pop-pre2}), and
$\TT \in \setmarknull{\G}$ from~(\ref{pop-pre3}). As $\invp$ is
defined in Fig.~\ref{fig:loopinv} in terms of conjuncts
(\ref{invbook})-(\ref{invreach}), each conjunct in (\ref{pop-post2})
is proved starting from the corresponding conjunct in
(\ref{pop-pre2}).
The subcase (\ref{invclosed}) of this proof is explicitly
Lemma~\ref{lemma:closedlocal}, and
Fig.~\ref{fig:pop-outline-non-spatial} presents the subcases
(\ref{invstack})-(\ref{invreach}) in the self-explanatory equational
style. The proofs follow the pattern of contextual localization from
Lemma~\ref{lemma:closedlocal}. They begin by decomposing $\G$ and
$\G'$ as in~(\ref{pop-pre1}) and~(\ref{pop-post1}), isolating the node
$\pp$ from the shared subcomponent $\G{\setminus}\pp$. Each subcase
then applies the appropriate morphism to distribute over the
decomposition, ultimately reducing the goal to a property of
$\G{\setminus}\pp$ that is already assumed. 

The remaining subcase (\ref{invbook}) doesn't have an equational
flavor so we give it here explicitly. That case requires showing that
(\ref{pop-pre1}-\ref{pop-pre3}) imply $\unique\ (\nnull \join \St)$,
$\pr = \llast\ (\nnull \join \St)$ and $\pp \in \nodes_0\ \G'$. The
only non-trivial property is $\pr = \llast\ (\nnull \join \St)$, i.e.,
that $\pr$ is the top of the stack $\St$. This holds because
(\ref{pop-pre2}) gives $\sconsecx{\G}{(\St \join \pp)} =
\erasure{\G}$, i.e., that the stack $\St \join \pp$ is a path in
$\G$. By definition of $\sconsec$ (\ref{invert-def}), $\pprev\ (\nnull
\join \St \join \pp)\ \pp = \G_r\ \pp = \pr$, and thus $\pr =
\llast\ (\nnull \join \St)$.

\subsection{Proof for General Graphs}\label{sec:global}
We can now establish the following general specification of
Schorr-Waite, which doesn't assume that $\G$ is connected from $r$.
\begin{equation*}
\begin{array}{c}
\spec{\gseg\ \GG \wedge \goodg\ \GG \wedge r \in \nodes\ \GG \wedge \allu{\GG}}\\
\textit{Schorr-Waite}\,(r)\\
\spec{\exists \G \ldot \gseg\ \G \wedge \erasure{\GG} = \erasure{\G}
      \wedge \G = {\G\filter{\killed}} \join {\G\filter{\notM}} \wedge \nodes\ {\G\filter{\killed}} = \reachtwo{(\GG\filter{\notM})}{r}}
\end{array}
\end{equation*}
As in specification~(\ref{sw-connect}), the precondition says that the
input graph $\GG$ is well-formed ($\gseg\ \GG$), closed, contains the
node $r$, and is unmarked ($\allu\GG$), but elides the conjunct about
connectedness. The postcondition posits an ending graph $\G$ which,
aside from node marking, equals the input graph ($\erasure{\GG} =
\erasure{\G}$), and splits into fully marked and unmarked parts ($\G =
        {\G\filter{\killed}} \join {\G\filter{\notM}}$), with the
        fully-marked part corresponding to the nodes initially
        reachable from $r$ ($\nodes\ {\G\filter{\killed}} =
        \reachtwo{(\GG\filter{\notM})}{r}$).

\makeatletter
\newcounter{codegsw}
\renewcommand{\lineno}{\stepcounter{codegsw}\textsc{\thecodegsw}.\quad}
\renewcommand{\linelb}[1]{{\refstepcounter{codegsw}\ltx@label{#1}\textsc{\thecodegsw}.\quad}}
\makeatother
\begin{figure}[!t] 
  \begin{align*}
\linelb{gl-ln1}  &\spec{\gseg\ \GG \wedge \goodg\ \GG \wedge r \in \nodes\ \GG \wedge \allu{\GG}}\\
\linelb{gl-ln2}    &\sspecopen{\exists \G_1 \G_2\ldot \gseg\ (\G_1 \join \G_2) \wedge \goodg\ \G_1 \wedge
      r \in \nodes\ \G_1 \wedge \nodes\ \G_1 = \reachtwo{(\G_1\filter{\notM})}{r}}\\
       &\opensspec{\hbox{} \wedge \nodes\ \G_1 = \reachtwo{((\G_1\join\G_2)\filter{\notM})}{r} \wedge \G_2 = {\G_2\filter{\notM}} \wedge \GG = \G_1 \join \G_2 }\nonumber\\
\linelb{gl-ln3}  &\quad\sspecopen{(\gseg\ \G_1 \wedge \goodg\ \G_1  \wedge r \in \nodes\ \G_1 \wedge \nodes\ \G_1 = \reachtwo{(\G_1\filter{\notM})}{r})}\\
    &\quad\opensspec{\textcolor{teal}{\hbox{} * (\gseg\ \G_2 \wedge \nodes\ \G_1 = \reachtwo{((\G_1\join\G_2)\filter{\notM})}{r}
      \wedge \allu{\G_2} \wedge \GG = \G_1 \join \G_2)}}\nonumber\\
\lineno    &\qquad\spec{\gseg\ \G_1 \wedge \goodg\ \G_1 \wedge r \in \nodes\ \G_1 \wedge \nodes\ \G_1 = \reachtwo{(\G_1\filter{\notM})}{r}}\\
    &\qquad\textit{Schorr-Waite}\,(r)\nonumber\\
\lineno &\qquad\spec{\exists \G_1'\ldot \gseg\ \G_1' \wedge \erasure{\G_1}\ = \erasure{\G_1'} \wedge \allm{\G_1'}}\\
\linelb{gl-ln6}  &\quad \sspecopen{(\exists \G_1'\ldot \gseg\ \G_1' \wedge \erasure{\G_1}\ = \erasure{\G_1'} \wedge \allm{\G_1'})}\\
    &\quad \opensspec{\textcolor{teal}{\hbox{} * (\gseg\ \G_2 \wedge \nodes\ \G_1 = \reachtwo{((\G_1\join\G_2)\filter{\notM})}{r} 
      \wedge \allu{\G_2} \wedge \GG = \G_1 \join \G_2)}}\nonumber\\
\linelb{gl-ln7}  &\quad \sspecopen{\exists \G_1' \ldot \gseg (\G_1' \join \G_2) \wedge \erasure{\GG}\ = \erasure{\G_1'\join \G_2}}\\
    &\quad \opensspec{\hbox{} \wedge \G_1' \join \G_2 = {(\G_1'\join\G_2)\filter{\killed}} \join {(\G_1'\join\G_2)\filter\notM} \wedge \nodes\ {(\G_1' \join \G_2)\filter{\killed}} = \reachtwo{(\GG\filter{\notM})}{r}}\nonumber\\
\linelb{gl-ln8}    &\spec{\exists \G \ldot \gseg\ \G \wedge \erasure{\GG} = \erasure{\G}
      \wedge \G = {\G\filter{\killed}} \join {\G\filter{\notM}} \wedge \nodes\ {\G\filter{\killed}} = \reachtwo{(\GG\filter{\notM})}{r}}
  \end{align*}
\caption{Proof outline for the general graph specification of Schorr-Waite.}\label{fig:gsw}
\end{figure}

The proof is in Fig.~\ref{fig:gsw}, and is obtained generically, by
framing~(\ref{sw-connect}) without reanalyzing the code. It utilizes
distributivity of $\gseg$ to set up the framing, and morphisms to
propagate the framed information from smaller to larger graph. This
propagation, carried out in the last step of the proof (from
line~\ref{gl-ln7} to line~\ref{gl-ln8}), is a form of contextual
localization in reverse, where instead of focusing inward, information
flows outward through the structure.

In more detail, $\GG$ is first split disjointly into $\G_1 =
{\GG\filter{(\reachtwo{\GG}{r})}}$ (part connected from $r$), and its
complement $\G_2$, so that $\GG = \G_1\join \G_2$. By definition then,
$\nodes\ \G_1 = \reachtwo{(\G_1 \join \G_2)}{r}$, so that by
Lemma~\ref{closed-sub2}, $\goodg\ \G_1$, $r \in \nodes\ \G_1$ and
$\nodes\ \G_1 = \reachtwo{\G_1}{r}$.
Because $\G_1$ and $\G_2$ are subgraphs of $\GG$, they are also
unmarked (${\G_1\filter{\notM}} = \G_1$ and ${\G_2\filter{\notM}} =
\G_2$), thus obtaining line~\ref{gl-ln2} in Fig.~\ref{fig:gsw}.
Distributivity of $\gseg$ (\ref{dist}) then derives line~\ref{gl-ln3},
which is a separating conjunction of the precondition
from~(\ref{sw-connect}) with the predicate that collects all the
conjuncts from line~\ref{gl-ln2} that refer to graph $\G_2$.  By
framing the specification~(\ref{sw-connect}), we can thus derive a
postcondition for Schorr-Waite in line~\ref{gl-ln6} that is a
separating conjunction of the conjuncts about $\G_2$, with the
postcondition from~(\ref{sw-connect}) that asserts the existence of a
graph $\G_1'$ such that $\erasure{\G_1} = \erasure{\G_1'}$ and $\G_1'
= {\G_1'}\filter{\killed}$.

Proceeding, line~\ref{gl-ln7} joins the graphs $\G_1'$ and $\G_2$
together as follows.
First, distributivity of $\gseg$ and erasure obtains $\gseg\ (\G_1'
\join \G_2)$, and $\erasure{\GG} = \erasure{\G_1 \join \G_2} = \erasure{\G_1' \join
  \G_2}$.
We also get ${\G_1'\filter{\notM}} =
{{\G_1'\filter{\killed}}\filter{\notM}} = \emptymap$ and
${\G_2\filter{\killed}} = {{\G_2\filter{\notM}}\filter{\killed}} =
\emptymap$ by Lemma~\ref{filter-disj}.
From here, 
$\G_1' \join
\G_2 = {\G_1'\filter\killed} \join {\G_2\filter\notM} =
{(\G_1'\join\G_2)\filter{\killed}} \join {(\G_1'\join\G_2)\filter\notM}$
and also 
$\nodes\ {(\G_1' \join \G_2)\filter{\killed}} = 
\nodes\ {\G_1'\filter{\killed}} = 
\nodes\ {\G_1'} = 
\nodes\ \erasure{\G_1'} = 
\nodes\ \erasure{\G_1} = 
\nodes\ {\G_1} = 
\reachtwo{(\GG\filter{\notM})}{r}$.
Finally, line~\ref{gl-ln8} abstracts over $\G = \G_1' \join \G_2$.

\section{Sketch of Union-Find Verification}\label{sec:unionfind}
\newcommand{\summit}{\textit{summit}}
\newcommand{\summits}{\textit{summits}}
\newcommand{\ccenter}[2]{\summit\ #1\ #2}
\newcommand{\ccentersg}[2]{\summits\ #1\ #2}
\newcommand{\ccenters}[1]{\summits\ #1}
\newcommand{\selfloops}{\textit{loops}}
\newcommand{\selfloopss}[1]{\selfloops\ #1}
\newcommand{\setx}{\textit{set}}
\newcommand{\set}[2]{\setx\ #1\ #2}
\newcommand{\unfoldset}[2]{\exists \G \ldot \gseg_1\ \G \wedge \ccenters{\G} = \selfloopss{\G} = \{#2\} \wedge \nodes\ \G = #1}

The union-find data structure manages a collection of disjoint sets,
supporting efficient set merging and membership queries. Each set is
represented in memory as an inverted tree, each element (a node in the
tree) points to its unique parent, and the root of the tree---which
points to itself---serves as the set’s representative. An element's
representative is obtained by following parent pointers upward. The
union operation merges two sets by making the root of one point to the
root of the other set, thereby unifying their representatives
(Fig.~\ref{fig:uftree}).

Because each node has a unique parent, union-find operates over
\emph{unary graphs}, with a $\gseg$ predicate analogous to that for
binary graphs from Section~\ref{sec:partialgraphs}.
\begin{align*}
  \gseg_1\, \emptymap\ &\eqdef \emp\\
  \gseg_1\, (x \mapsto y \join \G)\ &\eqdef x \Mapsto y * \gseg_1\, \G
\end{align*}
The main graph primitives in the specification for union-find are the following.
\[
\begin{array}{rcl}
\ccenter{\G}{x} & \eqdef &
  \left\{\begin{array}{ll}
    \bigcup\limits_{\Ga\,x} \ccenter{(\G{\setminus{x}})} & \mbox{if $x \in \nodes\ \G$} \\
    \{x\} & \mbox{otherwise}
  \end{array}\right.\\[5mm]
\ccenters{\G} & \eqdef & \bigcup\limits_{\nodes\ \G} \ccenter{\G}\\[5mm]
\selfloopss{\G} & \eqdef & \{x \mid x \in \Ga\ x\} 
\end{array}
\]
The set $\ccenter{\G}{x}$ contains the nodes of $\G$ that \emph{mark
an end (i.e., summit)} of a path from $x$. More precisely, $z \in
\ccenter{\G}{x}$ if there exists a path from $x$ to some node $y$ in
$\G$ such that $z$ is a child of $y$, and either $z \notin \nodes\ \G$
(i.e., the edge from $y$ to $z$ is dangling) or $z$ is already in the
path from $x$ to $y$ (thus starting a cycle).
The set $\ccenters{\G}$ collects all such ending nodes.
The set $\selfloopss\G$ collects the nodes of $\G$ that include
themselves in their adjacency list; that is, nodes forming cycles of
size $1$. 

\begin{figure}[t]
  \centering
  \includegraphics[trim=0mm 0 0 0,clip,width=0.9\textwidth]{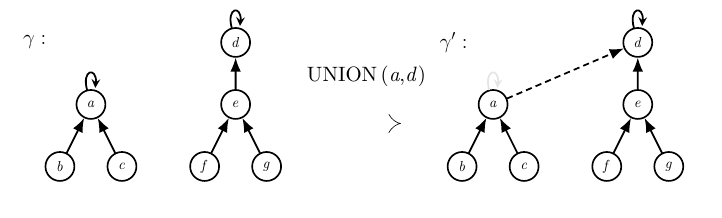} 
  \vspace{-5mm}
  \caption{Two disjoint sets encoded as inverted trees being merged by
    the operation of $\textrm{UNION}$. Initially, $a$ and $d$ represent the sets
    $\{a, b, c\}$ and $\{d, e, f, g\}$, respectively. After the
    union, $d$ represents all the elements.}
\label{fig:uftree}
\end{figure}

Fig.~\ref{fig:uftree} illustrates the relevance of the above
primitives for union-find. Specifically, $\ccenter$ computes the
representative of a given node (e.g, $\ccenter{\G}{c} = \{a\}$), as
the representatives in union-find are precisely the path-ending
nodes. Analogously, $\ccenters$ collects all the representatives
(e.g., $\ccenters{\G} = \{a, d\}$ and $\ccenters{\G'} = \{d\}$).
Importantly, $\ccenter$ and $\ccenters$ are somewhat more general
still, as they may return nodes that are outside of the graph
(dangling edges), or nodes that start cycles. In particular, if
$\ccenters{\G}$ only contains nodes in $\G$, then $\G$ is
closed. Furthermore, if $\ccenters{\G}$ only contains nodes that form
\emph{trivial} cycles (i.e., of size $1$, from $\selfloopss{\G}$),
then $\G$ is acyclic (modulo representatives), and thus an inverted
tree.

We thus use the above primitives not only to compute the
representatives, but also to define when a graph is an inverted
tree. We do so via the following predicate $\set{S}{x}$ which holds of
a heap that contains the layout of a \emph{single set} $S$ in a
union-find collection, whose representative is $x$.
\[
 \set{S}{x} \eqdef \unfoldset{S}{x}
\]
Indeed, the definition requires the existence of a unary graph $\G$
whose nodes are exactly $S$, such that $\G$ has $x$ as the unique
summit (i.e., $\G$ is closed), and $x$ forms a trivial cycle (i.e.,
$\G$ is inverted tree). For example, in the graph $\G$ in
Fig.~\ref{fig:uftree} we have $\set{\{a, b, c\}} a$ and $\set{\{d, e,
  f, g\}} d$.

\newcommand{\res}{\textit{result}}

The following Hoare triples specify the methods of union-find, where
the method's postcondition denotes the method's return value by the
dedicated variable $\res$.\footnote{This requires extending
the heretofore used standard logic of O'Hearn, Reynolds and Yang with
value-returning methods, but Hoare Type Theory admits such an
extension.}
\[
\begin{array}{rcl}
\spec{emp} & \textrm{NEW} & \spec{\set{\{\res\}}{\res}}\\
\spec{\set{S}{y} \wedge x \in S} & \textrm{FIND}\,(x) & \spec{\set{S}{y} \wedge \res = y}\\
\spec{\set{S_1}{x_1} * \set{S_2}{x_2}} & \textrm{UNION}\,(x_1, x_2) & \spec{\set{(S_1\,\dotcup\,S_2)}{\res} \wedge \res \in \{x_1, x_2\}}
\end{array}
\]
\textrm{NEW} starts from an empty heap and allocates a
node that forms a singleton set and serves as its own
representative. $\textrm{FIND}$ takes a node $x$ known to belong to
some set $S$, and returns the representative $y$ of
$S$. $\textrm{UNION}$ assumes that $x_1$ and $x_2$ are
representatives of disjoint sets and joins the sets, returning one of
$x_1$ or $x_2$ as the new representative (Fig.~\ref{fig:uftree}).
Notably, our specifications refer only to the disjoint sets each
operation manipulates, following the \emph{small footprint}
style. While this is the norm in separation logic generally, we're
unaware of prior union-find verifications that adopt it.

The implementation and the proof outlines for the methods are in the
appendix. Here, we only discuss how morphisms and distributivity help
with the verification of union-find.
The main challenge in this verification is establishing, in various
subproofs, the inclusion $\ccenters{\G} \subseteq \selfloopss{\G}$,
which is one side of the defining equation $\ccenters{\G} =
\selfloopss{\G}$ from the $\setx$ predicate. This closely resembles
the definition of $\goodg$ in Section~\ref{sec:partialgraphs} as
$\adj\ \G \subseteq \nodes_0\ \G$, where the fact that $\adj$ and
$\nodes$ are morphisms enabled contextual localization via
Lemma~\ref{lemma:closedlocal}.

In the case of union-find, we apply the same principle of localization
by leveraging distribution properties of $\selfloops$ and $\summits$. The former is a standard PCM morphism.
The latter, however, exhibits more nuanced
behavior. First, it distributes only under the specific condition that
all cycles in the graph are trivial---a property satisfied by
union-find graphs. Second, and more unusually, its distribution
doesn't follow the standard additive form of morphisms, but
\emph{subtracts} nodes from the opposing component to avoid
interference. 
\begin{align*}
  \summits\ (\G_1 \join \G_2) = (\summits\ \G_1){\setminus{\nodes\
    \G_2}} \cup (\summits\ \G_2){\setminus{\nodes\ \G_1}} 
\end{align*}
The intuition is that a summit in a subgraph $\G_1$ is a path-ending
node. If the path ends with a dangling edge into $\G_2$, the summit
ceases to be path-ending in the composition, where the path
continues. The equation accounts for this by explicitly removing the
overlap through set subtraction.

While the full theory of $\selfloops$ and $\summits$ is more
extensive, and holds for general (not just unary) graphs, the
distribution principle above suffices to draw a high-level analogy
with Schorr-Waite, and illustrate how contextual localization
applies equally well to union-find.

\section{Related work}\label{sec:related}

\paragraph*{\textbf{Proofs of Schorr-Waite in separation logic.}}
The starting point of our paper has been Yang's
proof~\cite{Yang2001AnEO,yang2001}, which is an early work on
separation logic generally, and the first work on graphs in separation
logic specifically. The distinction with our proof is that Yang
doesn't use mathematical graphs as an explicit argument to the $\gseg$
predicate, but rather relies on non-spatial proxy abstractions, such
as the spanning tree of the graph, and various subsets of nodes.
These proxy abstractions aren't independent, and to keep them
synchronized with each other and with the graph laid out in the heap,
the proof must frequently switch between non-spatial and spatial
reasoning, at the cost of significant formal overhead. In contrast, we
avoid the overhead by keeping most of the reasoning about graphs at
the non-spatial level.

Considering his first proof too complex, Yang tackled Schorr-Waite
again using \emph{relational} separation
logic~\cite{Yang2007RelationalSL}. Relational logic establishes a
contextual refinement between two programs; in the case of
Schorr-Waite, a refinement with the obvious depth-first-search (DFS)
implementation of graph marking.
While the resulting proof achieved a conceptual simplification over
the original non-relational proof, this is somewhat counter-intuitive,
as refinement between two programs is \emph{generally} a much stronger
property---and thus more demanding to prove---than establishing a pre-
and postcondition for a program.
Our result confirms this intuition. In comparison, an optimization of
Yang's relational proof is given by \citet{cre+kun:sefm11}, who show
that Schorr-Waite is contextually equivalent to DFS, and mechanize the
proof in around 3000 lines in Coq.

\paragraph*{\textbf{Non-separation proofs of Schorr-Waite.}}
Being a standard verification benchmark for graph algorithms,
Schorr-Waite has been verified numerous times, using a variety of
different approaches. These ranged from automated
ones~\cite{Suzuki76,log+rep+sag:sas06,roe:icalp77,leino:lpar10}, to
studies of the algorithm's mathematical
properties~\cite{Bubel2007,meh+nip:cade03,%
  Topor1979TheCO,Griffiths1979,Morris1982,%
  der:ipl80,Gries1979,bornat00,abr:fme03,hub+mar:sefm05}, to proofs
based on program
transformation~\cite{duf:icfem14,pre+bac:afp10,Gerhart1979,ward96,%
  gio+str+mat+pan:lopstr10,pre+bac:fac12,bro+pep:toplas}.
The important difference from this work is that we explicitly wanted
to support and utilize framing in our proof, as framing is the key
feature of separation logic.

\paragraph*{\textbf{Proofs of union-find.}}
The union-find structure~\cite{gal+fis:cacm64} is also a well-studied
graph benchmarks, with numerous correctness proofs, including in
separation logic. For instance, \citet{con+fil:ml07} verify a
persistent version in Coq that uses two functional arrays, while
\citet{lam+mei:2012} verify an implementation in
Imperative/HOL. \citet{neelk:phd} provides a non-mechanized proof, and
\citet{cha+pot:jar17} give a Coq proof with a complexity analysis,
both targeting a heap-allocated implementation, as we
do. \citet{wan+qin+moh+hob:oopsla19}, additionally verifies an
array-based variant in Coq.
Our approach differs in two key ways: we use partial graphs rather
than traditional graph theory in order to leverage PCM morphisms; we
also adopt small-footprint specifications that describe only the
modified disjoint subsets, as opposed to the whole structure.

\paragraph*{\textbf{Graphs in separation logic, without PCMs.}}
In the search for a $\gseg$ predicate that enables decomposition,
~\citet{bor+cal+ohe:space04} define a partial graph as a
recursive tree-shaped term, encoding the following strategy. Given
some default traversal order, the first time a node is encountered,
it's explicitly recorded in the term. Every later occurrence is
recorded as a reference (i.e., pointer) to the first one. This enables
that closed graphs can be encoded in a way that facilitates
decomposition, but is dependent on the traversal order, which is
problematic, as it prevents developing general libraries of lemmas
about graphs. We used the PCM of partial graphs to remove the
restriction to closed graphs, which also removes the traversal-order
dependence.

\citet{wang:phd20} and~\citet{wan+qin+moh+hob:oopsla19}
parametrize their $\gseg$ predicate with a closed mathematical graph,
and express program invariant in terms of it. Because a closed graph
doesn't necessarily decompose into closed subgraphs, these invariants
aren't considered under distribution. Thus, to prove that a state
modification maintains the graph invariant, one typically must resort
to reasoning about the whole graph. This global reasoning can be
ameliorated by introducing additional logical connectives and rules,
such as e.g., overlapping conjunction and ramified frame
rules~\cite{hob+vil:popl13}. In contrast, by relying on PCMs and
morphisms, we were able to stay within standard separation logic over
heaps.

\paragraph*{\textbf{Graphs in separation logic, with PCMs.}}
\citet{ser+nan+ban:pldi15} and~\citet{nan:oplss16} parametrized their
$\gseg$ predicate by a heap (itself a PCM), that serves as a
representation of a partial graph. Heaps obviously decompose under
disjoint union. However, since the $\gseg$ predicate includes the
restriction that the implemented graph must be closed, the predicate
itself doesn't distribute. This is worked around by introducing helper
relations for when one graph is a subgraph of another, which our
formulation provides for free (e.g., $\G_1$ is a subgraph of $\G$ iff
$\G = \G_1 \join \G_2$ for some $\G_2$). These works don't consider
PCM morphisms or Schorr-Waite.
 
More recently, \citet{kri+sum+wie:esop20}
and~\citet{mey+wie:wol:tacas23} introduced a theory of \emph{flows},
as a framework to study the decomposition of graph properties. Flows
evoke Yang's first proof, in that they serve as
decomposition-supporting proxies for the specification of the graph.
In contrast, we parametrize the reasoning by the whole graph, and compute proxies by
morphisms, when needed. That said, flows are intended for automated
reasoning, which we haven't considered.

Finally, \citet{cos+bro+pym:pregraphs} define the PCM of
\emph{pregraphs}. A pregraph is like our partial graph, but it may
contain incoming dangling edges, not just outgoing ones that sufficed
for us. The paper proceeds to define a separation logic over
pregraphs, much as separation logic over heaps is classically defined
using local actions~\cite{cal+ohe+yan:lics07}. In contrast, we use a
standard separation logic over heaps, where (partial) graphs are
merely a secondary user-level PCM. Having heaps and graphs coexist is
necessary for representing pointer-based graph algorithms such as
Schorr-Waite, but it requires PCM morphisms to mediate between the two.

\paragraph*{\textbf{Morphisms in separation logic.}}
While extant separation logics extensively rely on PCMs,
PCM morphisms remains underutilized. A notable exception is the work
of~\citet{far+nan+ban+del+fab:popl21}, who develop a theory of
\emph{partial PCM morphisms}---morphisms that distribute only under
certain conditions. These conditions give rise to \emph{separating
relations}, a novel algebraic structure with rich theoretical
properties. Together, morphisms and separating relations define when
one PCM is a sub-PCM of another, and how to refine Hoare logic triples
accordingly.

This refinement concept originated
with~\citet{nan+ban+del+fab:oopsla19}, who formulated it as a morphism
over \emph{resources}, which are algebraic structures modeling
state-transition systems for concurrency. Although resources involve a
form of structural inclusion,
they don't form PCMs, nor does their theory rely on join-based
decomposition. Moreover, resource morphisms act on programs, to
retroactively adapt a program's ghost code via a simulation function,
whereas PCM morphisms operate within logical assertions. As a result,
resources and their morphisms address orthogonal concerns to ours.
Finally, neither of the works applies their morphisms
to graphs.

\paragraph*{\textbf{Graphs algebraically.}}
Recent works in functional programming, theorem proving and category
theory have proposed treating graphs
algebraically~\cite{mas:arxiv21,mas:arxiv22,lie+sch:24,mok:17,mok:22},
though so far without application to separation logic. Most
recently,~\citet{ois+wu:popl25} share with us the representation of
graphs as maps, which they further endow with the algebraic structure
of rings, comprising distinct monoids for vertices and for edges. They
also propose several different algebraic constructions along with the
associated morphisms. Their application is in implementing graph
algorithms as state-free coinductive programs in Agda and Haskell.
In contrast, we use the algebra of graphs for specification and
verification of programs, not for their implementation, as our
programs in general, and our variant of Schorr-Waite in particular,
are implemented in an imperative pointer-based language customary for
separation logic. Our focus on separation logic also gives rise to
somewhat different monoidal structures. In particular, as we need to
relate graphs to pointers and framing, we focus on partiality and,
correspondingly, use a monoid of graphs whose join operation is very
different from the constructions considered in the above work.

\section{Conclusion}\label{sec:conclusion}
This paper establishes that graphs form a 
PCM when extended with \emph{dangling edges}, yielding \emph{partial
graphs}. The PCM structure facilitates natural composition through
subgraph joins, and induces \emph{PCM morphisms} as
structure-preserving functions. Our central contribution demonstrates
how PCM morphisms address separation logic's long-standing challenge
of effective graph verification.

Crucially, PCM morphisms enrich the foundational principle of
locality: while traditional framing enables spatial locality by
isolating heap portions, morphisms enable non-spatial locality by
isolating relevant graph subcomponents. The key mechanism is
\emph{contextual localization}; that is, distributing morphisms across
subgraph joins ($f(\G_1 \join \G_2) = f{\G_1} \join f{\G_2}$). Unlike
framing, which operates only at the top level of a specification,
contextual localization supports rewriting deeply inside a context.

We further employ \emph{higher-order combinators} like \textit{map} to
define helper notions for complex graph invariants (e.g., in
Schorr-Waite) and to support general lemmas that can be reused across
multiple instances. The integration of partial graphs, morphisms, and
higher-order combinators yields \emph{novel}, \emph{mechanized proofs}
for both \emph{Schorr-Waite algorithm} and the \emph{union-find data
structure}, achieving substantial conciseness through lemma reuse and
contextual localization.

All proofs are fully mechanized in Coq using the Hoare Type Theory
library. The mechanization uses functional (non-mutable) variables in
place of mutable ones used in the paper, but this is a technical
variation with no impact on the results.

\begin{acks}
We thank Anindya Banerjee and the anonymous ICFP reviewers for their
numerous useful comments. This work was partially supported by the
Spanish MICIU grants DECO PID2022-138072OB-I00, Maria de Maeztu
CEX2024-001471-M, PRODIGY TED2021-132464B-I00, and the ERC Starting
Grant CRETE GA:101039196. Any views, opinions and findings expressed
in the text are those of the authors and don't necessarily reflect the
views of the funding agencies.
\end{acks}

\vfill

\bibliography{bibliography} 


\begin{thebibliography}{62}


\ifx \showCODEN    \undefined \def \showCODEN     #1{\unskip}     \fi
\ifx \showISBNx    \undefined \def \showISBNx     #1{\unskip}     \fi
\ifx \showISBNxiii \undefined \def \showISBNxiii  #1{\unskip}     \fi
\ifx \showISSN     \undefined \def \showISSN      #1{\unskip}     \fi
\ifx \showLCCN     \undefined \def \showLCCN      #1{\unskip}     \fi
\ifx \shownote     \undefined \def \shownote      #1{#1}          \fi
\ifx \showarticletitle \undefined \def \showarticletitle #1{#1}   \fi
\ifx \showURL      \undefined \def \showURL       {\relax}        \fi
\providecommand\bibfield[2]{#2}
\providecommand\bibinfo[2]{#2}
\providecommand\natexlab[1]{#1}
\providecommand\showeprint[2][]{arXiv:#2}

\bibitem[Abrial(2003)]%
        {abr:fme03}
\bibfield{author}{\bibinfo{person}{Jean{-}Raymond Abrial}.}
  \bibinfo{year}{2003}\natexlab{}.
\newblock \showarticletitle{Event based sequential program development:
  application to constructing a pointer program}. In
  \bibinfo{booktitle}{\emph{Formal Methods Europe (FME)}}.
  \bibinfo{pages}{51--74}.
\newblock
\href{https://doi.org/10.1007/978-3-540-45236-2\_5}{doi:\nolinkurl{10.1007/978-3-540-45236-2\_5}}


\bibitem[Bornat(2000)]%
        {bornat00}
\bibfield{author}{\bibinfo{person}{Richard Bornat}.}
  \bibinfo{year}{2000}\natexlab{}.
\newblock \showarticletitle{Proving pointer programs in {H}oare logic}. In
  \bibinfo{booktitle}{\emph{Mathematics of Program Construction (MPC)}}.
  \bibinfo{pages}{102--126}.
\newblock
\href{https://doi.org/10.1007/10722010_8}{doi:\nolinkurl{10.1007/10722010_8}}


\bibitem[Bornat et~al\mbox{.}(2004)]%
        {bor+cal+ohe:space04}
\bibfield{author}{\bibinfo{person}{Richard Bornat}, \bibinfo{person}{Cristiano
  Calcagno}, {and} \bibinfo{person}{Peter O'Hearn}.}
  \bibinfo{year}{2004}\natexlab{}.
\newblock \showarticletitle{Local reasoning, separation and aliasing}. In
  \bibinfo{booktitle}{\emph{Workshop on {S}emantics, {P}rogram {A}nalysis and
  {C}omputing {E}nvironments for {M}emory {M}anagement (SPACE)}}.
\newblock
\urldef\tempurl%
\url{http://www.cs.ucl.ac.uk/staff/p.ohearn/papers/separation-and-aliasing.pdf}
\showURL{%
\tempurl}


\bibitem[Bornat et~al\mbox{.}(2005)]%
        {bor+cal+ohe+par:popl05}
\bibfield{author}{\bibinfo{person}{Richard Bornat}, \bibinfo{person}{Cristiano
  Calcagno}, \bibinfo{person}{Peter O'Hearn}, {and} \bibinfo{person}{Matthew
  Parkinson}.} \bibinfo{year}{2005}\natexlab{}.
\newblock \showarticletitle{Permission accounting in separation logic}. In
  \bibinfo{booktitle}{\emph{Principles of Programming Languages (POPL)}}.
  \bibinfo{pages}{259--270}.
\newblock
\href{https://doi.org/10.1145/1040305.1040327}{doi:\nolinkurl{10.1145/1040305.1040327}}


\bibitem[Broy and Pepper(1982)]%
        {bro+pep:toplas}
\bibfield{author}{\bibinfo{person}{Manfred Broy} {and} \bibinfo{person}{Peter
  Pepper}.} \bibinfo{year}{1982}\natexlab{}.
\newblock \showarticletitle{Combining algebraic and algorithmic reasoning: an
  approach to the {S}chorr-{W}aite algorithm}.
\newblock \bibinfo{journal}{\emph{ACM Trans. Program. Lang. Syst.}}
  \bibinfo{volume}{4}, \bibinfo{number}{3} (\bibinfo{year}{1982}),
  \bibinfo{pages}{362--381}.
\newblock
\href{https://doi.org/10.1145/357172.357175}{doi:\nolinkurl{10.1145/357172.357175}}


\bibitem[Bubel(2007)]%
        {Bubel2007}
\bibfield{author}{\bibinfo{person}{Richard Bubel}.}
  \bibinfo{year}{2007}\natexlab{}.
\newblock \showarticletitle{The {S}chorr-{W}aite-{A}lgorithm}.
\newblock In \bibinfo{booktitle}{\emph{Verification of Object-Oriented
  Software: The KeY Approach}}. \bibinfo{pages}{569--587}.
\newblock
\href{https://doi.org/10.1007/978-3-540-69061-0_15}{doi:\nolinkurl{10.1007/978-3-540-69061-0_15}}


\bibitem[Calcagno et~al\mbox{.}(2007)]%
        {cal+ohe+yan:lics07}
\bibfield{author}{\bibinfo{person}{Cristiano Calcagno},
  \bibinfo{person}{Peter~W. O'Hearn}, {and} \bibinfo{person}{Hongseok Yang}.}
  \bibinfo{year}{2007}\natexlab{}.
\newblock \showarticletitle{Local action and abstract separation logic}. In
  \bibinfo{booktitle}{\emph{Logic in Computer Science (LICS)}}.
  \bibinfo{pages}{366--378}.
\newblock
\href{https://doi.org/10.1109/LICS.2007.30}{doi:\nolinkurl{10.1109/LICS.2007.30}}


\bibitem[Chargu\'{e}raud and Pottier(2019)]%
        {cha+pot:jar17}
\bibfield{author}{\bibinfo{person}{Arthur Chargu\'{e}raud} {and}
  \bibinfo{person}{Fran\c{c}ois Pottier}.} \bibinfo{year}{2019}\natexlab{}.
\newblock \showarticletitle{Verifying the correctness and amortized complexity
  of a union-find implementation in separation logic with time credits}.
\newblock \bibinfo{journal}{\emph{J. Autom. Reasoning}}  \bibinfo{volume}{62}
  (\bibinfo{year}{2019}), \bibinfo{pages}{331--365}.
\newblock
\href{https://doi.org/10.1007/s10817-017-9431-7}{doi:\nolinkurl{10.1007/s10817-017-9431-7}}


\bibitem[Conchon and Filli\^{a}tre(2007)]%
        {con+fil:ml07}
\bibfield{author}{\bibinfo{person}{Sylvain Conchon} {and}
  \bibinfo{person}{Jean-Christophe Filli\^{a}tre}.}
  \bibinfo{year}{2007}\natexlab{}.
\newblock \showarticletitle{A persistent union-find data structure}. In
  \bibinfo{booktitle}{\emph{ML Workshop}}. \bibinfo{pages}{37--46}.
\newblock
\href{https://doi.org/10.1145/1292535.1292541}{doi:\nolinkurl{10.1145/1292535.1292541}}


\bibitem[Costa et~al\mbox{.}({[n.\,d.]})]%
        {cos+bro+pym:pregraphs}
\bibfield{author}{\bibinfo{person}{Diana Costa}, \bibinfo{person}{James
  Brotherston}, {and} \bibinfo{person}{David Pym}.}
  \bibinfo{year}{[n.\,d.]}\natexlab{}.
\newblock \bibinfo{booktitle}{\emph{Graph decomposition and local reasoning}}.
\newblock
\urldef\tempurl%
\url{http://www0.cs.ucl.ac.uk/staff/D.Pym/pregraphs.pdf}
\showURL{%
Retrieved 2024 from \tempurl}


\bibitem[Crespo and Kunz(2011)]%
        {cre+kun:sefm11}
\bibfield{author}{\bibinfo{person}{Juan~Manuel Crespo} {and}
  \bibinfo{person}{C{\'e}sar Kunz}.} \bibinfo{year}{2011}\natexlab{}.
\newblock \showarticletitle{A machine-checked framework for relational
  separation logic}. In \bibinfo{booktitle}{\emph{Software Engineering and
  Formal Methods (SEFM)}}. \bibinfo{pages}{122--137}.
\newblock
\href{https://doi.org/10.1007/978-3-642-24690-6_10}{doi:\nolinkurl{10.1007/978-3-642-24690-6_10}}


\bibitem[Dershowitz(1980)]%
        {der:ipl80}
\bibfield{author}{\bibinfo{person}{Nachum Dershowitz}.}
  \bibinfo{year}{1980}\natexlab{}.
\newblock \showarticletitle{The {S}chorr-{W}aite marking algorithm revisited}.
\newblock \bibinfo{journal}{\emph{Inform. Process. Lett.}}
  \bibinfo{volume}{11}, \bibinfo{number}{3} (\bibinfo{year}{1980}),
  \bibinfo{pages}{141--143}.
\newblock
\href{https://doi.org/10.1016/0020-0190(80)90130-1}{doi:\nolinkurl{10.1016/0020-0190(80)90130-1}}


\bibitem[Dinsdale-Young et~al\mbox{.}(2013)]%
        {din+bir+gar+par+yan:popl13}
\bibfield{author}{\bibinfo{person}{Thomas Dinsdale-Young},
  \bibinfo{person}{Lars Birkedal}, \bibinfo{person}{Philippa Gardner},
  \bibinfo{person}{Matthew Parkinson}, {and} \bibinfo{person}{Hongseok Yang}.}
  \bibinfo{year}{2013}\natexlab{}.
\newblock \showarticletitle{Views: compositional reasoning for concurrent
  programs}. In \bibinfo{booktitle}{\emph{Principles of Programming Languages
  (POPL)}}. \bibinfo{pages}{287--300}.
\newblock
\href{https://doi.org/10.1145/2429069.2429104}{doi:\nolinkurl{10.1145/2429069.2429104}}


\bibitem[Dufourd(2014)]%
        {duf:icfem14}
\bibfield{author}{\bibinfo{person}{Jean-Fran{\c{c}}ois Dufourd}.}
  \bibinfo{year}{2014}\natexlab{}.
\newblock \showarticletitle{Pointer program derivation using {C}oq: graphs and
  {S}chorr-{W}aite algorithm}. In \bibinfo{booktitle}{\emph{Formal Methods and
  Software Engineering (ICFEM)}}. \bibinfo{pages}{139--154}.
\newblock
\href{https://doi.org/10.1007/978-3-319-11737-9_10}{doi:\nolinkurl{10.1007/978-3-319-11737-9_10}}


\bibitem[Farka et~al\mbox{.}(2021)]%
        {far+nan+ban+del+fab:popl21}
\bibfield{author}{\bibinfo{person}{Franti\v{s}ek Farka},
  \bibinfo{person}{Aleksandar Nanevski}, \bibinfo{person}{Anindya Banerjee},
  \bibinfo{person}{Germ\'{a}n~Andr\'{e}s Delbianco}, {and}
  \bibinfo{person}{Ignacio F\'{a}bregas}.} \bibinfo{year}{2021}\natexlab{}.
\newblock \showarticletitle{On algebraic abstractions for concurrent separation
  logics}.
\newblock \bibinfo{journal}{\emph{Proc. ACM Program. Lang.}}
  \bibinfo{volume}{5}, \bibinfo{number}{POPL}, Article \bibinfo{articleno}{5}
  (\bibinfo{year}{2021}).
\newblock
\href{https://doi.org/10.1145/3434286}{doi:\nolinkurl{10.1145/3434286}}


\bibitem[Galler and Fisher(1964)]%
        {gal+fis:cacm64}
\bibfield{author}{\bibinfo{person}{Bernard~A. Galler} {and}
  \bibinfo{person}{Michael~J. Fisher}.} \bibinfo{year}{1964}\natexlab{}.
\newblock \showarticletitle{An improved equivalence algorithm}.
\newblock \bibinfo{journal}{\emph{Commun. ACM}} \bibinfo{volume}{7},
  \bibinfo{number}{5} (\bibinfo{year}{1964}), \bibinfo{pages}{301--303}.
\newblock
\href{https://doi.org/10.1145/364099.364331}{doi:\nolinkurl{10.1145/364099.364331}}


\bibitem[Gerhart(1979)]%
        {Gerhart1979}
\bibfield{author}{\bibinfo{person}{Susan~L. Gerhart}.}
  \bibinfo{year}{1979}\natexlab{}.
\newblock \showarticletitle{A derivation-oriented proof of the Schorr-Waite
  marking algorithm}.
\newblock In \bibinfo{booktitle}{\emph{Program Construction: International
  Summer School}}. \bibinfo{pages}{472--492}.
\newblock
\href{https://doi.org/10.1007/BFb0014678}{doi:\nolinkurl{10.1007/BFb0014678}}


\bibitem[Giorgino et~al\mbox{.}(2010)]%
        {gio+str+mat+pan:lopstr10}
\bibfield{author}{\bibinfo{person}{Mathieu Giorgino}, \bibinfo{person}{Martin
  Strecker}, \bibinfo{person}{Ralph Matthes}, {and} \bibinfo{person}{Marc
  Pantel}.} \bibinfo{year}{2010}\natexlab{}.
\newblock \showarticletitle{Verification of the {S}chorr-{W}aite algorithm -
  from trees to graphs}. In \bibinfo{booktitle}{\emph{Logic-Based Program
  Synthesis and Transformation (LOPSTR)}}. \bibinfo{pages}{67--83}.
\newblock
\href{https://doi.org/10.1007/978-3-642-20551-4_5}{doi:\nolinkurl{10.1007/978-3-642-20551-4_5}}


\bibitem[Grandury et~al\mbox{.}(2025a)]%
        {gra+nan+gry:icfp-artefact}
\bibfield{author}{\bibinfo{person}{Marcos Grandury},
  \bibinfo{person}{Aleksandar Nanevski}, {and} \bibinfo{person}{Alexander
  Gryzlov}.} \bibinfo{year}{2025}\natexlab{a}.
\newblock \bibinfo{booktitle}{\emph{Artifact for Verifying graph algorithms in
  separation logic: a case for an algebraic approach}}.
\newblock
\href{https://doi.org/10.5281/zenodo.15847649}{doi:\nolinkurl{10.5281/zenodo.15847649}}


\bibitem[Grandury et~al\mbox{.}(2025b)]%
        {gra+nan+gry:icfp25}
\bibfield{author}{\bibinfo{person}{Marcos Grandury},
  \bibinfo{person}{Aleksandar Nanevski}, {and} \bibinfo{person}{Alexander
  Gryzlov}.} \bibinfo{year}{2025}\natexlab{b}.
\newblock \showarticletitle{Verifying graph algorithms in separation logic: a
  case for an algebraic approach}.
\newblock \bibinfo{journal}{\emph{Proc. ACM Program. Lang.}}
  \bibinfo{volume}{9}, \bibinfo{number}{ICFP}, Article \bibinfo{articleno}{241}
  (\bibinfo{year}{2025}).
\newblock
\href{https://doi.org/10.1145/3747510}{doi:\nolinkurl{10.1145/3747510}}


\bibitem[Gries(1979)]%
        {Gries1979}
\bibfield{author}{\bibinfo{person}{David Gries}.}
  \bibinfo{year}{1979}\natexlab{}.
\newblock \showarticletitle{The {S}chorr-{W}aite graph marking algorithm}.
\newblock In \bibinfo{booktitle}{\emph{Program Construction: International
  Summer School}}. \bibinfo{pages}{58--69}.
\newblock
\href{https://doi.org/10.1007/BFb0014658}{doi:\nolinkurl{10.1007/BFb0014658}}


\bibitem[Griffiths(1979)]%
        {Griffiths1979}
\bibfield{author}{\bibinfo{person}{Michael Griffiths}.}
  \bibinfo{year}{1979}\natexlab{}.
\newblock \showarticletitle{Development of the {S}chorr-{W}aite algorithm}.
\newblock In \bibinfo{booktitle}{\emph{Program Construction: International
  Summer School}}. \bibinfo{pages}{464--471}.
\newblock
\href{https://doi.org/10.1007/BFb0014677}{doi:\nolinkurl{10.1007/BFb0014677}}


\bibitem[Hobor and Villard(2013)]%
        {hob+vil:popl13}
\bibfield{author}{\bibinfo{person}{Aquinas Hobor} {and} \bibinfo{person}{Jules
  Villard}.} \bibinfo{year}{2013}\natexlab{}.
\newblock \showarticletitle{The ramifications of sharing in data structures}.
  In \bibinfo{booktitle}{\emph{Principles of Programming Languages (POPL)}}.
  \bibinfo{pages}{523--536}.
\newblock
\href{https://doi.org/10.1145/2429069.2429131}{doi:\nolinkurl{10.1145/2429069.2429131}}


\bibitem[{HTT}(2010)]%
        {htt:github}
{HTT} \bibinfo{year}{2010}\natexlab{}.
\newblock \bibinfo{title}{{H}oare {T}ype {T}heory}.
\newblock
\urldef\tempurl%
\url{https://github.com/imdea-software/htt}
\showURL{%
Retrieved 2025 from \tempurl}


\bibitem[Hubert and March{\'e}(2005)]%
        {hub+mar:sefm05}
\bibfield{author}{\bibinfo{person}{Thierry Hubert} {and}
  \bibinfo{person}{Claude March{\'e}}.} \bibinfo{year}{2005}\natexlab{}.
\newblock \showarticletitle{A case study of {C} source code verification: the
  {S}chorr-{W}aite algorithm}.
\newblock \bibinfo{journal}{\emph{Software Engineering and Formal Methods
  (SEFM)}} (\bibinfo{year}{2005}), \bibinfo{pages}{190--199}.
\newblock
\href{https://doi.org/10.1109/SEFM.2005.1}{doi:\nolinkurl{10.1109/SEFM.2005.1}}


\bibitem[Ishtiaq and O'Hearn(2001)]%
        {ish+ohe:popl01}
\bibfield{author}{\bibinfo{person}{Samin Ishtiaq} {and}
  \bibinfo{person}{Peter~W. O'Hearn}.} \bibinfo{year}{2001}\natexlab{}.
\newblock \showarticletitle{{BI} as an assertion language for mutable data
  structures}. In \bibinfo{booktitle}{\emph{Principles of Programming Languages
  (POPL)}}. \bibinfo{pages}{14–26}.
\newblock
\href{https://doi.org/10.1145/360204.375719}{doi:\nolinkurl{10.1145/360204.375719}}


\bibitem[Jung et~al\mbox{.}(2015)]%
        {jun+swa+sie+sve+tur+bir+dre:popl15}
\bibfield{author}{\bibinfo{person}{Ralf Jung}, \bibinfo{person}{David Swasey},
  \bibinfo{person}{Filip Sieczkowski}, \bibinfo{person}{Kasper Svendsen},
  \bibinfo{person}{Aaron Turon}, \bibinfo{person}{Lars Birkedal}, {and}
  \bibinfo{person}{Derek Dreyer}.} \bibinfo{year}{2015}\natexlab{}.
\newblock \showarticletitle{Iris: monoids and invariants as an orthogonal basis
  for concurrent reasoning}. In \bibinfo{booktitle}{\emph{Principles of
  Porgramming Languages (POPL)}}. \bibinfo{pages}{637--650}.
\newblock
\href{https://doi.org/10.1145/2676726.2676980}{doi:\nolinkurl{10.1145/2676726.2676980}}


\bibitem[Kidney and Wu(2025)]%
        {ois+wu:popl25}
\bibfield{author}{\bibinfo{person}{Donnacha~Ois\'{\i}n Kidney} {and}
  \bibinfo{person}{Nicolas Wu}.} \bibinfo{year}{2025}\natexlab{}.
\newblock \showarticletitle{Formalising graph algorithms with coinduction}.
\newblock \bibinfo{journal}{\emph{Proc. ACM Program. Lang.}}
  \bibinfo{volume}{9}, \bibinfo{number}{POPL}, Article \bibinfo{articleno}{56}
  (\bibinfo{year}{2025}).
\newblock
\href{https://doi.org/10.1145/3704892}{doi:\nolinkurl{10.1145/3704892}}


\bibitem[Krishna et~al\mbox{.}(2020)]%
        {kri+sum+wie:esop20}
\bibfield{author}{\bibinfo{person}{Siddharth Krishna},
  \bibinfo{person}{Alexander~J. Summers}, {and} \bibinfo{person}{Thomas Wies}.}
  \bibinfo{year}{2020}\natexlab{}.
\newblock \showarticletitle{Local reasoning for global graph properties}. In
  \bibinfo{booktitle}{\emph{European Symposium on Programming (ESOP)}}.
  \bibinfo{pages}{308--335}.
\newblock
\href{https://doi.org/10.1007/978-3-030-44914-8_12}{doi:\nolinkurl{10.1007/978-3-030-44914-8_12}}


\bibitem[Krishnaswami(2011)]%
        {neelk:phd}
\bibfield{author}{\bibinfo{person}{Neelakantan~R. Krishnaswami}.}
  \bibinfo{year}{2011}\natexlab{}.
\newblock \emph{\bibinfo{title}{Verifying higher-order imperative programs with
  higher-order separation logic}}.
\newblock \bibinfo{thesistype}{Ph.\,D. Dissertation}. \bibinfo{school}{Carnegie
  Mellon University}.
\newblock
\urldef\tempurl%
\url{http://reports-archive.adm.cs.cmu.edu/anon/2012/CMU-CS-12-127.pdf}
\showURL{%
\tempurl}


\bibitem[Lammich and Meis(2012)]%
        {lam+mei:2012}
\bibfield{author}{\bibinfo{person}{Peter Lammich} {and} \bibinfo{person}{Rene
  Meis}.} \bibinfo{year}{2012}\natexlab{}.
\newblock \showarticletitle{A separation logic framework for Imperative {HOL}}.
\newblock \bibinfo{journal}{\emph{Arch. Formal Proofs}} (\bibinfo{year}{2012}).
\newblock
\urldef\tempurl%
\url{https://www.isa-afp.org/entries/Separation\_Logic\_Imperative\_HOL.shtml}
\showURL{%
\tempurl}


\bibitem[Leino(2010)]%
        {leino:lpar10}
\bibfield{author}{\bibinfo{person}{K.~Rustan~M. Leino}.}
  \bibinfo{year}{2010}\natexlab{}.
\newblock \showarticletitle{Dafny: an automatic program verifier for functional
  correctness}. In \bibinfo{booktitle}{\emph{Logic for Programming, Artificial
  Intelligence, and Reasoning (LPAR)}}. \bibinfo{pages}{348--370}.
\newblock
\href{https://doi.org/10.1007/978-3-642-17511-4_20}{doi:\nolinkurl{10.1007/978-3-642-17511-4_20}}


\bibitem[Liell-Cock and Schrijvers(2024)]%
        {lie+sch:24}
\bibfield{author}{\bibinfo{person}{Jack Liell-Cock} {and} \bibinfo{person}{Tom
  Schrijvers}.} \bibinfo{year}{2024}\natexlab{}.
\newblock \showarticletitle{Let a thousand flowers bloom: an algebraic
  representation for edge graphs}.
\newblock \bibinfo{journal}{\emph{The Art, Science, and Engineering of
  Programming}} \bibinfo{volume}{8}, \bibinfo{number}{3}
  (\bibinfo{year}{2024}).
\newblock
\showISSN{2473-7321}
\href{https://doi.org/10.22152/programming-journal.org/2024/8/9}{doi:\nolinkurl{10.22152/programming-journal.org/2024/8/9}}


\bibitem[Loginov et~al\mbox{.}(2006)]%
        {log+rep+sag:sas06}
\bibfield{author}{\bibinfo{person}{Alexey Loginov}, \bibinfo{person}{Thomas~W.
  Reps}, {and} \bibinfo{person}{Mooly Sagiv}.} \bibinfo{year}{2006}\natexlab{}.
\newblock \showarticletitle{Automated verification of the
  {D}eutsch-{S}chorr-{W}aite tree-traversal algorithm}. In
  \bibinfo{booktitle}{\emph{Static Analysis Symposium (SAS)}}.
  \bibinfo{pages}{261--279}.
\newblock
\href{https://doi.org/10.1007/11823230\_17}{doi:\nolinkurl{10.1007/11823230\_17}}


\bibitem[Master(2021)]%
        {mas:arxiv21}
\bibfield{author}{\bibinfo{person}{Jade Master}.}
  \bibinfo{year}{2021}\natexlab{}.
\newblock \bibinfo{title}{The open algebraic path problem}.
\newblock
\href{https://doi.org/10.48550/arXiv.2005.06682}{doi:\nolinkurl{10.48550/arXiv.2005.06682}}


\bibitem[Master(2022)]%
        {mas:arxiv22}
\bibfield{author}{\bibinfo{person}{Jade Master}.}
  \bibinfo{year}{2022}\natexlab{}.
\newblock \bibinfo{title}{How to compose shortest paths}.
\newblock
\href{https://doi.org/10.48550/arXiv.2205.15306}{doi:\nolinkurl{10.48550/arXiv.2205.15306}}


\bibitem[Mehta and Nipkow(2003)]%
        {meh+nip:cade03}
\bibfield{author}{\bibinfo{person}{Farhad Mehta} {and} \bibinfo{person}{Tobias
  Nipkow}.} \bibinfo{year}{2003}\natexlab{}.
\newblock \showarticletitle{Proving pointer programs in higher-order logic}. In
  \bibinfo{booktitle}{\emph{Conference on Automated Deduction (CADE)}}.
  \bibinfo{pages}{121--135}.
\newblock
\href{https://doi.org/10.1007/978-3-540-45085-6_10}{doi:\nolinkurl{10.1007/978-3-540-45085-6_10}}


\bibitem[Meyer et~al\mbox{.}(2023)]%
        {mey+wie:wol:tacas23}
\bibfield{author}{\bibinfo{person}{Roland Meyer}, \bibinfo{person}{Thomas
  Wies}, {and} \bibinfo{person}{Sebastian Wolff}.}
  \bibinfo{year}{2023}\natexlab{}.
\newblock \showarticletitle{Make flows small again: revisiting the flow
  framework}. In \bibinfo{booktitle}{\emph{Tools and Algorithms for the
  Construction and Analysis of Systems (TACAS)}}. \bibinfo{pages}{628--646}.
\newblock
\href{https://doi.org/10.1007/978-3-031-30823-9_32}{doi:\nolinkurl{10.1007/978-3-031-30823-9_32}}


\bibitem[Mokhov(2017)]%
        {mok:17}
\bibfield{author}{\bibinfo{person}{Andrey Mokhov}.}
  \bibinfo{year}{2017}\natexlab{}.
\newblock \showarticletitle{Algebraic graphs with class (functional pearl)}. In
  \bibinfo{booktitle}{\emph{International Symposium on Haskell}}.
  \bibinfo{pages}{2–13}.
\newblock
\href{https://doi.org/10.1145/3122955.3122956}{doi:\nolinkurl{10.1145/3122955.3122956}}


\bibitem[Mokhov(2022)]%
        {mok:22}
\bibfield{author}{\bibinfo{person}{Andrey Mokhov}.}
  \bibinfo{year}{2022}\natexlab{}.
\newblock \showarticletitle{United monoids: finding simplicial sets and
  labelled algebraic graphs in trees}.
\newblock \bibinfo{journal}{\emph{The Art, Science, and Engineering of
  Programming}} \bibinfo{volume}{6}, \bibinfo{number}{3}
  (\bibinfo{year}{2022}).
\newblock
\href{https://doi.org/10.22152/programming-journal.org/2022/6/12}{doi:\nolinkurl{10.22152/programming-journal.org/2022/6/12}}


\bibitem[Morris(1982)]%
        {Morris1982}
\bibfield{author}{\bibinfo{person}{Joseph~M. Morris}.}
  \bibinfo{year}{1982}\natexlab{}.
\newblock \showarticletitle{A proof of the {S}chorr-{W}aite algorithm}.
\newblock In \bibinfo{booktitle}{\emph{Theoretical Foundations of Programming
  Methodology: Lecture Notes of an International Summer School}}.
  \bibinfo{pages}{43--51}.
\newblock
\href{https://doi.org/10.1007/978-94-009-7893-5_5}{doi:\nolinkurl{10.1007/978-94-009-7893-5_5}}


\bibitem[Nanevski(2016)]%
        {nan:oplss16}
\bibfield{author}{\bibinfo{person}{Aleksandar Nanevski}.}
  \bibinfo{year}{2016}\natexlab{}.
\newblock \bibinfo{title}{Separation logic and concurrency}.
\newblock \bibinfo{howpublished}{Lecture notes for the Oregon Programming
  Languages Summer School (OPLSS)}.
\newblock
\urldef\tempurl%
\url{https://www.cs.uoregon.edu/research/summerschool/summer16/curriculum.php}
\showURL{%
\tempurl}


\bibitem[Nanevski et~al\mbox{.}(2019)]%
        {nan+ban+del+fab:oopsla19}
\bibfield{author}{\bibinfo{person}{Aleksandar Nanevski},
  \bibinfo{person}{Anindya Banerjee}, \bibinfo{person}{Germ\'{a}n~Andr\'{e}s
  Delbianco}, {and} \bibinfo{person}{Ignacio F\'{a}bregas}.}
  \bibinfo{year}{2019}\natexlab{}.
\newblock \showarticletitle{Specifying concurrent programs in separation logic:
  morphisms and simulations}.
\newblock \bibinfo{journal}{\emph{Proc. ACM Program. Lang.}}
  \bibinfo{volume}{3}, \bibinfo{number}{OOPSLA}, Article
  \bibinfo{articleno}{161} (\bibinfo{year}{2019}).
\newblock
\href{https://doi.org/10.1145/3360587}{doi:\nolinkurl{10.1145/3360587}}


\bibitem[Nanevski et~al\mbox{.}(2006)]%
        {nan+mor+bir:icfp06}
\bibfield{author}{\bibinfo{person}{Aleksandar Nanevski}, \bibinfo{person}{Greg
  Morrisett}, {and} \bibinfo{person}{Lars Birkedal}.}
  \bibinfo{year}{2006}\natexlab{}.
\newblock \showarticletitle{Polymorphism and separation in {H}oare {T}ype
  {T}heory}. In \bibinfo{booktitle}{\emph{International Conference on
  Functional Programming (ICFP)}}. \bibinfo{pages}{62–73}.
\newblock
\href{https://doi.org/10.1145/1159803.1159812}{doi:\nolinkurl{10.1145/1159803.1159812}}


\bibitem[Nanevski et~al\mbox{.}(2008)]%
        {nan+mor+shi+gov+bir:icfp08}
\bibfield{author}{\bibinfo{person}{Aleksandar Nanevski}, \bibinfo{person}{Greg
  Morrisett}, \bibinfo{person}{Avraham Shinnar}, \bibinfo{person}{Paul
  Govereau}, {and} \bibinfo{person}{Lars Birkedal}.}
  \bibinfo{year}{2008}\natexlab{}.
\newblock \showarticletitle{Ynot: dependent types for imperative programs}. In
  \bibinfo{booktitle}{\emph{International Conference on Functional Programming
  (ICFP)}}. \bibinfo{pages}{229–240}.
\newblock
\href{https://doi.org/10.1145/1411203.1411237}{doi:\nolinkurl{10.1145/1411203.1411237}}


\bibitem[Nanevski et~al\mbox{.}(2010)]%
        {nan+vaf+ber:popl10}
\bibfield{author}{\bibinfo{person}{Aleksandar Nanevski},
  \bibinfo{person}{Viktor Vafeiadis}, {and} \bibinfo{person}{Josh Berdine}.}
  \bibinfo{year}{2010}\natexlab{}.
\newblock \showarticletitle{Structuring the verification of heap-manipulating
  programs}. In \bibinfo{booktitle}{\emph{Principles of Programming Languages
  (POPL)}}. \bibinfo{pages}{261--274}.
\newblock
\href{https://doi.org/10.1145/1706299.1706331}{doi:\nolinkurl{10.1145/1706299.1706331}}


\bibitem[O'Hearn et~al\mbox{.}(2001)]%
        {ohe+rey+yan:csl01}
\bibfield{author}{\bibinfo{person}{Peter~W. O'Hearn}, \bibinfo{person}{John~C.
  Reynolds}, {and} \bibinfo{person}{Hongseok Yang}.}
  \bibinfo{year}{2001}\natexlab{}.
\newblock \showarticletitle{Local reasoning about programs that alter data
  structures}. In \bibinfo{booktitle}{\emph{Computer Science Logic (CSL)}}.
  \bibinfo{pages}{1--19}.
\newblock
\href{https://doi.org/10.1007/3-540-44802-0_1}{doi:\nolinkurl{10.1007/3-540-44802-0_1}}


\bibitem[Preoteasa and Back(2012)]%
        {pre+bac:fac12}
\bibfield{author}{\bibinfo{person}{Viorel Preoteasa} {and}
  \bibinfo{person}{Ralph{-}Johan Back}.} \bibinfo{year}{2012}\natexlab{}.
\newblock \showarticletitle{Invariant diagrams with data refinement}.
\newblock \bibinfo{journal}{\emph{Formal Aspects Comput.}}
  \bibinfo{volume}{24}, \bibinfo{number}{1} (\bibinfo{year}{2012}),
  \bibinfo{pages}{67--95}.
\newblock
\href{https://doi.org/10.1007/S00165-011-0195-2}{doi:\nolinkurl{10.1007/S00165-011-0195-2}}


\bibitem[Preoteasa and Back(2010)]%
        {pre+bac:afp10}
\bibfield{author}{\bibinfo{person}{Viorel Preoteasa} {and}
  \bibinfo{person}{Ralph-Johan Back}.} \bibinfo{year}{2010}\natexlab{}.
\newblock \showarticletitle{Verification of the {D}eutsch-{S}chorr-{W}aite
  graph marking algorithm using data refinement}.
\newblock \bibinfo{journal}{\emph{Arch. Formal Proofs}} (\bibinfo{year}{2010}).
\newblock
\urldef\tempurl%
\url{https://www.isa-afp.org/browser_info/current/AFP/GraphMarkingIBP/document.pdf}
\showURL{%
\tempurl}


\bibitem[Reynolds(2002)]%
        {rey:lics02}
\bibfield{author}{\bibinfo{person}{John~C. Reynolds}.}
  \bibinfo{year}{2002}\natexlab{}.
\newblock \showarticletitle{Separation logic: a logic for shared mutable data
  structures}. In \bibinfo{booktitle}{\emph{Logic in Computer Science (LICS)}}.
  \bibinfo{pages}{55--74}.
\newblock
\href{https://doi.org/10.1109/LICS.2002.1029817}{doi:\nolinkurl{10.1109/LICS.2002.1029817}}


\bibitem[Roever(1977)]%
        {roe:icalp77}
\bibfield{author}{\bibinfo{person}{Willem-Paul Roever}.}
  \bibinfo{year}{1977}\natexlab{}.
\newblock \showarticletitle{On backtracking and greatest fixpoints}. In
  \bibinfo{booktitle}{\emph{International Colloquium on Automata, Languages and
  Programming (ICALP)}}. \bibinfo{pages}{412--429}.
\newblock
\href{https://doi.org/10.1007/3-540-08342-1_32}{doi:\nolinkurl{10.1007/3-540-08342-1_32}}


\bibitem[Schorr and Waite(1967)]%
        {schorr1967}
\bibfield{author}{\bibinfo{person}{Herbert Schorr} {and}
  \bibinfo{person}{William~M. Waite}.} \bibinfo{year}{1967}\natexlab{}.
\newblock \showarticletitle{An efficient machine-independent procedure for
  garbage collection in various list structures}.
\newblock \bibinfo{journal}{\emph{Commun. ACM}} \bibinfo{volume}{10},
  \bibinfo{number}{8} (\bibinfo{year}{1967}), \bibinfo{pages}{501--506}.
\newblock
\href{https://doi.org/10.1145/363534.363554}{doi:\nolinkurl{10.1145/363534.363554}}


\bibitem[Sergey et~al\mbox{.}(2015a)]%
        {ser+nan+ban:pldi15}
\bibfield{author}{\bibinfo{person}{Ilya Sergey}, \bibinfo{person}{Aleksandar
  Nanevski}, {and} \bibinfo{person}{Anindya Banerjee}.}
  \bibinfo{year}{2015}\natexlab{a}.
\newblock \showarticletitle{Mechanized verification of fine-grained concurrent
  programs}. In \bibinfo{booktitle}{\emph{Programming Language Design and
  Implementation (PLDI)}}. \bibinfo{pages}{77--87}.
\newblock
\href{https://doi.org/10.1145/2737924.2737964}{doi:\nolinkurl{10.1145/2737924.2737964}}


\bibitem[Sergey et~al\mbox{.}(2015b)]%
        {ser+nan+ban:esop15}
\bibfield{author}{\bibinfo{person}{Ilya Sergey}, \bibinfo{person}{Aleksandar
  Nanevski}, {and} \bibinfo{person}{Anindya Banerjee}.}
  \bibinfo{year}{2015}\natexlab{b}.
\newblock \showarticletitle{Specifying and verifying concurrent algorithms with
  histories and subjectivity}. In \bibinfo{booktitle}{\emph{European Symposium
  on Programming (ESOP)}}. \bibinfo{pages}{333--358}.
\newblock
\href{https://doi.org/10.1007/978-3-662-46669-8_14}{doi:\nolinkurl{10.1007/978-3-662-46669-8_14}}


\bibitem[Suzuki(1976)]%
        {Suzuki76}
\bibfield{author}{\bibinfo{person}{Norihisa Suzuki}.}
  \bibinfo{year}{1976}\natexlab{}.
\newblock \emph{\bibinfo{title}{Automatic verification of programs with complex
  data structures}}.
\newblock \bibinfo{thesistype}{Ph.\,D. Dissertation}. \bibinfo{school}{Stanford
  University}.
\newblock


\bibitem[Topor(1979)]%
        {Topor1979TheCO}
\bibfield{author}{\bibinfo{person}{Rodney~W. Topor}.}
  \bibinfo{year}{1979}\natexlab{}.
\newblock \showarticletitle{The correctness of the {S}chorr-{W}aite list
  marking algorithm}.
\newblock \bibinfo{journal}{\emph{Acta Informatica}}  \bibinfo{volume}{11}
  (\bibinfo{year}{1979}), \bibinfo{pages}{211--221}.
\newblock
\href{https://doi.org/10.1007/BF00289067}{doi:\nolinkurl{10.1007/BF00289067}}


\bibitem[Wang(2019)]%
        {wang:phd20}
\bibfield{author}{\bibinfo{person}{Shengyi Wang}.}
  \bibinfo{year}{2019}\natexlab{}.
\newblock \emph{\bibinfo{title}{Mechanized verification of graph-manipulating
  programs}}.
\newblock \bibinfo{thesistype}{Ph.\,D. Dissertation}. \bibinfo{school}{School
  of Computing, National University of Singapore}.
\newblock
\urldef\tempurl%
\url{https://scholarbank.nus.edu.sg/handle/10635/166280}
\showURL{%
\tempurl}


\bibitem[Wang et~al\mbox{.}(2019)]%
        {wan+qin+moh+hob:oopsla19}
\bibfield{author}{\bibinfo{person}{Shengyi Wang}, \bibinfo{person}{Qinxiang
  Cao}, \bibinfo{person}{Anshuman Mohan}, {and} \bibinfo{person}{Aquinas
  Hobor}.} \bibinfo{year}{2019}\natexlab{}.
\newblock \showarticletitle{Certifying graph-manipulating {C} programs via
  localizations within data structures}.
\newblock \bibinfo{journal}{\emph{Proc. ACM Program. Lang.}}
  \bibinfo{volume}{3}, \bibinfo{number}{OOPSLA}, Article
  \bibinfo{articleno}{171} (\bibinfo{year}{2019}).
\newblock
\href{https://doi.org/10.1145/3360597}{doi:\nolinkurl{10.1145/3360597}}


\bibitem[Ward(1996)]%
        {ward96}
\bibfield{author}{\bibinfo{person}{Martin Ward}.}
  \bibinfo{year}{1996}\natexlab{}.
\newblock \showarticletitle{Derivation of data intensive algorithms by formal
  transformation: the {S}chorr-{W}aite graph marking algorithm}.
\newblock \bibinfo{journal}{\emph{IEEE Trans. Softw. Eng.}}
  \bibinfo{volume}{22}, \bibinfo{number}{9} (\bibinfo{year}{1996}),
  \bibinfo{pages}{665–686}.
\newblock
\href{https://doi.org/10.1109/32.541437}{doi:\nolinkurl{10.1109/32.541437}}


\bibitem[Yang(2001a)]%
        {Yang2001AnEO}
\bibfield{author}{\bibinfo{person}{Hongseok Yang}.}
  \bibinfo{year}{2001}\natexlab{a}.
\newblock \showarticletitle{An example of local reasoning in {BI} pointer
  logic: the {S}chorr−{W}aite graph marking algorithm}. In
  \bibinfo{booktitle}{\emph{Workshop on {S}emantics, {P}rogram {A}nalysis and
  {C}omputing {E}nvironments for {M}emory {M}anagement (SPACE)}}.
\newblock
\urldef\tempurl%
\url{https://citeseerx.ist.psu.edu/document?repid=rep1&type=pdf&doi=d93456b70d4ec3569c29409c73aabb0fddb6ed49}
\showURL{%
\tempurl}


\bibitem[Yang(2001b)]%
        {yang2001}
\bibfield{author}{\bibinfo{person}{Hongseok Yang}.}
  \bibinfo{year}{2001}\natexlab{b}.
\newblock \emph{\bibinfo{title}{Local reasoning for stateful programs}}.
\newblock \bibinfo{thesistype}{Ph.\,D. Dissertation}.
  \bibinfo{school}{University of Illinois at Urbana-Champaign}.
\newblock
\urldef\tempurl%
\url{https://drive.google.com/open?id=1gFdENbez-p6tCksvz2NgdU7y_aV0zHLW}
\showURL{%
\tempurl}


\bibitem[Yang(2007)]%
        {Yang2007RelationalSL}
\bibfield{author}{\bibinfo{person}{Hongseok Yang}.}
  \bibinfo{year}{2007}\natexlab{}.
\newblock \showarticletitle{Relational separation logic}.
\newblock \bibinfo{journal}{\emph{Theor. Comput. Sci.}}  \bibinfo{volume}{375}
  (\bibinfo{year}{2007}), \bibinfo{pages}{308--334}.
\newblock
\href{https://doi.org/10.1016/j.tcs.2006.12.036}{doi:\nolinkurl{10.1016/j.tcs.2006.12.036}}


\end{thebibliography}

\appendix

\section{Proof outline for the length-calculating program}\label{sec:lengthoutline}

\begin{align*}
& \spec{list\ \alpha_0\ (i, \nnull)}\\
&  n \assign 0;\\
& \spec{list\ \alpha_0\ (i, \nnull) \wedge n = 0}\\
&  j \assign i;\\
& \spec{i = j \wedge list\ \alpha_0\ (i, \nnull) \wedge n = 0}\\
& \spec{list\ \nil\ (i, j) * list\ \alpha_0\ (j, \nnull) \wedge n = \len{\nil} \wedge \alpha_0 = \nil \cat \alpha_0}\\
& \spec{\exists \alpha\ \beta\ldot list\ \alpha\ (i, j) * list\ \beta\ (j, \nnull) \wedge n = \len{\alpha} \wedge \alpha_0 = \alpha \cat \beta}\\
& \textrm{\textbf{while}}\ j \neq \nnull\ \textrm{\textbf{do}} \\
& \quad \spec{\exists \alpha\ \beta\ldot \lseg\ \alpha\ (i, j) * \lseg\ \beta\ (j, \nnull) \wedge n = \len{\alpha} \wedge \alpha_0 = \alpha \cat \beta \wedge j \neq \nnull}\\
& \quad \spec{\exists \alpha\ b\ \beta'\ldot \lseg\ \alpha\ (i, j) * \lseg\ (b\cons\beta')\ (j, \nnull) \wedge n = \len{\alpha} \wedge \alpha_0 = \alpha \cat (b \cons \beta')}\\
& \quad \spec{\exists \alpha\ b\ \beta'\ k\ldot \lseg\ \alpha\ (i, j) * j \Mapsto b, k * \lseg\ \beta'\ (k, \nnull) \wedge n = \len{\alpha} \wedge \alpha_0 = \alpha \cat (b \cons \beta')}\\
& \quad j \assign \deref{j.\textrm{next}};\\
& \quad \spec{\exists \alpha\ b\ \beta'\ j'\ldot \lseg\ \alpha\ (i, j') * j' \Mapsto b, j * \lseg\ \beta'\ (j, \nnull) \wedge n = \len{\alpha} \wedge \alpha_0 = \alpha \cat (b \cons \beta')}\\
& \quad \spec{\exists \alpha\ b\ \beta'\ldot \lseg\ (\alpha\cat\sngl{b})\ (i, j) * \lseg\ \beta'\ (j, \nnull) \wedge n = \len{\alpha} \wedge \alpha_0 = (\alpha\cat\sngl{b}) \cat \beta'}\\
& \quad n \assign n+1;\\
& \quad \spec{\exists \alpha\ b\ \beta'\ldot \lseg\ (\alpha\cat\sngl{b})\ (i, j) * \lseg\ \beta'\ (j, \nnull) \wedge n = \len{\alpha} + 1 \wedge \alpha_0 = (\alpha\cat\sngl{b}) \cat \beta'}\\
& \quad \spec{\exists \alpha\ b\ \beta'\ldot \lseg\ (\alpha\cat\sngl{b})\ (i, j) * \lseg\ \beta'\ (j, \nnull) \wedge n = \llen{\alpha\cat\sngl{b}} \wedge \alpha_0 = (\alpha\cat\sngl{b}) \cat \beta'}\\
& \quad \spec{\exists \alpha'\ \beta'\ldot \lseg\ \alpha'\ (i, j) * \lseg\ \beta'\ (j, \nnull) \wedge n = \len{\alpha'} \wedge \alpha_0 = \alpha' \cat \beta'}\\
& \textrm{\textbf{end while}} \\
&  \spec{\exists \alpha'\ \beta'\ldot list\ \alpha'\ (i, j) * list\ \beta'\ (j, \nnull) \wedge n = \len{\alpha'} \wedge \alpha_0 = \alpha' \cat \beta' \wedge j = \nnull}\\
&  \spec{\exists \alpha'\ \beta'\ldot list\ \alpha'\ (i, j) * list\ \beta'\ (j, \nnull) \wedge n = \len{\alpha'} \wedge \alpha_0 = \alpha' \cat \beta' \wedge j = \nnull \wedge \beta' = \nil}\\
&  \spec{\exists \alpha'\ldot list\ \alpha'\ (i, \nnull) * list\ \nil\ (\nnull, \nnull) \wedge n = \len{\alpha'} \wedge \alpha_0 = \alpha'}\\
&  \spec{list\ \alpha_0\ (i, \nnull) * emp \wedge n = \len{\alpha_0}}\\
&  \spec{list\ \alpha_0\ (i, \nnull) \wedge n = \len{\alpha_0}}
\end{align*}

\section{Proof outline for computing if $t$ is marked (or $\nnull$)}\label{sec:readtm-outline}
{\allowdisplaybreaks
  \begin{align*}
    &\spec{\exists \G \ldot \gseg\ \G \wedge \inv\ \G_0\ \G\ t\ p}\\
    &\quad \spec{\gseg\ \G \wedge \inv\ \G_0\ \G\ t\ p}\\
  &\quad \textrm{\textbf{if}}\ t = \nnull\ \textrm{\textbf{then}}\\
  &\qquad \spec{\gseg\ \G \wedge \inv\ \G_0\ \G\ t\ p \wedge t = null}\\
  &\qquad \tm \assign \textrm{true}\\
  &\qquad \spec{\gseg\ \G \wedge \inv\ \G_0\ \G\ t\ p \wedge t = null \wedge \tm = \textrm{true}}\\
  &\qquad \spec{\gseg\ \G \wedge \inv\ \G_0\ \G\ t\ p \wedge \tm = (t = null)}\\
  &\quad \textrm{\textbf{else}}\ \\
  &\qquad \spec{\gseg\ \G \wedge \inv\ \G_0\ \G\ t\ p \wedge t \neq null}\\
  &\qquad \spec{\exists m \ldot t \Mapsto m,-,- * \gseg\ \G{\setminus} t \wedge \inv\ \G_0\ \G\ t\ p}\\
  &\qquad \tmp \assign \deref{t.m};\\
  &\qquad \spec{\exists m \ldot t \Mapsto m,-,- * \gseg\ \G{\setminus} t \wedge \inv\ \G_0\ \G\ t\ p \wedge \tmp = m}\\
    &\qquad \tm \assign (\tmp \neq \notM)\\
    &\qquad \spec{\exists m \ldot t \Mapsto m,-,- * \gseg\ \G{\setminus} t \wedge \inv\ \G_0\ \G\ t\ p \wedge \tmp = m \wedge \tm = (\tmp \neq \notM)}\\
    &\qquad \spec{\exists m \ldot t \Mapsto m,-,- * \gseg\ \G{\setminus} t \wedge \inv\ \G_0\ \G\ t\ p \wedge \tm = (m \neq \notM)}\\
    &\qquad \spec{\gseg\ \G \wedge \inv\ \G_0\ \G\ t\ p \wedge \tm = (t.m \neq \notM)}\\
  &\qquad \spec{\gseg\ \G \wedge \inv\ \G_0\ \G\ t\ p \wedge \tm = (t \in \nodes\ {\G\filter{\leftM,\rightM,\killed}})}\\
    &\quad \textrm{\textbf{end if}}\\
    &\quad \spec{\gseg\ \G \wedge \inv\ \G_0\ \G\ t\ p \wedge \tm = (t = \nnull \vee t \in \nodes\ {\G\filter{\leftM,\rightM,\killed}})}\\
    &\quad \spec{\gseg\ \G \wedge \inv\ \G_0\ \G\ t\ p \wedge \tm = (t \in \setmarknull{\G})} \\
    &\spec{\exists \G \ldot \gseg\ \G \wedge \inv\ \G_0\ \G\ t\ p \wedge \tm = (t \in \setmarknull{\G})}
\end{align*}}
 
\section{Proof for SWING}

The pre- and postcondition for SWING derive from
lines~\ref{sw-ln18} and~\ref{sw-ln20} of Fig.~\ref{fig:sw}.
\begin{equation}\label{swing-spec}
\begin{array}[c]{c} 
\spec{\exists \G\ldot \gseg\ \G \wedge \invx{\GG}{\G}{t}{p}
    \wedge t \in \setmarknull{\G} \wedge \Gm\ p = \leftM}\\
\textrm{SWING}  \\
\spec{\exists \G'\ldot\ \gseg\ \G' \wedge \invx{\GG}{\G'}{t}{p}}
\end{array}
\end{equation}
The precondition says that the heap implements a well-formed graph
$\G$, that satisfies the invariant. Additionally, $t$ is marked or
$\nnull$ and $p$ is marked $\leftM$. The postcondition asserts that
the heap represents a new graph $\G'$ that satisfies the invariant for
the updated values of $t$ and $p$.
\makeatletter
\newcounter{codeswing}
\renewcommand{\lineno}{\stepcounter{codeswing}\textsc{\thecodeswing}.\quad}
\renewcommand{\linelb}[1]{{\refstepcounter{codeswing}\ltx@label{#1}\textsc{\thecodeswing}.\quad}}
\makeatother
{\allowdisplaybreaks
  \begin{align*}
\linelb{swing-ln1}    & \spec{\exists \G\ldot \gseg\ \G \wedge \invx{\GG}{\G}{t}{p}
    \wedge t \in \setmarknull{\G} \wedge \Gm\ p = \leftM}\\
\linelb{swing-ln2}  &\qquad \sspecopen{\gseg\ \G \wedge t = \TT \wedge p = \pp}\\
                  &\qquad \opensspec{\textcolor{teal}{\wedge \invx{\GG}{\G}{\TT}{\pp}
      \wedge \TT \in \setmarknull{\G} \wedge \G\ \pp = (\leftM, [\pl,\pr])}} \nonumber\\
\linelb{swing-ln3}  & \qquad\qquad \spec{\gseg\ (\pp \mapsto (\leftM, [\pl,\pr]) \join \G{\setminus} \pp) 
\wedge t = \TT \wedge p = \pp}\\
\linelb{swing-ln4} & \qquad\qquad \spec{(\pp \Mapsto \leftM, \pl, \pr \wedge t = \TT \wedge p = \pp)\
       \textcolor{teal}{*\ \gseg\ \G{\setminus} \pp}}\\
\linelb{swing-ln5}   &\qquad\qquad\qquad\spec{\pp \Mapsto \leftM,\pl,\pr \wedge t = \TT \wedge p = \pp}\\
    &\qquad\qquad\qquad\ \tmp_1 \assign \deref{p.r};\ tmp_2 \assign \deref{p.l};\ p.r \mutate \tmp_2;\
    p.l \mutate t;\ p.m \mutate \rightM;\ t \assign \tmp_1;  \nonumber\\
\linelb{swing-ln6}   &\qquad\qquad\qquad\spec{\pp \Mapsto \rightM,\TT,\pl \wedge t = \pr \wedge p = \pp}\\
\linelb{swing-ln7}    &\qquad\qquad\spec{(\pp \Mapsto \rightM,\TT,\pl \wedge t = \pr \wedge p = \pp)\
        \textcolor{teal}{*\ \gseg\ \G{\setminus} \pp}}\\ 
\linelb{swing-ln8}    &\qquad\qquad\spec{\gseg\ (\pp \mapsto (\rightM,[\TT,\pl]) \join \G{\setminus} \pp) \wedge
      t = \pr \wedge p = \pp}\\
\linelb{swing-ln9}  &\qquad \sspecopen{\gseg\ (\pp \mapsto (\rightM,[\TT,\pl]) \join \G{\setminus} \pp) \wedge
      t = \pr \wedge p = \pp}\\
    &\qquad \opensspec{\textcolor{teal}{\wedge \invx{\GG}{\G}{\TT}{\pp}
    \wedge \TT \in \setmarknull{\G} \wedge \G\ \pp = (\leftM, [p_l,p_r])}}\nonumber\\
\linelb{swing-ln10}  &\spec{\exists \G'\ldot \gseg\ \G' \wedge \invx{\GG}{\G'}{t}{p}}
\end{align*}}
\paragraph*{\textbf{Line~\ref{swing-ln9} implies line~\ref{swing-ln10}.}}
As for the case of POP, this step involves reformulating it as an implication,
where initial and final values of the graph, stack and nodes $t$ and $p$ are
made explicit.

\begin{align}
& \G = \pp \mapsto (\leftM, [\pl, \pr]) \join \G{\setminus}{\pp} \wedge \hbox{} \label{swing-pre1}\\
& \invpx{\GG}{\G}{\St}{\TT}{\pp} \wedge \hbox{} \label{swing-pre2}\\
& \TT \in \setmarknull{\G} \implies \hbox{}\label{swing-pre3}\\
\exists \G'\ldot & \G' = \pp \mapsto (\rightM,[\TT,\pl]) \join \G{\setminus}{\pp} \wedge \hbox{} \label{swing-post1}\\
& \invpx{\GG}{\G'}{\St}{\pr}{\pp}\label{swing-post2}
\end{align}

SWING differs from the other two operations in that the stack remains
unaltered. Furthermore, because we know $\pp \neq \nnull$,
$\unique\ (\nnull \join \St)$ and $p_0 = \llast\ (\nnull \join \St)$
it follows that $\St = \St' \join \pp$ for some sequence $\St'$.
The invariant (\ref{invbook}) $\pr \in \nodes_0\ \G'$ follows from
(\ref{swing-pre1}) and (\ref{swing-pre2}) as we know that
$\pr \in \adj\ \G$, $\goodg\ \G$ and $\nodes_0\ \G = \nodes_0\ \G'$.
 {\allowdisplaybreaks
 \begin{alignat*}{3}
   &(\mbox{\ref{invclosed}})\ \adj\ \G'=  && \text{Def. of $\G'$}\\
   &= \adj\ (\pp \mapsto (\rightM, [\TT, \pl]) \join \G{\setminus}{\pp}) && \text{Lem. \ref{sinks-dist} (distrib.)}\\
   &= \adj\ (\pp \mapsto (\rightM, [\TT, \pl])) \cup \adj\ (\G{\setminus}{\pp}) && \text{Def. of $\adj$ (\ref{sinks-def})} \\ 
   &= \{\TT,\pl\} \cup \adj\ (\G{\setminus}{\pp}) && \text{Set inclusion} \\
   &\subseteq \{\TT,\pl,\pr\} \cup \adj\ (\G{\setminus}{\pp}) && \text{Assump.~(\ref{swing-pre2}) \& (\ref{swing-pre3})} \\
             &\subseteq \nodes_0\ \G && \text{Def. of $\nodes$} \\
   &= \nodes_0\ \G' &&  \\[3mm]
    &(\mbox{\ref{invstack}})\ nodes\ {\G'\filter{\leftM,\rightM}} = &&\text{Def. of $\G'$} \\
    &=\nodes\ {(\pp \mapsto (\rightM, [\TT, \pl]) \join \G{\setminus}{\pp})\filter{\leftM,\rightM}} &&\text{Lem. \ref{nodes-dist} (distrib.)} \\
   &= \nodes\ {(\pp \mapsto (\rightM, [\TT, \pl]))\filter{\leftM,\rightM}}\ \dotcup\ \nodes\ {(\G{\setminus}{\pp})\filter{\leftM,\rightM}} && \text{Def. of \textit{filter} (\ref{filterc-def})}\\
   &= \nodes\ (\pp \mapsto (\rightM, [\TT, \pl]))\ \dotcup\ \nodes\ {(\G{\setminus}{\pp})\filter{\leftM,\rightM}} && \text{Def. of \textit{nodes}}\\
   &= \{\pp\}\ \dotcup\ \nodes\ {(\G{\setminus}{\pp})\filter{\leftM,\rightM}} && \text{Def. of \textit{nodes}} \\
   &= \nodes\ (\pp \mapsto (\leftM, [\pl, \pr]))\ \dotcup\ \nodes\ {(\G{\setminus}{\pp})\filter{\leftM,\rightM}} && \text{Def. of \textit{filter} (\ref{filterc-def})} \\
    &= \nodes\ {(\pp \mapsto (\leftM, [\pl, \pr]))\filter{\leftM,\rightM}}\ \dotcup\ \nodes\ {(\G{\setminus}{\pp})\filter{\leftM,\rightM}} && \text{Lem. \ref{nodes-dist} (distrib.)} \\
    &=\nodes\ {(\pp \mapsto (\leftM, [\pl, \pr]) \join \G{\setminus}{\pp})\filter{\leftM,\rightM}} && \text{Assump.~(\ref{swing-pre2})} \\
    &=\St && \\[3mm]
   &(\mbox{\ref{invinvert}})\ \sconsecx{\G'}{\St} = && \text{Def. of $\G'$} \\
             &= \sconsecx{(\pp \mapsto (\rightM, [\TT, \pl]) \join \G{\setminus}{\pp})}{\St} && \text{Lem. \ref{map-dist} (distrib.)} \\
             &= \sconsecx{(\pp \mapsto (\rightM, [\TT, \pl]))}{\St} \join \sconsecx{\G{\setminus}{\pp}}{\St} \hspace*{5mm}&& \text{Def. of $\sconsec$ (\ref{invert-def})} \\
             &= \pp \mapsto [\TT, \pprev\ (\nnull\join\St)\ \pp] \join \sconsecx{\G{\setminus}{\pp}}{\St}  && \text{Assump.~(\ref{swing-pre2})} \\
             &= \pp \mapsto [\TT, \pl] \join \erasure{\G{\setminus}{\pp}} && \text{Def. of \textit{erasure} (\ref{erasure-def})} \\
             &= \erasure{\pp \mapsto (\rightM, [\TT, \pl])} \join \erasure{\G{\setminus}{\pp}} && \text{Lem. \ref{erasure-dist} (distrib.)} \\
             &= \erasure{\pp \mapsto (\rightM, [\TT, \pl]) \join \G{\setminus}{\pp}} && \text{Def. of $\G'$} \\
             &= \erasure{\G'} & \\[3mm]
    &(\mbox{\ref{invrestore}})\ \gdiffx{\G'}{\St}{\pr}= && \text{Def. of $\G'$}\\
              &= \gdiffx{(\pp \mapsto (\rightM, [\TT, \pl]) \join \G{\setminus}{\pp})}{\St}{\pr} &&\text{Lem. \ref{map-dist} (distrib.)}\\
              &= \gdiffx{(\pp \mapsto (\rightM, [\TT, \pl]))}{\St}{\pr} \join \gdiffx{\G{\setminus}{\pp}}{\St}{\pr}\quad&&\text{Def. of $\gdiff$ (\ref{restore-def})}\\
   &= \pp \mapsto [\TT, \nnext\ (\St \join \pr)\ \pp] \join \gdiffx{\G{\setminus}{\pp}}{\St}{\pr}&&\text{Assump.~(\ref{swing-pre2})}\\
   &= \pp \mapsto [\TT, \pr] \join \gdiffx{\G{\setminus}{\pp}}{\St}{\pr}&&\text{Lem.~\ref{eq:gdiff}}\\
              &= \pp \mapsto [\TT, \pr] \join \gdiffx{\G{\setminus}{\pp}}{\St}{\TT}&&\text{Def. of $\gdiff$ (\ref{restore-def})}\\
    &= \gdiffx{(\pp \mapsto (\leftM, [\pl, \pr]))}{\St}{\TT} \join\ \gdiffx{\G{\setminus}{\pp}}{\St}{\TT}&&\text{Lem. \ref{map-dist} (distrib.)}\\ 
   &= \gdiffx{(\pp \mapsto (\leftM, [\pl, \pr]) \join \G{\setminus}{\pp})}{\St}{\TT} &&\text{Def. of $\G$}\\
   &= \gdiffx{\G}{\St}{\TT} &&\text{Assump.~(\ref{swing-pre2})}\\
              &= \erasure{\GG} &&\\[3mm]
    &(\mbox{\ref{invreach}})\ \nodes\ {\G'\filter{\notM}} = && \text{Def. of $\G'$}\\
    &= \nodes\ {(\pp \mapsto (\rightM, [\TT,\pl]) \join \G{\setminus}{\pp})\filter{\notM}} &&\text{Lem.~\ref{filterc-dist} (distrib.)} \\
    &= \nodes\ ({(\pp \mapsto (\rightM, [\TT,\pl]))\filter{\notM}} \join {(\G{\setminus}{\pp})\filter{\notM}}) &&\text{Lem.~\ref{nodes-dist} (distrib.)} \\
    &= \nodes\ {(\pp \mapsto (\rightM, [\TT,\pl]))\filter{\notM}}\ \dotcup\ \nodes\ {(\G{\setminus}{\pp})\filter{\notM}} &&\text{Def. of \textit{filter} (\ref{filterc-def})} \\
    &= \nodes\ e\ \dotcup\ \nodes\ {(\G{\setminus}{\pp})\filter{\notM}} &&\text{Def. of \textit{filter} (\ref{filterc-def})} \\
    &= \nodes\ {(\pp \mapsto (\leftM, [\pl, \pr]))\filter{\notM}}\ \dotcup\ \nodes\ {(\G{\setminus}{\pp})\filter{\notM}} &&\text{Lem.~\ref{nodes-dist} (distrib.)} \\
    &= \nodes\ ({(\pp \mapsto (\leftM, [\pl, \pr]))\filter{\notM}} \join {(\G{\setminus}{\pp})\filter{\notM}}) &&\text{Lem.~\ref{filterc-dist} (distrib.)} \\
    &= \nodes\ {(\pp \mapsto (\leftM, [\pl, \pr]) \join \G{\setminus}{\pp})\filter{\notM}} &&\text{Def. of $\G$} \\
     &= \nodes\ {\G\filter{\notM}} && \text{Assump.~(\ref{swing-pre2})}\\
     &\subseteq {\textstyle \bigcup\limits_{\St'\cdot \pp}} (\reachone{\G\filter{\notM}}\ \circ\ \G_r) \cup \reachtwo{\G\filter{\notM}}{\TT} && \text{$\TT \notin \nodes\ {\G\filter{\notM}}$} \\
     &= {\textstyle\bigcup\limits_{\St'\cdot \pp}} (\reachone{\G\filter{\notM}}\ \circ\ \G_r)&& \text{Comm.\&Assoc. of $\cup$}\\
     &= {\textstyle\bigcup\limits_{\St'}}\ (\reachone{\G\filter{\notM}}\ \circ\ \G_r) \cup\reachtwo{\G\filter{\notM}}{\pr} && \text{${\G\filter{\notM}} = {\G'\filter{\notM}}$}\\
     &= {\textstyle\bigcup\limits_{\St'}}\ (\reachone{\G'\filter{\notM}}\ \circ\ \G_r) \cup\reachtwo{\G'\filter{\notM}}{\pr} && \text{$\G_r = \G_r'$ on $\St'$}\\
     &= {\textstyle\bigcup\limits_{\St'}}\ (\reachone{\G'\filter{\notM}}\ \circ\ \G'_r) \cup\reachtwo{\G'\filter{\notM}}{\pr} && \text{$\G'_r\ \pp \notin \nodes\ {\G'\filter{\notM}}$}\\
     &= {\textstyle\bigcup\limits_{\St'\cdot \pp}}\ (\reachone{\G'\filter{\notM}}\ \circ\ \G'_r) \cup\reachtwo{\G'\filter{\notM}}{\pr} &&
 \end{alignat*}}

\newcommand{\tl}{t_l}
\newcommand{\tr}{t_r}

\section{Proof for PUSH}
The pre- and postcondition for PUSH derive from
lines~\ref{sw-ln23} and~\ref{sw-ln25} of Fig.~\ref{fig:sw}.
\begin{equation}\label{push-spec}
\begin{array}[c]{c} 
\spec{\exists \G\ldot \gseg\ \G \wedge \invx{\GG}{\G}{t}{p} \wedge t \notin \setmarknull{\G}}\\
\textrm{PUSH}  \\
\spec{\exists \G'\ldot\ \gseg\ \G' \wedge \invx{\GG}{\G'}{t}{p}}
\end{array}
\end{equation}
The precondition says that the heap implements a well-formed graph
$\G$, that satisfies the invariant and $t$ is an unmarked node
different from $\nnull$. The postcondition asserts that the heap
represents a new graph $\G'$ that satisfies the invariant for the
updated values of $t$ and $p$.

\makeatletter
\newcounter{codepush}
\renewcommand{\lineno}{\stepcounter{codepush}\textsc{\thecodepush}.\quad}
\renewcommand{\linelb}[1]{{\refstepcounter{codepush}\ltx@label{#1}\textsc{\thecodepush}.\quad}}
\makeatother
{\allowdisplaybreaks
  \begin{align*}
\linelb{push-ln1}    & \spec{\exists \G\ldot \gseg\ \G \wedge \invx{\GG}{\G}{t}{p} \wedge t \notin \setmarknull{\G}}\\
\linelb{push-ln2}  &\qquad \sspecopen{\gseg\ \G \wedge t = \TT \wedge p = \pp}\\
                  &\qquad \opensspec{\textcolor{teal}{\wedge \invx{\GG}{\G}{\TT}{\pp}
      \wedge \G\ \TT = (\notM, [\tl,\tr])}} \nonumber\\
\linelb{push-ln3}  & \qquad\qquad \spec{\gseg\ (\TT \mapsto (\notM, [\tl,\tr]) \join \G{\setminus} \TT) 
\wedge t = \TT \wedge p = \pp}\\
\linelb{push-ln4} & \qquad\qquad \spec{(\TT \Mapsto \notM, \tl, \tr \wedge t = \TT \wedge p = \pp)\
       \textcolor{teal}{*\ \gseg\ \G{\setminus} \TT}}\\
\linelb{push-ln5}   &\qquad\qquad\qquad\spec{\TT \Mapsto \notM, \tl, \tr \wedge t = \TT \wedge p = \pp}\\
    &\qquad\qquad\qquad\ \tmp \assign \deref{t.l};\ t.l \mutate p;\ t.m \mutate \leftM;\
    p \assign t;\ t \assign \tmp;  \nonumber\\
\linelb{push-ln6}   &\qquad\qquad\qquad\spec{\TT \Mapsto \leftM,\pp,\tr \wedge t = \tl \wedge p = \TT}\\
\linelb{push-ln7}    &\qquad\qquad\spec{(\TT \Mapsto \leftM,\pp,\tr \wedge t = \tl \wedge p = \TT)\
        \textcolor{teal}{*\ \gseg\ \G{\setminus} \TT}}\\ 
\linelb{push-ln8}    &\qquad\qquad\spec{\gseg\ (\TT \mapsto (\leftM,[\pp,\tr]) \join \G{\setminus} \TT) \wedge t = \tl \wedge p = \TT}\\
\linelb{push-ln9}  &\qquad \sspecopen{\gseg\ (\TT \mapsto (\leftM,[\pp,\tr]) \join \G{\setminus} \TT) \wedge t = \tl \wedge p = \TT}\\
    &\qquad \opensspec{\textcolor{teal}{\wedge \invx{\GG}{\G}{\TT}{\pp}
    \wedge \G\ \TT = (\notM, [\tl,\tr])}}\nonumber\\
\linelb{push-ln10}  &\spec{\exists \G'\ldot \gseg\ \G' \wedge \invx{\GG}{\G'}{t}{p}}
\end{align*}}

\paragraph*{\textbf{Line~\ref{push-ln9} implies line~\ref{push-ln10}.}}
Proving this last step corresponds to proving the following
implication. 
\begin{align}
& \G = \TT \mapsto (\notM, [\tl, \tr]) \join \G{\setminus}{\TT} \wedge \hbox{} \label{push-pre1}\\
& \invpx{\GG}{\G}{\St}{\TT}{\pp} \implies \hbox{}\label{push-pre2}\\
\exists \G'\ldot & \G' = \TT \mapsto (\leftM,[\pp,\tr]) \join \G{\setminus}{\TT} \wedge \hbox{} \label{push-post1}\\
& \invpx{\GG}{\G'}{(\St \join \TT)}{\tl}{\TT}\label{push-post2}
\end{align}

Proving invariant (\ref{invbook}) requires showing that (\ref{push-pre1}) and
(\ref{push-pre2}) imply $\unique\ (\nnull \join \St \join \TT)$,
$\TT = \llast\ (\nnull \join \St \join \TT)$ and
$\tl \in \nodes_0\ \G'$.
From $\G_m\ \TT = \notM$ and $\nodes\ {\G\filter{\leftM,\rightM}} = \St$ it
follows $\unique\ (\nnull \join \St \join \TT)$.
By \ref{push-pre1} it follows $\tl \in \adj\ \G$. Then $\goodg\ \G$ implies
that $\tl \in \nodes_0\ \G$. And finally $\nodes_0\ \G = \nodes_0\ \G'$ so
$\tl \in \nodes_0\ \G'$. 

 {\allowdisplaybreaks
 \begin{alignat*}{3}
  &(\mbox{\ref{invclosed}})\ \adj\ \G'=  && \text{Def. of $\G'$}\\
  &= \adj\ (\TT \mapsto (\leftM, [\pp, \tr]) \join \G{\setminus}{\TT}) && \text{Lem. \ref{sinks-dist} (distrib.)}\\
  &= \adj\ (\TT \mapsto (\leftM, [\pp, \tr])) \cup \adj\ (\G{\setminus}{\TT}) && \text{Def. of $\adj$ (\ref{sinks-def})} \\ 
  &= \{\pp,\tr\} \cup \adj\ (\G{\setminus}{\TT}) && \text{Set inclusion} \\
  &\subseteq \{\pp,\tl,\tr\} \cup \adj\ (\G{\setminus}{\TT}) && \text{Assump.~(\ref{push-pre2})} \\
            &\subseteq \nodes_0\ \G && \text{Def. of $\nodes$} \\
  &= \nodes_0\ \G' &&  \\[3mm]
  &(\mbox{\ref{invstack}})\ nodes\ {\G'\filter{\leftM,\rightM}} = &&\text{Def. of $\G'$} \\
  &=\nodes\ {(\TT \mapsto (\leftM, [\pp, \tr]) \join \G{\setminus}{\TT})\filter{\leftM,\rightM}} &&\text{Lem. \ref{filterc-dist} (distrib.)} \\
  &=\nodes\ ({(\TT \mapsto (\leftM, [\pp, \tr]))\filter{\leftM,\rightM}} \join {(\G{\setminus}{\TT})\filter{\leftM,\rightM}}) &&\text{Def. of \textit{filter} (\ref{filterc-def})} \\
   &=\nodes\ (\TT \mapsto (\leftM, [\pp, \tr]) \join {(\G{\setminus}{\TT})\filter{\leftM,\rightM}}) &&\text{Lem.~\ref{nodes-dist} (distrib.)} \\
  &= \nodes\ {(\TT \mapsto (\leftM, [\pp, \tr]))\ \dotcup\ \nodes\ (\G{\setminus}{\TT})\filter{\leftM,\rightM}} && \text{Def. of $\nodes$}\\
  &= \{\TT\}\ \dotcup\ \nodes\ {(\G{\setminus}{\TT})\filter{\leftM,\rightM}} && \text{Def. of \textit{filter} (\ref{filterc-def})} \\
  &= \{\TT\}\ \dotcup\ \nodes\ {(\TT \mapsto (\notM,[\tl,\tr]) \join \G{\setminus}{\TT})\filter{\leftM,\rightM}} && \text{Def. of $\G$} \\
   &= \{\TT\}\ \dotcup\ \nodes\ {\G\filter{\leftM,\rightM}} && \text{Assump.~(\ref{push-pre2})} \\
  &= \{\TT\}\ \dotcup\ \St && \text{Set equality} \\
  &= (\St \join \TT) && \\[3mm]
  &(\mbox{\ref{invinvert}})\ \sconsecx{\G'}{(\St \join \TT)} = && \text{Def. of $\G'$} \\
            &= \sconsecx{(\TT \mapsto (\leftM, [\pp, \tr]) \join \G{\setminus}{\TT})}{(\St \join \TT)} && \text{Lem. \ref{map-dist} (distrib.)} \\
            &= \sconsecx{(\TT \mapsto (\leftM, [\pp, \tr]))}{(\St \join \TT)} \join \sconsecx{\G{\setminus}{\TT}}{(\St \join \TT)} \hspace*{5mm}&& \text{Def. of $\sconsec$ (\ref{invert-def})} \\
  &= \TT \mapsto [\pprev\ (\nnull \join \St \join \TT)\ \TT, \tr] \join \sconsecx{\G{\setminus}{\TT}}{(\St \join \TT)} && \text{Assump.~(\ref{push-pre2})} \\
  &= \TT \mapsto [\pp, \tr] \join \sconsecx{\G{\setminus}{\TT}}{(\St \join \TT)} && \text{Lem.~\ref{eq:invert}} \\
            &= \TT \mapsto [\pp, \tr] \join \sconsecx{\G{\setminus}{\TT}}{\St} && \text{Def. of \textit{erasure} (\ref{erasure-def})} \\
            &= \erasure{\TT \mapsto (\leftM, [\pp, \tr])} \join \erasure{\G{\setminus}{\TT}} && \text{Lem. \ref{erasure-dist} (distrib.)} \\
            &= \erasure{\TT \mapsto (\leftM, [\pp, \tr]) \join \G{\setminus}{\TT}} && \text{Def. of $\G'$} \\
            &= \erasure{\G'} && \\[3mm]
  &(\mbox{\ref{invrestore}})\ \gdiffx{\G'}{(\St \join \TT)}{\tl}= && \text{Def. of $\G'$}\\
            &= \gdiffx{(\TT \mapsto (\leftM, [\pp, \tr]) \join \G{\setminus}{\TT})}{(\St \join \TT)}{\tl} &&\text{Lem. \ref{map-dist} (distrib.)}\\
  &= \gdiffx{(\TT \mapsto (\leftM, [\pp, \tr]))}{(\St \join \TT)}{\tl} && \\
   &\qquad \join \gdiffx{\G{\setminus}{\TT}}{(\St \join \TT)}{\tl}&&\text{Def. of $\gdiff$ (\ref{restore-def})}\\
            &= \TT \mapsto [\nnext\ (\St \join \TT \join \tl)\ \TT, \tr] \join \gdiffx{\G{\setminus}{\TT}}{(\St \join \TT)}{\tl}&&\text{Assump.~(\ref{push-pre2})}\\
  &= \TT \mapsto [\tl, \tr] \join \gdiffx{\G{\setminus}{\TT}}{(\St \join \TT)}{\tl}&&\text{Lem.~\ref{eq:gdiff}}\\
            &= \TT \mapsto [\tl, \tr] \join \gdiffx{\G{\setminus}{\TT}}{\St}{\TT}&&\text{Def. of $\gdiff$ (\ref{restore-def})}\\
  &= \gdiffx{(\TT \mapsto (\notM, [\tl, \tr]))}{\St}{\TT} &&\\
   &\qquad \join\ \gdiffx{\G{\setminus}{\TT}}{\St}{\TT}&&\text{Lem. \ref{map-dist} (distrib.)}\\ 
  &= \gdiffx{(\TT \mapsto (\notM, [\tl, \tr]) \join \G{\setminus}{\TT})}{\St}{\TT} &&\text{Def. of $\G$}\\
   &= \gdiffx{\G}{\St}{\TT} &&\text{Assump.~(\ref{push-pre2})}\\
            &= \erasure{\GG} &&\\[3mm]
   &(\mbox{\ref{invreach}})\ \nodes\ {\G'\filter{\notM}} = && \text{Def. of $\G'$}\\
   &= \nodes\ {(\TT \mapsto (\leftM, [\pp,\tr]) \join \G{\setminus}{\TT})\filter{\notM}} &&\text{Lem.~\ref{filterc-dist} (distrib.)} \\
    &= \nodes\ ({(\TT \mapsto (\leftM, [\pp,\tr]))\filter{\notM}} \join {(\G{\setminus}{\TT})\filter{\notM}}) &&\text{Lem.~\ref{nodes-dist} (distrib.)} \\
    &= \nodes\ {(\TT \mapsto (\leftM, [\pp,\tr]))\filter{\notM}}\ \dotcup\ \nodes\ {(\G{\setminus}{\TT})\filter{\notM}} &&\text{Def. of \textit{filter} (\ref{filterc-def})} \\
    &= \nodes\ e\ \dotcup\ \nodes\ {(\G{\setminus}{\TT})\filter{\notM}} &&\text{Def. of $\nodes$} \\
    &= \nodes\ {(\G{\setminus}{\TT})\filter{\notM}} &&\text{Set difference} \\
   &= (\{\TT\}\ \dotcup\ \nodes\ {(\G{\setminus}{\TT})\filter{\notM}}){\setminus}{\TT} &&\text{Def. of $\nodes$} \\
   &= (\nodes\ (\TT \mapsto (\notM, [\tl,\tr]))\ \dotcup\ \nodes\ {(\G{\setminus}{\TT})\filter{\notM}}){\setminus}{\TT} &&\text{Def. of \textit{filter} (\ref{filterc-def})} \\
   &= (\nodes\ {(\TT \mapsto (\notM, [\tl,\tr]))\filter{\notM}}\ \dotcup\ \nodes\ {(\G{\setminus}{\TT})\filter{\notM}}){\setminus}{\TT} &&\text{Lem.~\ref{nodes-dist} (distrib.)} \\
   &= (\nodes\ ({\TT \mapsto (\notM, [\tl,\tr]) \join \G{\setminus}{\TT})\filter{\notM}}){\setminus}{\TT} &&\text{Def. of $\G$} \\
   &= (\nodes\ {\G\filter{\notM}}){\setminus}{\TT}&& \text{Assump.~(\ref{push-pre2})}\\
    &\subseteq ({\textstyle \bigcup\limits_{\St}} (\reachone{\G\filter{\notM}}\ \circ\ \G_r) \cup
     \reachtwo{\G\filter{\notM}}{\TT}){\setminus}{\TT} && \text{Lem.~\ref{notin-reach-single} \& \ref{in-reach-single}} \\
   &= ({\textstyle \bigcup\limits_{\St}} (\reachone{(\G\filter{\notM}){\setminus}{\TT}}\ \circ\ \G_r) \cup \reachtwo{\G\filter{\notM}}{\TT}){\setminus}{\TT} && \text{Def. of $\rreach$ (\ref{def:reachvia})} \\
   &= ({\textstyle \bigcup\limits_{\St}} (\reachone{(\G\filter{\notM}){\setminus}{\TT}}\ \circ\ \G_r) \cup \{\TT\} \cup {\textstyle \bigcup\limits_{\{\tl,\tr\}}} \reachone{(\G\filter{\notM})}){\setminus}{\TT} && \text{Set difference} \\
   &= {\textstyle \bigcup\limits_{\St}} (\reachone{(\G\filter{\notM}){\setminus}{\TT}}\ \circ\ \G_r) \cup
     {\textstyle \bigcup\limits_{\{\tl,\tr\}}} \reachone{(\G\filter{\notM}){\setminus}{\TT}} && \text{$(\G\filter{\notM}){\setminus{\TT}} = \G'\filter{\notM}$} \\
   &= {\textstyle \bigcup\limits_{\St}} (\reachone{\G'\filter{\notM}})\ \circ\ \G_r) \cup
     {\textstyle \bigcup\limits_{\{\tl,\tr\}}} \reachone{\G'\filter{\notM}} && \text{$\G_r = \G_r'$ on $(\St \join \TT$)} \\
   &= {\textstyle \bigcup\limits_{(\St \join \TT)}} (\reachone{\G'\filter{\notM}}\ \circ\ \G'_r) \cup
   \reachtwo{\G'\filter{\notM}}{\tl} &&
\end{alignat*} 
}

\section{Union-Find Data Structure}
\newcommand{\forest}[1]{\textit{partitioned}\ #1}
\newcommand{\closed}{\textit{closed}}
\newcommand{\preacyclic}{\textit{preacyclic}}
\newcommand{\bireachU}{\cycles}
\newcommand{\danglings}{\textit{dangls}}
\newcommand{\cycles}{\textit{cycles}}

\subsection{Non-spatial definitions}

\begin{align*}
  \cycles\ \G \eqdef & \{x\ |\ x \in {\textstyle \bigcup\limits_{y \in \Ga\ x}} \rreach\ \G\ y \} \\
  \preacyclic\ \G \eqdef & \cycles\ \G \subseteq \selfloopss{\G} \\
  \danglings\ \G \eqdef & (\adj\ \G){\setminus{\nodes\ \G}}
\end{align*}
The function $\cycles$ computes the set of nodes that constitute a
cycle in the graph. More precisely, it identifies nodes that are
reachable from one of their adjacent nodes. This avoids the trivial
case where every node is reachable from itself.  $\textit{Loops}$ are
cycles of size one—that is, nodes that explicitly include themselves
in their adjacency list. Then a $\preacyclic$ graph is one where every
cycle is of size one. Loops are given a special status as they are
cycles that do not break under decomposition.  $\danglings\ \G$
selects the dangling nodes of a graph, those that are pointed at by a
node in the graph but are not themselves in the graph. Notice that
$\nnull$ may be in this set.

\begin{lemma}[Union-find abstractions]\label{lemma:ufabstractions}
\quad
\begin{enumerate}[ref=\ref{lemma:ufabs} (\arabic*)]
\item $\danglings\ (\G_1 \join \G_2) = (\danglings\
    \G_1){\setminus{\nodes\ \G_2}} \cup (\danglings\
    \G_2){\setminus{\nodes\ \G_1}}$ \label{nodes-dist-ufabs} 
\item $\selfloopss\ \G \subseteq \bireachU\ \G \subseteq \nodes\ \G$ \label{cycles-in-nodes}
\item $\danglings\ \G \cap \nodes\ \G = \emptyset$ \label{danglings-nodes-disj}
\end{enumerate}
\end{lemma}

\begin{lemma}[Summit characterization]\label{summit-characterization}
  Let $\G$ be a graph and $x \in \nodes\ \G$. Then,
  $z \in \ccenter{\G}{x}$ iff there exists a path from $x$ to $y$ in
  $\G$ such that $z$ is a child of $y$, and either
  $z \notin \nodes\ \G$ (i.e., the edge from $y$ to $z$ is dangling)
  or $z$ is already in the path from $x$ to $y$.
\end{lemma}

\begin{lemma}[Summits]\label{lemma:ufabs}
\quad
\begin{enumerate}[ref=\ref{lemma:ufabs} (\arabic*)]
\item $\ccenters{\G} = \bireachU\ \G\ \dotcup\ \danglings\ \G $ \label{summits-char} 
\item If $\ccenters{\G_1} = \ccenters{\G_2}$ then $\ccenters{(\G_1 \join \G_2)} = \ccenters{\G_1}$ \label{summits-eq}
\item If $x \in \selfloops\ \G$ then $\ccenters{\G{\setminus{x}}} \subseteq \ccenters{\G} $ \label{summits-free}
\end{enumerate}
\end{lemma}

\begin{lemma}[Summits of inverted forest]\label{lemma:summits}
  Let $\G = (\G_1 \join \G_2)$ be an an inverted forest
  ($\ccenters{\G} \subseteq \selfloopss{\G}$) then
\begin{enumerate}[ref=\ref{lemma:summits} (\arabic*)]
\item $\adj\ \G \subseteq \nodes\ \G$ ($\goodg\ \G$) \label{summits-closed}
\item $\bireachU\ \G \subseteq \selfloopss{\G}$ ($\preacyclic\ \G$)\label{summits-preacyclic}
\item $\ccenters{(\G_1 \join \G_2)} = (\ccenters{\G_1}){\setminus{\nodes\ \G_2}}\ \dotcup\ (\ccenters{\G_2}){\setminus{\nodes\ \G_1}}$
\item $\ccenter{\G}{x} = \ccenter{\G}{(\G\ x)}$
\end{enumerate}
\end{lemma}

\begin{lemma}[Preacyclic mutation]\label{preacyclic-mutation}
  Let $\G$ be a unary graph such that $\preacyclic\ \G$, and
  $x \in \nodes\ \G$ and $y \notin \nodes\ \G$. The graph $\G'$
  obtained by modifying $x$'s successor to $y$, also satisfies
  $\preacyclic\ \G'$.
\end{lemma}

\subsection{Proof outline for NEW}
\makeatletter
\newcounter{codenew}
\renewcommand{\lineno}{\stepcounter{codenew}\textsc{\thecodenew}.\quad}
\renewcommand{\linelb}[1]{{\refstepcounter{codenew}\ltx@label{#1}\textsc{\thecodenew}.\quad}}
\makeatother
{\allowdisplaybreaks
\begin{align*}
\linelb{new-ln1} &\spec{\emp}\\
  \lineno  &p \assign \textit{alloc}\ \nnull;\\
  \linelb{new-ln3} &\spec{p \Mapsto \nnull}\\
  \lineno &p.next \mutate p;\\
  \linelb{new-ln5} &\spec{p \Mapsto p}\\
  \linelb{new-lnr} & \textbf{return}\ p \\
  \linelb{new-ln6} &\spec{p \Mapsto p \wedge \res = p}\\
  \linelb{new-ln7} &\spec{\gseg_1 (p \mapsto p) \wedge \ccenters{(p \mapsto p)} = \selfloopss{(p \mapsto p)} = \nodes\ (p \mapsto p) = \{p\} \wedge \res = p} \\
  \linelb{new-ln8} &\spec{\unfoldset{\{p\}}{p} \wedge \res = p} \\
  \linelb{new-ln9} &\spec{\set{\{\res\}}{\res}}
\end{align*}
}
The first two commands are standard separation logic, the third
command in line~\ref{new-lnr} corresponds to a value returning
function supported by Hoare Type Theory. The step from
line~\ref{new-ln6} to line~\ref{new-ln7} lifts the reasoning from the
heap to the abstract graph and establishes that the singleton graph is
indeed an inverted tree with $\res$ as its representative. Finally the
conjunct are folded into the set predicate.

\subsection{Proof outline for FIND}
\makeatletter
\newcounter{codefind}
\renewcommand{\lineno}{\stepcounter{codefind}\textsc{\thecodefind}.\quad}
\renewcommand{\linelb}[1]{{\refstepcounter{codefind}\ltx@label{#1}\textsc{\thecodefind}.\quad}}
\makeatother
{\allowdisplaybreaks
\begin{align*}
  \lineno &\spec{\set{S}{y} \wedge x \in S}\\
  \lineno &\spec{\unfoldset{S}{y} \wedge x \in S}\\
  \lineno &\quad\spec{\gseg_1\ \G \wedge
            \ccenters{\G} = \selfloopss{\G} = \{y\} \wedge nodes\ \G = S \wedge x \in S} \\
  \lineno &\quad\spec{\gseg_1 (x \mapsto \G\ x \join \G{\setminus{x}}) \wedge x \in S
  \textcolor{teal}{\hbox{}\wedge\ \ccenters{\G} = \selfloopss{\G} = \{y\} \wedge nodes\ \G = S}} \\
  \lineno &\qquad\spec{x \Mapsto \G\ x * \gseg_1\ \G{\setminus{x}} \wedge x \in S} \\
  \linelb{find-ln6}  &\qquad p \assign x.next;\\
  \lineno &\qquad\spec{x \Mapsto \G\ x * \gseg_1\ \G{\setminus{x}} \wedge 
            x \in S \wedge p = \G\ x} \\
  \lineno &\qquad\spec{\gseg_1\ (x \mapsto \G\ x \join \G{\setminus{x}}) \wedge 
            x \in S\wedge p = \G\ x} \\ 
  \linelb{find-ln9} &\qquad\spec{\gseg_1\ \G \wedge x \in S \wedge p = \G\ x} \\
  \linelb{find-ln10}  &\qquad\textrm{\textbf{while}}\ p \neq x\ \textrm{\textbf{do}}\\ 
  \lineno &\quad\qquad\spec{\gseg_1\ \G \wedge x \in S \wedge p = \G\ x \neq x} \\
  \lineno &\quad\qquad x \mutate p;\\
  \linelb{find-ln16} &\quad\qquad\spec{\gseg_1\ \G \wedge  x \in S \wedge x = p} \\
  \lineno &\quad\qquad\spec{\gseg_1 (x \mapsto \G\ x \join \G{\setminus{x}}) \wedge
            x \in S \wedge x = p} \\
  \lineno &\quad\qquad\spec{x \Mapsto \G\ x * \gseg_1\ \G{\setminus{x}} \wedge 
            x \in S \wedge x = p} \\
  \lineno &\quad\qquad p \mutate x.next;\\
  \lineno &\quad\qquad\spec{x \Mapsto \G\ x * \gseg_1\ \G{\setminus{x}} \wedge 
            x \in S \wedge p = \G\ x} \\
  \lineno &\quad\qquad\spec{\gseg_1\ (x \mapsto nx \join \G{\setminus{x}}) \wedge 
            x \in S \wedge p = \G\ x} \\ 
  \linelb{find-ln19} &\quad\qquad\spec{\gseg_1\ \G \wedge x \in S \wedge p = \G\ x} \\  
  \lineno &\qquad\textrm{\textbf{end while}}\\
  \linelb{find-ln21} &\qquad\spec{\gseg_1\ \G \wedge x \in S \wedge p = \G\ x = x} \\
  \linelb{find-ln22} &\qquad\textbf{return}\ x\\
  \linelb{find-ln25} &\qquad\spec{\gseg_1\ \G \wedge x \in S \wedge p = \G\ x = x = \res}\\
  \linelb{find-ln26} &\quad\spec{\gseg_1\ \G \wedge \res = y \textcolor{teal}{\hbox{}\wedge 
            \ccenters{\G} = \selfloopss{\G} = \{y\} \wedge nodes\ \G = S} } \\
  \lineno &\spec{\unfoldset{S}{y} \wedge \res = y} \\ 
  \lineno &\spec{\set{S}{y} \wedge \res = y}
\end{align*}}

The first five lines of the proof outline unfold the
spatial predicates, first set, then $\gseg$, to reveal the node
corresponding to the argument $x$. The non-spatial conjuncts involving
$\G$ and $y$ are framed early on, as these values remain unchanged
throughout the execution and will later be reintroduced without
modification.
The command at line~\ref{find-ln6} assigns to $p$ the successor of
$x$. In line \ref{find-ln10}, $p$ is compared to $x$. A mismatch
implies that the root node, which points to itself, has not yet been
found, and thus the loop must be entered.
Within the loop, the first command advances $x$ to its parent
node. Line~\ref{find-ln16} requires showing that the new value of $x$
remains within $S$. This holds because $x$ is guaranteed to be in
$\nodes\ \G$, given that $\goodg\ \G$ holds by
lemma~\ref{summits-closed} and $\nodes\ \G = S$.
Next, $p$ is updated to be the parent of the current $x$.
By the end of the loop in line~\ref{find-ln19}, its invariant from
line~\ref{find-ln9} is reestablished.
When the loop terminates, the invariant together with the negation of
the loop guard leads to the conclusion in line~\ref{find-ln21}: $x$ is
a node in $S$ and points to itself in $\G$.
Line~\ref{find-ln22} sets the result of the program to be x, which by
the non-spatial conjunct framed at the beginning we know must be $y$,
as it is established to be the only self-pointing node in $\G$.

\subsection{Proof outline for UNION}
\makeatletter
\newcounter{codeunion}
\renewcommand{\lineno}{\stepcounter{codeunion}\textsc{\thecodeunion}.\quad}
\renewcommand{\linelb}[1]{{\refstepcounter{codeunion}\ltx@label{#1}\textsc{\thecodeunion}.\quad}}
\makeatother
{\allowdisplaybreaks
  \begin{align*}
    \linelb{union-ln1} &\spec{\set{S_1}{x_1} * \set{S_2}{x_2}}\\
    \linelb{union-ln2} &\quad\sspecopen{\gseg_1\ \G_1 \textcolor{teal}{\hbox{}* \gseg_1\ \G_2}}\\
    \lineno &\quad\opensspecopen{\textcolor{teal}{\hbox{}\wedge \ccenters{\G_1} = \selfloopss{\G_1} = \{x_1\} \wedge nodes\ \G_1 = S_1}}\\
    \lineno &\quad\opensspec{\textcolor{teal}{\hbox{}\wedge \ccenters{\G_2} = \selfloopss{\G_2} = \{x_2\} \wedge nodes\ \G_2 = S_2}}\\
    \linelb{union-ln5} &\qquad\spec{\gseg_1 (x_1 \mapsto x_1 \join \G_1{\setminus{x_1}})}\\
    \linelb{union-ln6} &\qquad\spec{x_1 \Mapsto x_1 * \gseg_1 (\G_1{\setminus{x_1}})}\\
    \linelb{union-ln7} &\qquad x_1.next \mutate x_2;\\
    \lineno &\qquad\spec{x_1 \Mapsto x_2 * \gseg_1 (\G_1{\setminus{x_1}})}\\
    \lineno &\qquad \textbf{return}\ x_2 \\
    \lineno &\qquad\spec{x_1 \Mapsto x_2 * \gseg_1 (\G_1{\setminus{x_1}}) \wedge \res = x_2}\\
    \linelb{union-ln11} &\qquad\spec{\gseg_1 (x_1 \mapsto x_2 \join \G_1{\setminus{x_1}}) \wedge \res = x_2}\\
    \linelb{union-ln12} &\quad\sspecopen{\gseg_1 (x_1 \mapsto x_2 \join \G_1{\setminus{x_1}}) \wedge \res = x_2 \textcolor{teal}{\hbox{}* \gseg_1\ \G_2}}\\
    \lineno &\quad\opensspecopen{\textcolor{teal}{\hbox{}\wedge \ccenters{\G_1} = \selfloopss{\G_1} = \{x_1\} \wedge nodes\ \G_1 = S_1}}\\
    \linelb{union-ln14} &\quad\opensspec{\textcolor{teal}{\hbox{}\wedge \ccenters{\G_2} = \selfloopss{\G_2} = \{x_2\} \wedge nodes\ \G_2 = S_2}}\\
    \linelb{union-ln15} &\quad\sspecopen{\gseg_1 (x_1 \mapsto x_2 \join \G_1{\setminus{x_1}}) * \gseg_1\ \G_2 \wedge \res = x_2}\\
    \lineno &\quad\opensspecopen{\hbox{}\wedge \ccenters{(x_1 \mapsto x_2 \join \G_1{\setminus{x_1}})} = \{x_2\}  }\\
    \lineno &\quad\opensspecopen{\hbox{}\wedge \selfloopss{(x_1 \mapsto x_2 \join \G_1{\setminus{x_1}})} = \emptyset } \\
    \lineno &\quad\opensspecopen{\hbox{}\wedge nodes\ (x_1 \mapsto x_2 \join \G_1{\setminus{x_1}}) = S_1}\\
    \lineno &\quad\opensspec{\hbox{}\wedge \ccenters{\G_2} = \selfloopss{\G_2} = \{x_2\} \wedge nodes\ \G_2 = S_2}\\
    \linelb{union-ln20} &\quad\sspecopen{\gseg_1 (x_1 \mapsto x_2 \join \G_1{\setminus{x_1}} \join \G_2) \wedge \res = x_2}\\
    \lineno &\quad\opensspecopen{\hbox{}\wedge \ccenters{(x_1 \mapsto x_2 \join \G_1{\setminus{x_1}} \join \G_2)} = \{x_2\}  }\\
    \lineno &\quad\opensspecopen{\hbox{}\wedge \selfloopss{(x_1 \mapsto x_2 \join \G_1{\setminus{x_1}} \join \G_2)} = \{x_2\} } \\
    \lineno &\quad\opensspec{\hbox{}\wedge nodes\ (x_1 \mapsto x_2 \join \G_1{\setminus{x_1}} \join \G_2) = S_1 \join S_2}\\    
    \linelb{union-ln24} &\spec{\set{(S_1 \join S_2)}{\res} \wedge \res \in \{x_1, x_2\}}
\end{align*}}

Similar to the FIND proof, the initial lines of the proof outline
(lines~\ref{union-ln1}-\ref{union-ln6}) unfold the spatial predicates
isolating the node corresponding to the argument $x_1$.
Line~\ref{union-ln2} unfolds the set predicate, instantiates the
existentially quantified graphs as $\G_1$ and $\G_2$ and frames out 
$\gseg_1\ \G_2$ as well as the non-spatial facts about the initial
disjoint graphs.
Line~\ref{union-ln5} rewrites by the equality
$\G_1 = x_1 \mapsto x_1 \join \G_1{\setminus{x_1}}$ which follows from
$x_1 \in \selfloopss{\G_1}$.
The command in line~\ref{union-ln7} makes $x_1$ point to $x_2$. The
second and final command sets $x_2$ as the result of the program.
After the final command, in line~\ref{union-ln11} reasoning is lifted
from the heap level and in line~\ref{union-ln12} the spatial conjuncts
initially framed are reintroduced.
In line~\ref{union-ln15} the abstractions that constitute the
specification are recomputed for the mutated graph
$x_1 \mapsto x_2 \join \G_1{\setminus{x_1}}$, exploiting their
distributivity. Before applying the distributivity of
$\textit{summits}$, preacyclicity is proved for the mutated graph by
application of lemma~\ref{preacyclic-mutation}. 
\[
\begin{array}[c]{rl>{\ \ }l<{}}
&\hspace{-5mm} \ccenters{(x_1 \mapsto x_2 \join \G_1{\setminus{x_1}})} = & \\
               &= (\ccenters{(x_1 \mapsto x_2)}){\setminus{\nodes\ (\G_1{\setminus{x_1}})}} & \text{Distrib. of \textit{summits}}\\
               &\ \cup\  (\ccenters{\G_1{\setminus{x_1}}}){\setminus{\nodes\ (x_1 \mapsto x_2)}}  &  \\
               &= \{x_2\}{\setminus{\nodes\ (\G_1{\setminus{x_1}})}} \cup
                 (\ccenters{\G_1{\setminus{x_1}}}){\setminus{x_1}}  & \text{Def. of \textit{summits} and $\nodes$}  \\
               &= \{x_2\} \cup (\ccenters{\G_1{\setminus{x_1}}}){\setminus{x_1}}  & \text{$x_2 \in \nodes\ \G_2$ and}  \\
  & &\ \nodes\ \G_1{\setminus{x_1}} \cap \nodes\ \G_2 = \emptyset \\
&= \{x_2\} & \text{Lem.~\ref{summits-free} and Assump. in lines~\ref{union-ln12}-\ref{union-ln14}}\\[3mm]
  &\hspace{-5mm} \selfloopss{(x_1 \mapsto x_2 \join \G_1{\setminus{x_1}})} = & \\
&= \selfloopss{(x_1 \mapsto x_2)}\ \dotcup\ \selfloopss{(\G_1{\setminus{x_1}})}  & \text{Distrib. of $\selfloops$}  \\
  &= \selfloopss{(\G_1{\setminus{x_1}})}  & \text{Def. of $\selfloops$} \\
&= \emptyset & \text{Assump. in lines~\ref{union-ln12}-\ref{union-ln14}} \\[3mm]
  &\hspace{-5mm} \nodes\ (x_1 \mapsto x_2 \join \G_1{\setminus{x_1}}) = & \\
&= \nodes\ (x_1 \mapsto x_2)\ \dotcup\ \nodes\ (\G_1{\setminus{x_1}})  & \text{Distrib. of $\nodes$} \\
&= \{x\}\ \dotcup\ \nodes\ ({\G_1{\setminus{x_1}}})  & \text{Def. of $\nodes$} \\
&= \nodes\ (x_1 \mapsto x_1)\ \dotcup\ \nodes\ (\G_1{\setminus{x_1}})  & \text{Def. of $\nodes$}  \\
  &= \nodes\ (x_1 \mapsto x_1 \join \G_1{\setminus{x_1}})  & \text{Distrib. of $\nodes$} \\
               &= S_1 & \text{Assump. in lines~\ref{union-ln12}-\ref{union-ln14}}
\end{array} 
\]

The last step in the proof outline starting in
line~\ref{union-ln20} is joining both subgraphs
$x_1 \mapsto x_2 \join \G_1{\setminus{x_1}}$ and $\G_2$. Both
$\selfloopss$ and $\nodes$ use plain distributivity. In contrast, for
\textit{summits}, instead of proving that the joined graph is
preacyclic in order to apply distributivity, we exploit the fact that
both subgraphs have the same \textit{summits} and apply
lemma~\ref{summits-eq} directly.
\[
  \begin{array}[c]{rl>{\ \ }l<{}}
  &\hspace{-5mm} \selfloopss{(x_1 \mapsto x_2 \join \G_1{\setminus{x_1}} \join \G_2)} = & \\
&= \selfloopss{(x_1 \mapsto x_2 \join \G_1{\setminus{x_1}})}\ \dotcup\ \selfloopss{\G_2}  & \text{Distrib. of $\selfloops$}  \\
&= \{x_2\} &  \text{Assump. in lines~\ref{union-ln12}-\ref{union-ln14}} \\[3mm]
  &\hspace{-5mm} \nodes\ (x_1 \mapsto x_2 \join \G_1{\setminus{x_1}} \join \G_2) = & \\
&= \nodes\ (x_1 \mapsto x_2 \join \G_1{\setminus{x_1}})\ \dotcup\ \nodes\ \G_2  & \text{Distrib. of $\nodes$} \\
&= S_1\ \dotcup\ S_2  & \text{Assump. in lines~\ref{union-ln12}-\ref{union-ln14}}
\end{array} 
\]

The proof concludes in line~\ref{union-ln24} by folding the desired
spatial predicate.

\end{document}